\documentclass[11pt,a4paper,fleqn]{article}

\usepackage{amsmath,amssymb,amsthm,enumerate}%,backref
\usepackage[mathscr]{eucal}
\usepackage{hyperref}
\hypersetup{colorlinks, linkcolor=blue, citecolor=blue, urlcolor=blue}

\allowdisplaybreaks
\newcommand{\p}{\partial}
\newcommand{\const}{\mathop{\rm const}\nolimits}

\newcommand{\sgn}{\mathop{\rm sgn}\nolimits}
\newcommand{\rank}{\mathop{\rm rank}\nolimits}
\newcommand{\spanindex}{{\mbox{\tiny$\langle\,\rangle$}}}

\newcommand{\todo}[1][\null]{\ensuremath{\clubsuit}}

\newcommand{\noprint}[1]{}
\newcommand{\checked}[1][\null]{\ensuremath{\boldsymbol{\surd}}}

\newtheorem{theorem}{Theorem}%[section]
\newtheorem{lemma}[theorem]{Lemma}
\newtheorem{corollary}[theorem]{Corollary}
\newtheorem{proposition}[theorem]{Proposition}
\newtheorem*{proposition*}{Proposition}
{\theoremstyle{definition}
\newtheorem{definition}[theorem]{Definition}

\newtheorem{remark}[theorem]{Remark}
\newtheorem*{notation*}{Notation}
}

\begin{document}
\par\noindent {\LARGE\bf
Equivalence groupoids and group classification\\ of multidimensional nonlinear Schr\"odinger equations
\par}
%Equivalence groupoids and group classification of multidimensional nonlinear Schrödinger equations

\vspace{5mm}\par\noindent{\large
C\'elestin Kurujyibwami$^{\dag}$ and Roman O. Popovych$^\ddag$
}\par\vspace{2mm}\par

\vspace{2mm}\par\noindent{\it
$^\dag$\,College of Science and Technology, University of Rwanda, P.O.\,Box: 3900, Kigali, Rwanda}\\
$\phantom{^\dag}$\,E-mail: celeku@yahoo.fr

\vspace{2mm}\par\noindent{\it
$^\ddag$\,Fakult\"at f\"ur Mathematik, Universit\"at Wien, Oskar-Morgenstern-Platz 1, 1090 Wien, Austria\\
$\phantom{^\ddag}$\,Institute of Mathematics of NAS of Ukraine, 3 Tereshchenkivska Str., 01024 Kyiv, Ukraine}\\
$\phantom{^\ddag}$\,E-mail: rop@imath.kiev.ua

\vspace{6mm}\par\noindent\hspace*{10mm}\parbox{140mm}{\small
We study admissible and equivalence point transformations between generalized multidimensional nonlinear Schr\"odinger equations
and classify Lie symmetries of such equations.
We begin with a wide superclass of Schr\"odinger-type equations, which includes all the other classes considered in the paper.
Showing that this superclass is not normalized, we partition it into two disjoint normalized subclasses,
which are not related by point transformations.
Further constraining the arbitrary elements of the superclass,
we construct a hierarchy of normalized classes of Schr\"odinger-type equations.
This gives us an appropriate normalized superclass
for the non-normalized class of multidimensional nonlinear Schr\"odinger equations with potentials and modular nonlinearities
and allows us to partition the latter class into three families of normalized subclasses.
After a preliminary study of Lie symmetries of nonlinear Schr\"odinger equations with potentials and modular nonlinearities
for an arbitrary space dimension, we exhaustively solve the group classification problem for such equations
in space dimension two.
}

\noprint{
MSC: 35Q55, 35A30, 35B06
35-XX Partial differential equations
 35Axx General topics
  35A30 Geometric theory, characteristics, transformations [See also 58J70, 58J72]
 35Bxx Qualitative properties of solutions
  35B06 Symmetries, invariants, etc.
 35Qxx Partial differential equations of mathematical physics and other areas of application [See also 35J05, 35J10, 35K05, 35L05]
  35Q55 NLS equations (nonlinear Schroedinger equations) {For dynamical systems and ergodic theory, see 37K10}

Keywords:
nonlinear Schroedinger equations %nonlinear Schrodinger equations
equivalence groupoid
group classification of differential equations
Lie symmetry
equivalence group
point transformation
algebraic method of group classification
}

\section{Introduction}

Nonlinear Schr\"odinger-type equations appear in physics in many contexts.
They model nonlinear physical systems in hydrodynamics, optics, acoustics,
quantum condensates, heat pulses in solids, plasma physics, quantum mechanics
and biomolecular dynamics, to name just a few areas; see for instance
\cite{Belmonte-Beitia&Perez&Vekslerchick2008,clar1993a,
Gagnon&Winternitz1988,Gagnon&Winternitz1993,Garcia-Ripoll&Perez-Garcia&Vekslerchik2001}
and references therein.

There is a vast literature on the study of Schr\"odinger equations
within the framework of symmetry analysis of differential equations,
including their Lie symmetries and point transformations between them;
see, e.g., the reviews in~\cite{clar1992a}, \cite[Chapter~17]{CRC_vol2}, \cite[Section~4]{Popovych&Kunzinger&Eshragi2010} and below.
This study was systematically begun in the 1970's with linear Schr\"odinger equations
\cite{boy1975,Miller1977,Niederer1972,Niederer1973,Niederer1974}.
Nevertheless, linear Schr\"odinger equations are quite different from nonlinear ones
in symmetry and other related properties, and their study needs specific methods,
including modifications of methods of group classification;
see~\cite{Kurujyibwami&Basarab-Horwath&Popovych2018,Nikitin2017} and references therein.

To the best of our knowledge, the first paper
with essential usage of specific Lie symmetries of nonlinear Schr\"odinger equations
was the paper~\cite{Talanov1970}.
It was indicated there that the (1+2)-dimensional cubic Schr\"odinger equations
admits additional (conformal) symmetries
in comparison with the case of general power modular nonlinearity $|\psi|^\gamma\psi$.
This displays the fact that the power $\gamma=2$ is critical for space dimension two.
Lie symmetries of nonlinear Schr\"odinger equations
were systematically considered for the first time in~\cite{BoyerSharpWinternitz1976},
where an inverse group classification problem for (1+1)-dimensional nonlinear Schr\"odinger equations
of the form $i\psi_t+\psi_{xx}+F(t,x,\psi,\psi^*)=0$ was solved.
More specifically, subalgebras of the essential~\cite{popo2008a} Lie invariance algebra
of the (1+1)-dimensional free Schr\"odinger equations were classified
and equations of the above form that are invariant with respect to the listed subalgebras were constructed.
Later, other inverse group classification problems for nonlinear Schr\"odinger equations
were also considered, see, e.g., \cite{Fushchych&Cherniha1995,Gungor1999,ride1993a,Stoimenov&Henkel2005}.
Some Lie reductions of nonlinear Schr\"odinger equations with power modular nonlinearity,
were carried out in~\cite{anco2013a,Fushchych&Moskaliuk1981,Fushchych&Serov1987,Gagnon&Winternitz1990a}.
In~\cite{Gagnon&Winternitz1988,Gagnon89a,Gagnon89b,Gagnon89c,Gagnon89d},
classical Lie symmetry analysis was comprehensively carried out
for (1+3)-dimensional cubic-quintic Schr\"odinger equations,
whose nonlinear terms are of the form $(a_2|\psi|^4+a_1|\psi|^2+a_0)\psi$,
where $a_0$, $a_1$ and~$a_2$ are real constants with $(a_1,a_2)\ne(0,0)$.
This included the computation of the maximal Lie invariance algebras of such equations
depending on values of the parameters~$a$'s, the classification of subalgebras of these algebras,
Lie reductions using the obtained subalgebras, the construction of exact invariant solutions
and additional analysis of some specific submodels like the spherical submodel.
Lie reductions and invariant solutions of coupled systems of two (1+2)-dimensional nonlinear Schr\"odinger equations
were considered in~\cite{Gagnon1992}.
Partially invariant solutions of (1+1)-dimensional nonlinear Schr\"odinger equations
of the form $i\psi_t+\psi_{xx}+f\psi+g\psi_x=0$ with complex-valued parameter functions~$f$ and~$g$ of~$(|\psi|,|\psi|_x)$
were constructed in~\cite{mart1992a}.
Classical Lie reductions and nonclassical reductions for generalized nonlinear Schr\"odinger equations
of the form $i\psi_t+\psi_{xx}+b_1(|\psi|^2\psi)_x+b_2(|\psi|^2)_x\psi+a_2|\psi|^4\psi+a_1|\psi|^2\psi=0$
with $b_1,b_2\in\mathbb C$, $a_1,a_2\in\mathbb R$ and $(b_1,b_2,a_2)\ne(0,0,0)$
were carried out in~\cite{clar1992a,flor1990a,flor1992a}.
Analogous reductions were discussed in~\cite{clar1992a} for the case of space dimension three
and in~\cite{clar1993a} for similar cylindrical equations.
The integration of obtained reduced equations gave a number of exact solutions for the above Schr\"odinger equations.
Conditional symmetries of nonlinear Schr\"odinger equations with modular nonlinearity, $f(|\psi|)\psi$,
and of phase Schr\"odinger equations whose nonlinear terms are
of the form $i\beta\psi\mathop{\rm arg}\psi^2+f(|\psi|)\psi$ with $\beta\in\mathbb R_{\ne0}$
were obtained in~\cite{Fushchych&Chopyk1993} for specific differential constraints.
The systems of determining equations for Lie symmetries of
systems of nonlinear Schr\"odinger equations with inhomogeneous nonlinearities
were used in~\cite{Wittkopf2004,Wittkopf&Reid2000,Wittkopf&Reid2001}
for testing a software for solving overdetermined systems of PDE's.
Special attention was paid to the case of power modular nonlinearity
with potentials that are quadratic with respect to space variables.
In~\cite{Nikitin&Popovych2001}, a new method of group classification,
the so-called method of furcate splitting, was suggested,
which is especially efficient for group classification of classes of differential equations
whose arbitrary elements are functions of one or two arguments.
Using this advanced method, the complete group classification
of nonlinear Schr\"odinger equations of the form $i\psi_t+\Delta \psi+F(\psi,\psi^*)=0$
was carried out for an arbitrary space dimension.
An original method of generalized B\"acklund--Darboux transformations was applied 
in~\cite{Sakhnovich2001,Sakhnovich2006a,Sakhnovich2006b,Sakhnovich2008}
to finding exact solutions of the coupled matrix nonlinear Schr\"odinger equations,
including those with an external potential.
Exact solutions of semilinear radial Schr\"odinger equations with power modular nonlinearities 
were constructed in~\cite{anco2015c} via combining the method of group foliation with separation of variables.
% using the version of generalized B\"acklund--Darboux transformation that is based on operator identities.
Note that a collection of exact solutions of various nonlinear Schr\"odinger equations
was presented in~\cite{poly2012A}.
\looseness=-1

\looseness=-1
The algebraic method of group classification was first applied to Schr\"odinger equations
in~\cite{Gagnon&Winternitz1993}, where admissible transformations and Lie symmetries
for variable-coefficient generalizations of (1+1)-dimensional cubic Schr\"odinger equations
of the form $i\psi_t+G(t,x)\psi_{xx}+W(t,x)|\psi|^2\psi+V(t,x)\psi=0$
with complex-valued functions~$G$, $W$ and~$V$ of~$(t,x)$ with $\mathop{\rm Re}G\ne0$ and $\mathop{\rm Re}W\ne0$
were studied.
Exact solutions of such equations with $G=1$ were considered in~\cite{Belmonte-Beitia&Perez&Vekslerchick2008}.
Lie symmetries of more general variable-coefficient cubic-quintic nonlinear Schr\"odinger equations
were classified in~\cite{ozem2013a}.
Some results of~\cite{ride1993a} were extended in~\cite{Zhdanov&Roman2000}
via classifying (1+1)-dimensional nonlinear Schr\"odinger equations of the form
$i\psi_t+\psi_{xx}+F(t,x,\psi,\psi^*,\psi_x,\psi^*_x)=0$
with Lie-symmetry groups of dimensions one, two and three.
The equations invariant with respect to the Galilei group and its natural extensions
were selected among the listed ones,
which was used for symmetry classification of Galilei-invariant complex Doebner--Goldin models;
see also~\cite{Fushchych&Chopyk&Nattermann&Scherer1995}
for the classical Lie symmetry analysis of these models.
Group classification problems for various classes of nonlinear Schr\"odinger equations with modular nonlinearities and potentials
were solved on the series of papers
\cite{Ivanova2002,Ivanova&Popovych&Eshraghi2005,
Popovych&Eshraghi2004Mogran,Popovych&Ivanova&Eshraghi2004Cubic,
Popovych&Ivanova&Eshraghi2004Gamma}.
Such equations are important for applications and were intensively studied
within the framework of group analysis of differential equations,
theory of partial differential equations, etc.,
see e.g.
\cite{Baumann&Nonnenmacher1987,BialynickiSw2004,Doebner&Goldin&Nattermann1999,
Froehlich&Gustafson&Jonsson&Sigal2004,keng2007a,Perez-Garcia&Torres&Konotop2006}.
In fact, the consideration of the above group classification problems stimulated
revisiting the entire framework of the algebraic method of group classification,
which was based on the notion of normalized class of differential equations
\cite{popo2006a,Popovych2006ap,Popovych&Kunzinger&Eshragi2010}.

\looseness=-1
In the present paper, we study admissible point transformations within classes of multi-dimensional generalized nonlinear Schr\"odinger equations.
This allows us to classify Lie symmetries of multidimensional nonlinear Schr\"odinger equations with potentials and modular nonlinearities
using the algebraic method of group classification of differential equations.
We essentially generalize results of~\cite{Popovych&Kunzinger&Eshragi2010}
by considering an arbitrary space dimension and more general classes of nonlinear Schr\"odinger equations.
For this, we combine and further develop various techniques of modern group analysis of differential equations.
We apply
splitting into normalized subclasses,
gauging arbitrary elements by equivalence transformations,
looking for normalized superclasses
and singling out appropriate subalgebras of the corresponding equivalence algebras.
The basic notions and concepts used in this paper,
such as class of differential equations, equivalence groupoid, equivalence group, equivalence algebra, group classification
and normalization properties of a class of differential equations,
can be found in
\cite{bihl2012b,Kurujyibwami&Basarab-Horwath&Popovych2018,Opanasenko&Bihlo&Popovych2017,
popo2006a,Popovych&Kunzinger&Eshragi2010,Vaneeva&Bihlo&Popovych2020}
as well as in references therein.
See also \cite{blum2010A,Olver1986,Ovsiannikov1982}
for the general theory of group analysis of differential equations.

Looking for an appropriate normalized superclass for the class of nonlinear Schr\"odinger equations with potentials and modular nonlinearities,
we begin with the study of point transformations within the very wide class~$\mathscr N$ of generalized (1+$n$)-dimensional
($n\geqslant2$) nonlinear Schr\"odinger equations with ``variable mass'' of the form
\begin{gather}\label{general form of generalizedSchEqs}
i\psi_t+G(t,x,\psi,\psi^*,\nabla\psi,\nabla\psi^*)\psi_{aa}+F(t,x,\psi,\psi^*,\nabla\psi,\nabla\psi^*)=0,
\end{gather}
where $G$ and $F$ are smooth complex-valued functions of their arguments with~$G\ne 0$.
The class~$\mathscr N$ is the superclass for all classes of Schr\"odinger equations considered in the present paper.
This class is not normalized but can be partitioned into the two normalized subclasses~$\mathscr N_0$ and~$\bar{\mathscr N}_0$
singled out within the class $\mathscr N$ by the constraints~$G^*\ne-G$ and~$G^*=-G$, respectively.

\begin{notation*}
Throughout the paper,
$t$ and $x=(x_1,\dots,x_n)$ are the real independent variables,
$\psi$~is the unknown complex-valued function of $t$ and $x$, and $\rho:=|\psi|$.
For a complex value $\beta$,
the notation $\beta ^*$ denotes its conjugate, and we define
$\hat\beta=\beta$ if $T_t>0$ and $\hat\beta=\beta^*$ if $T_t<0$,
where $T$ is $t$-component of point transformations in the space with the coordinates $(t,x,\psi,\psi^*)$.
%\[\hat\beta=\beta \quad\mbox{if}\quad T_t>0  \quad\mbox{and}\quad \hat\beta=\beta^* \quad\mbox{if}\quad T_t<0.\]
Subscripts of functions denote differentiation with respect to the corresponding variables.
The indices $a$ and $b$ run from $1$ to $n$, the indices $\mu$ and $\nu$ run from $0$ to $n$, $x_0:=t$,
and we assume summation over repeated indices.
$\nabla\psi=(\psi_1,\dots,\psi_n)$ and $\nabla\psi^*=(\psi^*_1,\dots,\psi^*_n)$.
The total derivative operator $\mathrm D_\mu$ is defined as
$\mathrm D_\mu=\p_\mu+\psi_\mu\p_\psi+\psi^*_\mu\p_{\psi^*}+\psi_{\mu\nu}\p_{\psi_\nu}+\psi^*_{\mu\nu}\p_{\psi^*_\nu}+\cdots$.
Given a class of differential equations,
$\pi$ denotes the natural projection of the joint space of the variables and the arbitrary elements on the space of the variables only.
$E$ is the $n\times n$ identity matrix.
\end{notation*}

Choosing specific forms of $G$ and $F$, we further single out several subclasses from the normalized class~$\mathscr N_0$.
An important subclass is the normalized class~$\bar{\mathscr F}$ singled out from $\mathscr N_0$
by the condition that $G$ is a (nonzero) real constant.
The class~$\bar{\mathscr F}$ is mapped by a family of scalings of~$t$ and alternating the sign of this variable,
which are equivalence transformations of the class~$\bar{\mathscr F}$,
to its normalized subclass~$\mathscr F$ associated with the constraint $G=1$,
i.e., consisting of equations of the form
\begin{gather}\label{general form of generalizedSchEqs withconstantmass}
i\psi_t+\psi_{aa}+F(t,x,\psi,\psi^*,\nabla\psi,\nabla\psi^*)=0.
\end{gather}
We combine this consideration of the class~$\mathscr F$ with $n\geqslant2$
with earlier results on the class~$\mathscr F$ with $n=1$ from~\cite{Popovych&Kunzinger&Eshragi2010}
and thus further assume $n\geqslant1$.
We simultaneously constrain the arbitrary elements $G$ and $F$ by the equations $G=1$ and $F_{\psi_a}=F_{\psi^*_a}=0$
to single out the normalized class $\mathscr F_1$ whose equations are of the general form
\begin{gather}\label{general form of NLSchEqs withconstantmass}
i\psi_t+\psi_{aa}+F(t,x,\psi,\psi^*)=0.
\end{gather}
Additionally restricting $F$ to be $F=S(t,x,\rho)\psi$,
we obtain the normalized class~$\mathscr S$ of multidimensional
nonlinear Schr\"odinger equations of the form
\begin{equation}\label{eq:ClassS}
i\psi_t+\psi_{aa}+S(t,x,\rho)\psi=0,\quad S_\rho\ne 0,
\end{equation}
where $S$ is an arbitrary smooth complex-valued function of its arguments.
After a preliminary study of the Lie symmetries of equations from the class~$\mathscr S$,
we further constrain the arbitrary function $S$ by putting $S(t,x,\rho)=f(\rho)+V(t,x)$, $f_\rho\ne 0$.
This gives the non-normalized class~$\mathscr V$ of multidimensional nonlinear Schr\"odinger equations
with potentials and modular nonlinearities of the form
\begin{equation}\label{MNLinSchEqs_2}
i\psi_t+\psi_{aa}+f(\rho)\psi+V(t,x)\psi=0,\quad f_\rho\ne 0,
\end{equation}
where $V$ is an arbitrary smooth complex-valued potential depending on $t$ and $x$,
and $f$ is an arbitrary complex-valued nonlinearity depending only on $\rho$.
The class~$\mathscr S$ is a quite narrow and appropriate normalized superclass of the class~$\mathscr V$.
Using results on the class~$\mathscr S$, we partition the class~$\mathscr V$ into normalized subclasses,
study properties of Lie symmetries of equations from this class
and solve completely the group classification problem for the class~$\mathscr V$ with $n=2$.
We note that the case $n=1$ and the case $f=\rho^2$ with $n=2$ were studied in~\cite{Popovych&Kunzinger&Eshragi2010}.

The paper is organized as follows.

In Section~\ref{sect.pointtransformationsmultinonlin}
we compute the equivalence groupoid~$\mathcal G^\sim_\mathscr N$ of the class~$\mathscr N$.
This allows us to find the equivalence group~$G^\sim_\mathscr N$ of the class $\mathscr N$
and the equivalence groups~\smash{$G^\sim_{\mathscr N_0}$} and~\smash{$G^\sim_{\bar{\mathscr N}_0}$} of its subclasses~$\mathscr N_0$ and~$\bar{\mathscr N}_0$,
to check the normalization of these subclasses and to show that equations from these subclasses
are not related by point transformations.
As a result, we obtain a partition of the equivalence groupoid~$\mathcal G^\sim_\mathscr N$,
$\mathcal G^\sim_{\mathscr N}=\mathcal G^\sim_{\mathscr N_0}\bigsqcup\mathcal G^\sim_{\bar{\mathscr N}_0}$.

A hierarchy of normalized subclasses of~$\mathscr N_0$ is constructed in Section~\ref{sectionequivalence grooup}.
For each next class in the chain $\mathscr N_0\supset\tilde{\mathscr F}\supset\mathscr F\supset\mathscr F_1\supset\mathscr S$,
we use results on its directly preceding normalized superclass for describing the equivalence groupoid of this class,
which reduces to finding the equivalence group of this class and checking its normalization.
For the class~$\mathscr S$, we additionally derive its equivalence algebra~$\mathfrak g^\sim_\mathscr S$
as the set constituted by the infinitesimal generators of one-parameter subgroups
of its equivalence group~$G^\sim_{\mathscr S}$.

Section~\ref{sec:PreliminaryAnalysisOfLieSymsOfClassS} is devoted to analyzing
properties of the maximal Lie invariance algebras, $\mathfrak g_S$, of equations, $\mathcal L_S$, from the class~$\mathscr S$.
Applying the Lie infinitesimal criterion, we derive the system of determining equations
for Lie symmetries of equations~$\mathcal L_S$.
The analysis of this system yields the kernel invariance algebra~$\mathfrak g^\cap$, which is the intersection of all~$\mathfrak g_S$,
and a preliminary description of the algebras~$\mathfrak g_S$.
This description includes
the general form of elements of~$\mathfrak g_S$, the classifying equations for functions and constants parameterizing these elements,
the least upper bound of the dimensions of the algebras~$\mathfrak g_S$, which is equal to $n(n+3)/2+4$,
as well as estimates of dimensions of distinguished subalgebras of~$\mathfrak g_S$.

Section~\ref{sec:NSchEPMN} deals with the preliminary scheme of group classification for the class~$\mathscr V$
in the case of an arbitrary space dimension~$n$, which is based on the previous study of the class~$\mathscr S$.
We find the equivalence group of the class~$\mathscr V$ and show that this class, unlike the class~$\mathscr S$, is not normalized.
Then the class $\mathscr V$ is partitioned
into the disjoint normalized subclasses~$\mathscr V'$, $\mathscr P_0$ and~$\mathscr P_\lambda$, $\lambda\in\mathbb R_{\ne0}$,
which are respectively related, up to gauge equivalence transformations in~$\mathscr V$,
to general, logarithmic and power values of the arbitrary element~$f$, i.e.,
$f=f(\rho)$, $f=\delta\ln\rho$ and $f=\delta\rho^\lambda$,
where $\delta$ is an arbitrary complex constant.
The equivalence groups of these subclasses are constructed as well.
In order to study Lie symmetries of equations from these subclasses,
we treat the nonlinear terms in a special way:
by fixing $f(\rho)$ in $\mathscr V'$
and $\delta$ (after gauging to $|\delta|=1$ and $\delta_2\geqslant 0$ by scaling equivalence transformations) in~$\mathscr P_\lambda$, $\lambda\in\mathbb R$,
and then considering~$V$ as the only arbitrary element.
We refine estimates for the dimensions of maximal Lie invariance algebras and of their distinguished subalgebras,
which we obtain for equations from the superclass~$\mathscr S$, for each of the above normalized subclasses with specific nonlinearity~$f$.

The complete group classifications of these subclasses in the (1+2)-dimensional case are presented
in Sections~\ref{sec:NSchEPMNVf(1+2)D}--\ref{sec:NSchEPPowerMN(1+2)D}, respectively.
We introduce $\mathcal G^\sim_{\mathscr V}$-invariant parameters,
which depend on values of the arbitrary elements~$f$ and~$V$
and are related to the dimensions of distinguished subalgebras of the corresponding maximal Lie invariance algebras.
The ranges of these parameters can be estimated using the derived estimations of the subalgebra dimensions;
in fact, they take only small integer values.
We use them to distinguish subalgebras of the projections of the equivalence algebras of the above normalized subclasses
to the space with the coordinates $(t,x,\psi,\psi^*)$
that are appropriate as the maximal Lie invariance algebras of some equations from these classes.
Therefore, it is natural to split the solution of the group classification problems in these subclasses
into different cases depending on values of the introduced parameters.
In this way, we obtain a complete list of $\mathcal G^\sim_{\mathscr V}$-inequivalent Lie symmetry extensions in the class $\mathscr V$ with $n=2$
as the union of the group classification lists for the subclasses~$\mathscr V'$, $\mathscr P_0$ and $\mathscr P_\lambda$, $\lambda\in\mathbb R_{\ne0}$,
up to $G^\sim_{\mathscr V'}$-, $G^\sim_{\mathscr P_0}$- and $G^\sim_{\mathscr P_\lambda}$-equivalences, respectively.

As mentioned above, the results of the present paper generalize those given in~\cite{Popovych&Kunzinger&Eshragi2010}
for the (1+1)-dimensional nonlinear Schr\"odinger equations with potentials and modular nonlinearities
and for (1+2)-dimensional cubic Sch\-r\"odinger equations with potentials.
It is relevant to also complete the study for the case of spatial dimension three as the most important from the physical point of view.
Note that the complexity of group classification of nonlinear Schr\"odinger equations with potentials and modular nonlinearities
exponentially grows when the spatial dimension increases.

We intend to exploit the results of this paper to obtain Lie reductions and invariant solutions
for the class of nonlinear Schr\"odinger equations of the form~\eqref{MNLinSchEqs_2} in future publications.

\section[Admissible and equivalence transformations within the superclass]
{Admissible and equivalence transformations\\ within the superclass}\label{sect.pointtransformationsmultinonlin}

Using the direct method, we compute the equivalence groupoid~$\mathcal G^\sim_\mathscr N$
for the class $\mathscr N$ with $n\geqslant2$,
which consists of the equations of the form~\eqref{general form of generalizedSchEqs}.
In the course of this and other computations,
we should assume $(\psi,\psi^*)$ to be the tuple of the dependent variables
and extend the corresponding arbitrary-element tuple with the complex conjugates
of its components, e.g., $(G,F)$ with~$G^*$ and~$F^*$.
Confining a given differential function of $(\psi,\psi^*)$ to the solution set of an equation from the class $\mathscr N$,
we should supplement this equation by its complex conjugate,
thus considering the system of two equations instead of the single equation;
see~\cite[Section~1]{Kurujyibwami&Basarab-Horwath&Popovych2018}.
We find all point transformations of the general form
\begin{gather}\label{eq:GenPointTrans}
\varphi\colon\ \tilde t=T(t,x,\psi,\psi^*),\,\
\tilde x_a =X^a(t,x,\psi,\psi^*),\,\
\tilde\psi=\Psi(t,x,\psi,\psi^*),\,\
\tilde\psi^*=\Psi^*(t,x,\psi,\psi^*)
\end{gather}
with ${\rm d}T\wedge{\rm d}X^1\wedge\dots\wedge{\rm d}X^n\wedge{\rm d}\Psi\wedge{\rm d}\Psi^*\ne 0$,
that map a fixed equation~$\mathcal L_{GF}$ from the class~$\mathscr N$
to an equation~$\mathcal L_{\tilde G\tilde F}$,
\begin{equation*}
i\tilde \psi_{\tilde t}+
\tilde G(\tilde t,\tilde x,\tilde\psi,\tilde\psi^*,\tilde\nabla\tilde\psi,\tilde\nabla\tilde\psi^*)\tilde\psi_{\tilde x_a\tilde x_a}
+\tilde F(\tilde t,\tilde x,\tilde\psi,\tilde\psi^*,\tilde\nabla\tilde\psi,\tilde\nabla\tilde\psi^*)=0,
\end{equation*}
from the same class.

\begin{theorem}\label{thm:ClassNEquivGroupoid}
The equivalence groupoid~$\mathcal G^\sim_\mathscr N$ of the class $\mathscr N$ with $n\geqslant2$
is constituted by triples of the form $((G,F),\varphi,(\tilde G,\tilde F))$.
Here $\varphi$ is a point transformation of the form~\eqref{eq:GenPointTrans}
whose components satisfy the equations
\begin{subequations}\label{eq:ClassNEquivGroupoid}
\begin{gather}\label{eq:ClassNEquivGroupoidA}
T_{\psi}=T_{\psi^*}=0,\quad T_a=0, \quad
X^a_\psi=X^a_{\psi^*}=0,\quad
X^b_aX^c_a=H\delta_{bc},
\end{gather}
for some positive smooth real-valued function~$H$ of $(t,x)$, $T_t\ne0$, $\delta_{bc}$ is the Kronecker delta and,
\begin{gather*}
\mbox{if}\quad G^*\ne-G, \quad\mbox{then}\quad \Psi_\psi\Psi_{\psi^*}= 0.
\end{gather*}
The transformed arbitrary elements~$\tilde G$ and~$\tilde F$ are given by
\begin{gather}\label{eq:ClassNEquivGroupoidB}
\begin{split}
&\tilde G= \dfrac{H}{T_t}
\left\{
\begin{array}{lll}
G       &\mbox{if}&\Psi_\psi \ne0,\\[1ex]
(-G^*)  &\mbox{if}&\Psi_{\psi^*}\ne 0,
\end{array}
\right.
\end{split}
\\[1ex]
 \label{eq:ClassNEquivGroupoidC}
 \tilde F=\dfrac{\Psi_\psi}{T_t}F-\dfrac{\Psi_{\psi^*}}{T_t} F^*+
\left(\hat \Delta\Psi+\dfrac{n-2}2\frac{H_a}H\mathrm D_a\Psi\right)\dfrac{\tilde G}{H}
-\dfrac i{T_t}\left(\Psi_t-X^b_t\dfrac{X^b_a}H\mathrm D_a\Psi\right),
\end{gather}
\end{subequations}
where
$\hat \Delta:=\p_{aa}+2\psi_a\p_{a\psi}+2\psi^*_a\p_{a\psi^*}
+\psi_a\psi_a\p_{\psi\psi}+2\psi_a\psi^*_a \p_{\psi\psi^*}+\psi^*_a\psi^*_a\p_{\psi^*\psi^*}$.
\end{theorem}

\begin{proof}
Let $\varphi$ be a point transformation of the form~\eqref{eq:GenPointTrans} connecting two equations $\mathcal L_{GF}$
and $\mathcal L_{\tilde G\tilde F}$ from the class~$\mathscr N$.
Applying the total derivative operators $\mathrm D_\mu$'s to~$\tilde \psi(\tilde t,\tilde x)=\Psi(t,x,\psi,\psi^*)$
and $\tilde \psi^*(\tilde t,\tilde x)=\Psi^*(t,x,\psi,\psi^*)$,
we obtain the equations
\[
\tilde \psi_{\tilde x_\nu}\mathrm D_\mu X^\nu=\mathrm D_\mu\Psi,\quad
\tilde \psi^*_{\tilde x_\nu}\mathrm D_\mu X^\nu=\mathrm D_\mu\Psi^*,
\]
where we denote $X^0:=T$, and $x_0:=t$. After rearranging their terms, these equations take the form
\begin{gather*}
(\Psi_\psi-\tilde \psi_{\tilde x_\nu}X^\nu_\psi)\psi_\mu+(\Psi_{\psi^*}-\tilde \psi_{\tilde x_\nu}X^\nu_{\psi^*})\psi^*_\mu=-(\Psi_\mu-\tilde \psi_{\tilde x_\nu}X^\nu_\mu),\\
(\Psi^*_\psi-\tilde \psi^*_{\tilde x_\nu}X^\nu_\psi)\psi_\mu+(\Psi^*_{\psi^*}-\tilde \psi^*_{\tilde x_\nu}X^\nu_{\psi^*})\psi^*_\mu=-(\Psi^*_\mu-\tilde \psi^*_{\tilde x_\nu}X^\nu_\mu).
\end{gather*}
Denoting $W=W(t,x,\psi,\psi^*,\tilde\nabla\tilde\psi):=\Psi-\tilde\psi_{\tilde x_\nu}X^\nu$,
we solve the above equations with respect to~$\psi_\mu$ and~$\psi^*_\mu$:
\begin{gather}\label{firstderivatives}
\psi_\mu=-\frac{W_\mu W^*_{\psi^*}-W^*_\mu W_{\psi^*}}Y,\quad \psi^*_\mu=-\frac{W_\psi W^*_\mu-W^*_{\psi}W_\mu}Y,
\end{gather}
where
\begin{equation}\label{Yequation}
\begin{split}
Y:={}& W_\psi W^*_{\psi^*}-W^*_\psi W_{\psi^*}\\
={}&\left|
\begin{array}{lll}
 \Psi_\psi & \Psi_{\psi^*}\\
  \Psi^*_\psi &\Psi^*_{\psi^*}
\end{array}
\right|
-\left|
\begin{array}{lll}
 X^\nu_\psi &  X^\nu_{\psi^*}\\
  \Psi^*_\psi &\Psi^*_{\psi^*}
\end{array}
\right|\tilde\psi_{\tilde x_\nu}
-\left|
\begin{array}{lll}
 \Psi_\psi & \Psi_{\psi^*}\\
  X^\mu_\psi &X^\mu_{\psi^*}
\end{array}
\right|\tilde\psi^*_{\tilde x_\mu}
+\left|
\begin{array}{lll}
 X^\nu_\psi & X^\nu_{\psi^*}\\
X^\mu_\psi &X^\mu_{\psi^*}
\end{array}
\right|\tilde\psi_{\tilde x_\nu}\tilde\psi^*_{\tilde x_\mu}.
\end{split}
\\[1ex]
\end{equation}
Note that $Y\ne0$. Indeed, if $Y=0$, then each of the determinants in~\eqref{Yequation} vanishes,
which implies that the tuples $(T_\psi, X^1_\psi,\dots, X^n_\psi,\Psi_\psi,\Psi^*_\psi)$ and
$(T_{\psi^*},X^1_{\psi^*},\dots, X^n_{\psi^*},\Psi_{\psi^*},\Psi^*_{\psi^*})$ are linearly dependent,
leading to a contradiction of the invertibility of the point transformation.

We compute $\psi_{aa}$ using the representations~\eqref{firstderivatives} for~$\psi_a$:
\begin{gather*}
\begin{split}
\psi_{aa}&=\mathrm D_a\psi_a
=\psi_a\tilde \psi_{\tilde x_\mu}\mathrm D_a\tilde\psi_{\tilde x_\mu}+\psi_a\tilde \psi^*_{\tilde x_\mu}\mathrm D_a\tilde\psi^*_{\tilde x_\mu}+R
\\ 
&=\frac 1Y(\mathrm D_aX^\mu)(\mathrm D_aX^\nu)(\tilde\psi_{\tilde x_\mu\tilde x_\nu}W^*_{\psi^*}-\tilde\psi^*_{\tilde x_\mu\tilde x_\nu}W_{\psi^*})+R,
\end{split}
\end{gather*}
where $R$ is an expression not involving second derivatives of $\psi$ and $\psi^*$,
the precise form of which is not essential.
We substitute the derived expression for~$\psi_{aa}$ into the equation~$\mathcal L_{GF}$
and then substitute for $\tilde\psi_{\tilde x_1\tilde x_1}$ and $\tilde\psi^*_{\tilde x_1\tilde x_1}$
in view of the equation~$\mathcal L_{\tilde G\tilde F}$ and its conjugate.
The obtained equation~$\hat{\mathcal L}$ can be split with respect to $\tilde \psi_{\tilde t\tilde t}$ and $\tilde \psi^*_{\tilde t\tilde t}$.
Collecting the coefficients of the first degrees of these derivatives leads to
the equations $W^*_{\psi^*}(\mathrm D_aT)(\mathrm D_aT)=0$ and $W_{\psi^*}(\mathrm D_aT)(\mathrm D_aT)=0$.
Since $(W^*_{\psi^*},W_{\psi^*})\ne(0,0)$ in view of $Y\ne0$,
these equations obviously imply $(\mathrm D_aT)(\mathrm D_aT)=0$.
The last equation can be split with respect to~$\psi_a$ and~$\psi^*_a$,
which leads to $T_\psi=T_{\psi^*}=0$, $T_aT_a=0$. Hence $T_a=0$ and $T_t\ne0$.
Since $n\geqslant2$, we can also split the equation~\smash{$\hat{\mathcal L}$}
with respect to~$\tilde \psi_{\tilde x_b\tilde x_c}$ and~$\tilde \psi^*_{\tilde x_b\tilde x_c}$ with $(b,c)\ne(1,1)$.
We collect the coefficients of the first degrees of these derivatives.
Taking into account the inequality $(W^*_{\psi^*},W_{\psi^*})\ne(0,0)$,
we derive the equations $(\mathrm D_aX^b)(\mathrm D_aX^c)=0$, $b\ne c$,
and $(\mathrm D_aX^1)(\mathrm D_aX^1)=\dots=(\mathrm D_aX^n)(\mathrm D_aX^n)$,
which can further be split with respect to~$\psi_a$ and~$\psi^*_a$.
The resulting system includes the equations
\begin{gather*}
X^b_\psi X^c_\psi=0,\quad X^b_{\psi^*}X^c_{\psi^*}=0,\quad b\ne c,\quad
(X^1_\psi)^2=\dots=(X^n_\psi)^2,\quad
(X^1_{\psi^*})^2=\dots=(X^n_{\psi^*})^2, \\
X^b_aX^c_a=0,\quad b\ne c,\quad X^1_aX^1_a=\dots=X^n_aX^n_a
\end{gather*}
implying $X^a_\psi=X^a_{\psi^*}=0$ and $X^b_aX^c_a=H\delta_{bc}$ for some positive smooth real-valued function~$H$ of $(t,x)$.
Using these results for $T$ and $X^a$ in~\eqref{firstderivatives}
we obtain $Y=\Psi_\psi\Psi^*_{\psi^*}-\Psi^*_\psi\Psi_{\psi^*}$,
\begin{gather*}
\psi_\mu=
\frac{X^\nu_\mu\Psi^*_{\psi^*}}Y \tilde\psi_{\tilde x_\nu}
-\frac{X^\nu_\mu\Psi_{\psi^*}}Y \tilde\psi^*_{\tilde x_\nu}
-\frac{\Psi_\mu\Psi^*_{\psi^*}-\Psi^*_\mu\Psi_{\psi^*}}Y,\\
\psi_{aa} =\frac HY\left(\Psi^*_{\psi^*}\tilde \psi_{\tilde x_b\tilde x_b}-\Psi_{\psi^*}\tilde \psi^*_{\tilde x_b\tilde x_b}\right)
+\tilde\psi_{\tilde x_b} \mathrm D_a\dfrac{X^b_a\Psi^*_{\psi^*}}Y
-\tilde\psi^*_{\tilde x_b}\mathrm D_a\dfrac{X^b_a\Psi_{\psi^*}}Y-\mathrm D_a\dfrac{\Psi_a\Psi^*_{\psi^*}-\Psi^*_a\Psi_{\psi^*}}Y.\!
\end{gather*}
Due to the above  expressions for~$\psi_t$ and~$\psi_{aa}$ we can easily split the equation~\smash{$\hat{\mathcal L}$}
with respect to $\tilde\psi_{\tilde t}$ and $\tilde\psi^*_{\tilde t}$.
Collecting the coefficients of the first degrees of these derivatives gives
\[
(HG-T_t\tilde G)\Psi^*_{\psi^*}=0,\quad
(HG+T_t\tilde G^*)\Psi_{\psi^*}=0,
\]
and thus the condition~\eqref{eq:ClassNEquivGroupoidB} holds
since $(\Psi^*_{\psi^*},\Psi_{\psi^*})\ne(0,0)$.
If $\Psi^*_{\psi^*}\Psi_{\psi^*}\ne 0$, then $HG=T_t\tilde G=-T_t\tilde G^*$ and $HG=T_t\tilde G=-HG^*$,
i.e., $\tilde G=-\tilde G^*$, $G=-G^*$.
Finally, collecting the remaining terms gives the relation between~$F$ and $\tilde F$.
\end{proof}

\begin{remark}
The determining equations~\eqref{eq:ClassNEquivGroupoidA} for elements
of the equivalence groupoid~$\mathcal G^\sim_\mathscr N$
represent the principal structure properties of equations
from the class $\mathscr N$ with $n\geqslant2$
or, more precisely, of systems that each consists of such an equation and its complex conjugate.
Thus, the equations $T_{\psi}=T_{\psi^*}=T_a=0$,
$X^a_\psi=X^a_{\psi^*}=0$ and $X^b_aX^c_a=H\delta_{bc}$
are associated with the evolution form of these systems, their quasi-linearity
and the involvement of the derivatives of~$\psi$ and~$\psi^*$ of the highest (second) order via Laplacian,
respectively; cf.\ \cite{blum1990c, Kingston&Sophocleous1998} for the case of single equations.
\end{remark}

\begin{corollary}\label{CorolEquivalence group}
The class~$\mathscr N$ is not normalized.
Its equivalence group~$G^\sim_{\mathscr N}$ consists of the point transformations
in the space of $(t,x,\psi,\psi^*,\nabla\psi,\nabla\psi^*, G,G^*,F,F^*)$,%
\footnote{%
According to the definition of equivalence groups, the group~$G^\sim_{\mathscr N}$
acts in the space with the coordinates $(t,x,\psi_{(2)},\psi^*_{(2)}, G,G^*,F,F^*)$,
where the subscript ``$(2)$'' of a dependent variable denotes the collection of jet coordinates corresponding to
the derivatives of this dependent variable up to order two, including the dependent variable itself as its zeroth-order derivative.
At~the same time, the arbitrary elements~$F$ and~$G$ depend only on the jet variables $(t,x,\psi,\psi^*,\nabla\psi,\nabla\psi^*)$.
In view of Theorem~\ref{thm:ClassNEquivGroupoid}, admissible transformations in the class~$\mathscr N$
preserve the subspace of the second-order jet space $\mathrm J^2(t,x|\psi,\psi^*)$ with these coordinates.
Hence, we can restrict the space underlying the group~$G^\sim_{\mathscr N}$
and thus assume that its elements act in the space with the coordinates $(t,x,\psi,\psi^*,\nabla\psi,\nabla\psi^*)$.
It~suffices to present only the transformation components for $(t,x,\psi)$ and for the arbitrary elements.
The transformation components for derivatives of~$\psi$ are constructed from the $(t,x,\psi)$-components
by the standard prolongation using the chain rule.
The transformation components for derivatives of~$\psi^*$ and for the complex conjugates of complex-valued arbitrary elements
are obtained by conjugating their counterparts for derivatives of~$\psi$ and for the corresponding arbitrary elements,
cf.~\cite[Section~1]{Kurujyibwami&Basarab-Horwath&Popovych2018}.
We use the same approach for all the subclasses of~$\mathscr N$ and their reparameterizations
considered in the present paper.
}
where the components for $t$, $x$ and $\psi$ are of the form~\eqref{eq:GenPointTrans}
with~$T$, $X^a$ and~$\Psi$ satisfying the equations~\eqref{eq:ClassNEquivGroupoidA} and $\Psi_\psi\Psi_{\psi^*}=0$,
and the components for~$G$ and~$F$ are of the form~\eqref{eq:ClassNEquivGroupoidB} and~\eqref{eq:ClassNEquivGroupoidC},
respectively.
\end{corollary}

Since the class~$\mathscr N$ is not normalized,
we partition it into two disjoint subclasses $\mathscr N_0$ and~$\bar{\mathscr N}_0$,
which are singled out by the constraints~$G^*\ne-G$ and~$G^*=-G$, respectively.
It is clear that there are no point transformations mapping equations from the class~$\mathscr N_0$
to equations from the class~$\bar{\mathscr N}_0$.
We can obtain their corresponding equivalence groupoids~$\mathcal G^\sim_{\mathscr N_0}$
and $\mathcal G^\sim_{\bar{\mathscr N}_0}$, and as well as their equivalence groups~$G^\sim_{\mathscr N_0}$
and $G^\sim_{\bar{\mathscr N}_0}$ from Theorem~\ref{thm:ClassNEquivGroupoid}.

\begin{corollary}\label{cor:ClassN0}
The class~$\mathscr N$ is partitioned by the constraints $G^*\ne-G$ and $G^*=-G$
into the two disjoint subclasses~$\mathscr N_0$ and~$\bar{\mathscr N}_0$, respectively,
which are normalized and are not related by point transformations.
The equivalence group~$G^\sim_{\mathscr N_0}$ coincides
with the equivalence group~$G^\sim_{\mathscr N}$ of the entire class~$\mathscr N$.
The equivalence group~$G^\sim_{\bar{\mathscr N}_0}$ consists of the point transformations
in the space of $(t,x,\psi,\psi^*,\nabla\psi,\nabla\psi^*, G,G^*,F,F^*)$,
where the components for $t$, $x$ and $\psi$ are of the form~\eqref{eq:GenPointTrans}
with~$T$, $X^a$ and~$\Psi$ satisfying the equations~\eqref{eq:ClassNEquivGroupoidA},
and the components for~$G$ and~$F$ are of the form~\eqref{eq:ClassNEquivGroupoidB} and~\eqref{eq:ClassNEquivGroupoidC}.
\end{corollary}

In other words, the partition $\mathscr N=\mathscr N_0\bigsqcup\bar{\mathscr N}_0$ of the class~$\mathscr N$
leads to a partition of its equivalence groupoid~$\mathcal G^\sim_{\mathscr N}$,
$\mathcal G^\sim_{\mathscr N}=\mathcal G^\sim_{\mathscr N_0}\bigsqcup\mathcal G^\sim_{\bar{\mathscr N}_0}$,
where the equivalence groupoids~$\mathcal G^\sim_{\mathscr N_0}$ and~$\mathcal G^\sim_{\bar{\mathscr N}_0}$
are generated by the corresponding equivalence groups~$G^\sim_{\mathscr N_0}$ and~$G^\sim_{\bar{\mathscr N}_0}$.

\begin{corollary}\label{cor:SubclassesOfN}
If a subclass of the class~$\mathscr N$ intersects both the subclasses~$\mathscr N_0$ and~$\bar{\mathscr N}_0$,
then its equivalence group is a subgroup of~$G^\sim_{\mathscr N}$.
\end{corollary}

\section{Transformational properties of intermediate subclasses}\label{sectionequivalence grooup}

Selecting the equations with $G=1$ from the class~$\mathscr N$,
we obtain the subclass $\mathscr F$ of~$\mathscr N$
whose equations are of the form~\eqref{general form of generalizedSchEqs withconstantmass}.
It is easy to single out its equivalence groupoid~$\mathcal G^\sim_{\mathscr F}$ as a subgroupoid
from the equivalence groupoid~$\mathcal G^\sim_{\mathscr N}$ computed in Theorem~\ref{thm:ClassNEquivGroupoid} for $n\geqslant2$.
In the case $n=1$, the equivalence groupoid~$\mathcal G^\sim_{\mathscr F}$ was computed
in~\cite[Theorem~1]{Popovych&Kunzinger&Eshragi2010}.
Merging the above cases of~$n$, we formulate the following assertion.

\begin{proposition}\label{pro:ClassFEquivGroupoid}
The equivalence groupoid~$\mathcal G^\sim_{\mathscr F}$ of the class~$\mathscr F$
consists of triples of the form $(F,\varphi,\tilde F)$,
where $\varphi$ is a point transformation in the space of variables,
given by
\begin{subequations}\label{eq:ClassFEquivGroupoid}
\begin{gather}\label{eq:ClassFEquivGroupoidA}
\tilde t=T, \quad
\tilde x_a=|T_t|^{1/2}O^{ab}x_b+\mathcal X^a,\quad
\tilde \psi=\Psi(t,x,\hat \psi),
\end{gather}
and the target value $\tilde F$ of the arbitrary element is expressed via its source value $F$ as
\begin{gather}\label{eq:ClassFEquivGroupoidB}
\begin{split}
\tilde F={}&\frac{\Psi_{\hat\psi}}{|T_t|}\hat F-i\frac{\Psi_t}{T_t}+i\biggl(\frac{T_{tt}}{2|T_t|^2}x_a
+\frac{\varepsilon'}{|T_t|^{3/2}}\mathcal X^b_tO^{ba}\biggr)\big(\Psi_a+\Psi_{\hat\psi}\hat \psi_a\big)
\\[1ex]
&{}-\dfrac1{|T_t|}\big(\Psi_{aa}+2\Psi_{a\hat\psi}\hat\psi_a+\Psi_{\hat\psi\hat\psi}\hat\psi_a\hat\psi_a\big),
\end{split}
\end{gather}
\end{subequations}
where $T$ and $\mathcal X^a$ are arbitrary smooth real-valued functions of $t$ with $T_t\ne 0$,
$\Psi$ is an arbitrary smooth complex-valued function of $t$, $x$ and $\hat\psi$ with $\Psi_{\hat\psi}\not=0$,
$O=(O^{ab})$~is an arbitrary constant orthogonal $n\times n$ matrix, and $\varepsilon'=\sgn T_t$.
\end{proposition}

\begin{proof}
Since the case $n=1$ was studied in~\cite[Theorem~1]{Popovych&Kunzinger&Eshragi2010},
it suffices to consider the case $n\geqslant2$.
Setting $G=1$ and $\tilde G=1$ in the description of admissible transformations of~$\mathscr N$
that is given in Theorem~\ref{thm:ClassNEquivGroupoid},
we obtain that transformational parts~$\varphi$ of admissible transformations of the class~$\mathscr F$
are of the form~\eqref{eq:GenPointTrans}, 
where the components satisfy the equations~\eqref{eq:ClassNEquivGroupoidA} with $H=|T_t|$ and the equation $\Psi_\psi\Psi_{\psi^*}= 0$.
This expression for~$H$ and the condition \smash{$\Psi_{\hat\psi}\ne0$} follow from the equation~\eqref{eq:ClassNEquivGroupoidB}.
As a result, we obtain the representation~\eqref{eq:ClassFEquivGroupoidA} for~$\varphi$.
The expression for the target value $\tilde F$ of the arbitrary element is derived from~\eqref{eq:ClassNEquivGroupoidC}
in view of~\eqref{eq:ClassFEquivGroupoidA}.
\end{proof}

The analysis of structure of~$\mathcal G^\sim_{\mathscr F}$ shows
that it suffices to choose the space with the coordinates $(t,x,\psi,\psi^*,\nabla\psi,\nabla\psi^*,F,F^*)$
as the space underlying the equivalence group $G^\sim_{\mathscr F}$ of the class~$\mathscr F$.

\begin{corollary}\label{cor:ClassFEquivGroup}
The class~$\mathscr F$ is normalized.
The equivalence group $G^\sim_{\mathscr F}$ of~$\mathscr F$ is constituted by the point transformations
in the space with the coordinates $(t,x,\psi,\psi^*,\nabla\psi,\nabla\psi^*,F,F^*)$
whose components for $(t,x,\psi,F)$ are of the form~\eqref{eq:ClassFEquivGroupoidA}--\eqref{eq:ClassFEquivGroupoidB},
and the other components are obtained via the complex conjugation or the prolongation by the chain rule.
\end{corollary}

The continuous equivalence transformations of the class~$\mathscr F$ are singled out
from the group $G^\sim_{\mathscr F}$ by the constraints $T_t>0$ and $\det O=1$.
Therefore, the class~$\mathscr F$ possesses only two discrete equivalence transformations that
are independent up to combining with each other and with continuous equivalence transformations.
These are
the space reflection $\tilde t=t,$ $\tilde x_a=-x_a,$ $\tilde x_b=x_b,$ $b\ne a$, $\tilde \psi=\psi,$ $\tilde F=F$
for a fixed~$a$
and the Wigner time reflection $\tilde t=-t,$ $\tilde x=x,$ $\tilde \psi=\psi^*$,~$\tilde F=F^*$.
The above property of the class~$\mathscr F$ is inherited
by its subclasses $\mathscr F_1$, $\mathscr S$, $\mathscr V$, $\mathscr V'$, $\mathscr P_\lambda$, $\lambda\in\mathbb R$,
which are considered below.

\begin{remark}\label{cor:ClassTildeF}
From the point of view of physical applications,
an important class of generalized Schr\"odinger equations is the wider subclass~$\tilde{\mathscr F}$ of~$\mathscr N_0$
that is singled out by the constraint that $G$ is a positive real constant, $G\in\mathbb R_{>0}$.
The class~$\tilde{\mathscr F}$ is normalized as well.
Similarly to the class~$\mathscr F$, the equivalence group of~$\tilde{\mathscr F}$ consists of the point transformations
in the space with the coordinates $(t,x,\psi,\psi^*,\nabla\psi,\nabla\psi^*,G,F,F^*)$
whose $(t,\psi,F)$-components are given by~\eqref{eq:ClassFEquivGroupoidA}--\eqref{eq:ClassFEquivGroupoidB},
the $x$-components is modified, in comparison to~\eqref{eq:ClassFEquivGroupoidA}, as $\tilde x_a=c|T_t|^{1/2}O^{ab}x_b+\mathcal X^a$,
the $G$-component is $\tilde G=c^2 G$
and the other components are again obtained via the complex conjugation or the prolongation by the chain rule.
Here $c$ runs through $\mathbb R_{>0}$.
The extension of~$\mathscr F$ to~$\tilde{\mathscr F}$ is not essential since
the class~$\tilde{\mathscr F}$ can be mapped to its subclass~$\mathscr F$ via gauging of~$G$ by scalings of~$x$.
Relaxing the constraint $G\in\mathbb R_{>0}$ to $G\in\mathbb R_{\ne0}$
merely introduce the simultaneous alternating of the signs of $t$, $F$ and~$G$ into the corresponding equivalence group.
In the same way, the constraint $G=1$ can be relaxed to  $G\in\mathbb R_{>0}$ or $G\in\mathbb R_{\ne0}$
for the subclasses of~$\mathscr F$ considered below,
and these relaxations are also not essential in the course of the study of admissible transformations and Lie symmetries.
\end{remark}

Constraining the arbitrary element $F$ with the equations $F_{\psi_a}=F_{\psi^*_a}=0$,
we single out the subclass~$\mathscr F_1$, whose equations are of the form~\eqref{general form of NLSchEqs withconstantmass}.
Then the expression~\eqref{eq:ClassFEquivGroupoidB} for~$\tilde F$ implies
the equations $\Psi_{a\hat\psi}=\frac12X^a_bX^b_t\Psi_{\hat\psi}$ and $\Psi_{\hat\psi\hat\psi}=0$.
Integrating them, we obtain the expression for~$\Psi$ and then substitute it into~\eqref{eq:ClassFEquivGroupoid},
which gives the description of the admissible transformations of the class~$\mathscr F_1$.

\begin{proposition}\label{pro:ClassF1EquivGroupoid}
The class~$\mathscr F_1$ is normalized.
Its equivalence group~$G^\sim_{\mathscr F_1}$ is a subgroup of~$G^\sim_{\mathscr F}$
and consists of point transformations of the form~\eqref{eq:ClassFEquivGroupoidA}--~\eqref{eq:ClassFEquivGroupoidB},
where
\begin{gather*}
\Psi=\exp\left(\frac i8\frac{T_{tt}}{|T_t|}\,x_ax_a+
\frac i2\frac{\varepsilon'\mathcal X^b_t}{|T_t|^{1/2}}\,O^{ba}x_a+i\Sigma+Z
 \right)\hat\psi+\Psi^0,
\end{gather*}
$T$, $\mathcal X^a$, $\Sigma$ and $Z$ are arbitrary smooth real-valued functions of $t$ with $T_t\ne 0$,
$\Psi^0$ is an arbitrary smooth complex-valued function of $t$ and $x$,
and $O=(O^{ab})$~is an arbitrary constant orthogonal $n\times n$  matrix.
\end{proposition}

We now turn our attention to the subclass of $\mathscr F_1$
singled out by the constraints
\[
\psi\left(\frac F\psi\right)_{\!\psi}\!-\psi^*\left(\frac F\psi\right)_{\!\psi^*}\!=0, \quad
\psi\left(\frac F\psi\right)_{\!\psi}\!+\psi^*\left(\frac F\psi\right)_{\!\psi^*}\!\ne0,
\]
which are jointly equivalent to the representation $F=S(t,x,\rho)\psi$ with $S_\rho\ne0$,
where $\rho:=|\psi|$.
The reparameterization of this subclass by assuming $S:=F/\psi$ to be the arbitrary element instead of~$F$
leads to the class~$\mathscr S$ of equations of the form~\eqref{eq:ClassS}.
In terms of~$S$, the above auxiliary system for the arbitrary element takes the form
\begin{equation}\label{eq:ClassSAuxiliarySystem}
\psi S_\psi-\psi^*S_{\psi^*}=0, \quad \psi S_\psi+\psi^*S_{\psi^*}\not=0.
\end{equation}
Proposition~\ref{pro:ClassF1EquivGroupoid} implies
that the equivalence group~$G^\sim_{\mathscr S}$ of the class~$\mathscr S$
is imbedded as a subgroup into the group~$G^\sim_{\mathscr F_1}$,
which is associated with the constraint $\Psi^0=0$.
A similar claim holds for the equivalence groupoid of the class~$\mathscr S$.

\begin{theorem}\label{thm:ClassSGequiv}
The class $\mathscr S$ is normalized.
Its equivalence group~$G^\sim_{\mathscr S}$ consists of the point transformations
in the space with the coordinates $(t,x,\psi,\psi^*,S,S^*)$ that are of the form
\begin{subequations}\label{eq:ClassSEquivGroup}
\begin{gather}\label{eq:ClassSEquivGroupA}
\tilde t=T, \quad
\tilde x_a=|T_t|^{1/2}O^{ab}x_b+\mathcal X^a,
\\ \label{eq:ClassSEquivGroupB}
\tilde \psi=\exp\left(\frac i8\frac{T_{tt}}{|T_t|}\,x_ax_a+
\frac i2\frac{\varepsilon'\mathcal X^b_t}{|T_t|^{1/2}}\,O^{ba}x_a+i\Sigma+Z\right)\hat\psi,
\\ \label{eq:ClassSEquivGroupC}
\tilde S=\frac{\hat S}{|T_t|}+
\frac{2T_{ttt}T_t-3T_{tt}{}^2}{16\varepsilon'T_t{}^3}x_ax_a
+\frac{\varepsilon'}2\left(\frac{\mathcal X^b_t}{T_t}\right)_{\!t} \frac{O^{ba}x_a}{|T_t|^{1/2}}
+\frac{\Sigma_t-iZ_t}{T_t}-\frac{\mathcal X^a_t\mathcal X^a_t+inT_{tt}}{4T_t{}^2}.
\end{gather}
\end{subequations}
Here $T$, $\mathcal X^a$, $Z$ and $\Sigma$ are arbitrary smooth real-valued functions of~$t$
with $T_t\ne 0$, $\varepsilon'=\sgn T_t$ and $O=(O^{ab})$~is an arbitrary constant orthogonal $n\times n$ matrix.
\end{theorem}

Theorem~\ref{thm:ClassSGequiv} is important for the group classification of any subclass of the class~$\mathscr S$.
Admissible and equivalence transformations for any such subclass are deduced from transformations of the form~\eqref{eq:ClassSEquivGroup}.

\begin{corollary}\label{cor:DiffInvForS}
The ratio $\rho S_{\rho\rho}/S_\rho$ is a differential invariant of the subgroup of~$G^\sim_{\mathscr S}$
associated with the condition $T_t>0$.
Moreover, if this ratio is real-valued, then it is a differential invariant of the entire group~$G^\sim_{\mathscr S}$.
\end{corollary}

To find the equivalence algebra $\mathfrak g^\sim_\mathscr S$ of the class~$\mathscr S$,
we use the knowledge of the equivalence group~$G^\sim_\mathscr S$ of the class~$\mathscr S$
as described in Theorem~\ref{thm:ClassSGequiv}.
We evaluate the set of all infinitesimal generators of one-parameter subgroups of the group~$G^\sim_\mathscr S$
by representing the parameter function $\Sigma$ as $\Sigma=\frac14\mathcal X^a\mathcal X^a_t+\bar\Sigma$,
where $\bar\Sigma$ is a smooth function of~$t$, to ensure the existence of such subgroups.
Then, we successively assume one of the parameters~$T$, $O$, $\mathcal X^a$, $\bar\Sigma$ and~$Z$
to depend on a continuous parameter~$\delta$ and setting the other parameters to their trivial values,
which are $t$ for $T$, $E$ for $O$ and zeroes for $\mathcal X^a$, $\bar\Sigma$ and~$Z$.
This procedure leads to the components of the associated infinitesimal generator
$\tau\p_t+\xi^a\p_a+\eta\p_\psi+\eta^*\p_{\psi^*}+\theta\p_S+\theta^*\p_{S^*}$, computed~as
\[
\tau=\frac{\mathrm d\tilde t}{\mathrm d\delta}\Big|_{\delta=0},\quad
\xi^a=\frac{\mathrm d\tilde x}{\mathrm d\delta}\Big|_{\delta=0},\quad
\eta=\frac{\mathrm d\tilde\psi}{\mathrm d\delta}\Big|_{\delta=0},\quad
\theta=\frac{\mathrm d\tilde S}{\mathrm d\delta}\Big|_{\delta=0}.
\]

\begin{corollary}\label{cor:ClassSgequiv}
The equivalence algebra of the class~$\mathscr S$ is the algebra
\[
\mathfrak g^\sim_\mathscr S=\big\langle\hat D(\tau),\,\hat J_{ab},\,a<b,\,\hat P(\chi),\,\hat M(\sigma),\,\hat I(\zeta)\big\rangle,
\]
where $\chi:=(\chi^1,\dots,\chi^n)$;
$\tau$, $\chi^a$, $\sigma$ and~$\zeta$ run through the set of smooth real-valued functions of~$t$,
\begin{gather*}
\hat D(\tau)=\tau\p_t+\frac12\tau_tx_a\p_a+\frac i{8}\tau_{tt}x_ax_a \left(\psi\p_\psi-\psi^*\p_{\psi^*}\right)
\\ \phantom{\hat D(\tau)=}
{}-\left(\tau_tS-\frac18\tau_{ttt}x_ax_a+i\frac{\tau_{tt}}{4}\right)\p_S- \left(\tau_tS^{*}-\frac18\tau_{ttt}x_ax_a-i\frac{\tau_{tt}}{4}\right)\p_{S^*},
\\
\hat J_{ab}=x_a\p_b-x_b\p_a,\quad a\ne b,\\
\hat P(\chi)=\chi^a\p_a+\frac i2\chi^a_tx_a\left(\psi\p_\psi-\psi^*\p_{\psi^*}\right)+\frac12\chi^a_{tt}x_a\left(\p_S+ \p_{S^*}\right),
\\
\hat M(\sigma)=i\sigma(\psi\p\psi-\psi\p\psi^*)+\sigma_t(\p_S+\p_S^*),\\
\hat I(\zeta)=\zeta(\psi\p\psi+\psi^*\p\psi^*)-i\zeta_t(\p_S+\p_S^*).
\end{gather*}
\end{corollary}

\section[Preliminary analysis of Lie symmetries of equations from an intermediate class]
{Preliminary analysis of Lie symmetries of equations\\ from an intermediate class}
\label{sec:PreliminaryAnalysisOfLieSymsOfClassS}

For a fixed value of the arbitrary element~$S$,
let $\mathfrak g_S$ denote the maximal Lie invariance algebra
of the corresponding equation~$\mathcal L_S$ from the class $\mathscr S$.
Elements of~$\mathfrak g_S$ are vector fields in the space of variables $(t,x,\psi,\psi^*)$
of the form $Q=\tau\p_t+\xi^a\p_a+\eta\p_\psi+\eta^*\p_{\psi^*}$,
where the components $\tau$ and $\xi^a$ (resp.\ $\eta$)
are real-valued (resp.\ complex-valued) smooth functions of $(t,x,\psi,\psi^*)$
that satisfy the infinitesimal invariance criterion for the equation~$\mathcal L_S$,
\begin{equation}\label{eq:ClassSInvCriterion}
Q_{(2)}\left(i\psi_t+\psi_{aa}+S(t,x,\rho)\psi\right)\big|_{\mathcal L_S}
=\left(i\eta^t+\eta^{aa}+(\tau S_t+\xi^a S_a)\psi+\rho S_\rho\eta\right)\big|_{\mathcal L_S}=0.
\end{equation}
Here $\eta^*$ is the complex conjugate of~$\eta$.
$Q_{(2)}$ is the second prolongation of the vector field $Q$
whose components~$\eta^t$ and~$\eta^{ab}$ correspond to the jet coordinates~$\psi_t$ and~$\psi_{ab}$, respectively,
\begin{gather*}
\eta^t=\mathrm D_t\left(\eta-\tau\psi_t-\xi^a\psi_a\right)+\tau\psi_{tt}+\xi^a\psi_{ta},\quad
\eta^{ab}=\mathrm D_a\mathrm D_b\left(\eta-\tau\psi_t-\xi^c\psi_c\right)+\tau\psi_{tab}+\xi^c\psi_{abc}.
\end{gather*}
\begin{subequations}\label{eq:ClassSDetEqs}
Recall that $\mathrm D_t$ and $\mathrm D_a$ are the operators of total derivatives
with respect to~$t$ and $x_a$, respectively.
We expand the condition~\eqref{eq:ClassSInvCriterion}, confine it to the manifold defined by the equation~$\mathcal L_S$
in the underlying jet space by substituting  $\psi_t=i\psi_{aa}+iS\psi$ and~$\psi^*_t=-i\psi^*_{aa}-iS^*\psi^*$
and split the obtained equation
with respect to the parametric derivatives~$\psi_{ta}$, $\psi^*_{ta}$, $\psi_{ab}$, $\psi^*_{ab}$, $\psi_a$ and $\psi^*_a$.
After a re-arrangement, we derive the system of determining equations for the components of a vector field~$Q\in\mathfrak g_S$,
\begin{gather}\label{eq:ClassSDetEqsA}
\tau_\psi=\tau_{\psi^*}=\tau_a=0,\quad
\xi^a_\psi=\xi^a_{\psi^*}=0,\quad
\tau_t=2\xi^1_1=\dots=2\xi^n_n,\quad
\xi^a_b+\xi^b_a=0,\ a\ne b,
\\\label{eq:ClassSDetEqsB}
\eta_{\psi^*}=\eta_{\psi\psi}=0,\quad
2\eta_{\psi a}=i\xi^a_t,\quad \psi\eta_\psi=\eta,
\\[1ex]\label{eq:ClassSDetEqsC}
i\eta_t+\eta_{aa}+(\tau S_t+\xi^a S_a)\psi+\rho S_\rho\mathop{\rm Re}\eta_\psi+\tau_t S=0.
\end{gather}
\end{subequations}
The general solution of the subsystem~\eqref{eq:ClassSDetEqsA}--\eqref{eq:ClassSDetEqsB} is
\begin{gather*}
\tau=\tau(t),\quad
\xi^a=\frac12\tau_tx_a+\kappa_{ab}x_b+\chi^a,\quad
\eta=\left(\frac i8\tau_{tt}x_ax_a+\frac i2\chi^a x_a+i\sigma+\zeta\right)\psi,
\end{gather*}
where~$\tau$, $\chi^a$, $\sigma$ and~$\zeta$ are smooth real-valued functions of~$t$,
and $(\kappa_{ab})$ is a constant skew-symmetric matrix.
Substituting these expressions into the equation~\eqref{eq:ClassSDetEqsC},
we get the classifying condition for the Lie symmetry vector fields of equations from the class~$\mathscr S$.

\begin{theorem}\label{thm:ClassSMIA}
The maximal Lie invariance algebra~$\mathfrak g_S$ of an equation~$\mathcal L_S$ from the class~$\mathscr S$ consists
of the vector fields of the form  $Q=D(\tau)-\sum_{a<b}\kappa_{ab}J_{ab}+P(\chi)+\sigma M+\zeta I$,
where
\begin{gather*}
D(\tau)=\tau\p_t+\frac12\tau_tx_a\p_a+\frac18\tau_{tt}x_ax_aM,\quad
J_{ab}=x_a\p_b-x_b\p_a,\quad a\ne b,\\
P(\chi)=\chi^a\p_a+\frac12\chi^a_tx_aM,\quad
M=i\psi\p_\psi-i\psi^*\p_{\psi^*},\quad
I=\psi\p_\psi+\psi^*\p_{\psi^*},
\end{gather*}
the parameters $\tau$, $\chi^a$, $\sigma$ and $\zeta$ are arbitrary real-valued smooth functions of $t$
and $(\kappa_{ab})$ is an arbitrary constant skew-symmetric $n\times n$ matrix
that together satisfy the classifying condition
\begin{equation}\label{NSchEGMNClassifyingCondition}
\tau S_t+\left(\frac12\tau_tx_a+\kappa_{ab}x_b+\chi^a\right)S_a+\zeta\rho S_\rho+\tau_tS
=\frac18\,\tau_{ttt}x_ax_a+\frac12\chi^a_{tt}x_a+\sigma_t-i\zeta_t-i\frac n4\,\tau_{tt}.
\end{equation}
\end{theorem}
Varying the arbitrary element $S$ and splitting the classifying condition~\eqref{NSchEGMNClassifyingCondition}
with respect to derivatives of $S$,
we find the system of determining equations for elements of the kernel Lie invariance algebra $\mathfrak g^\cap$
of equations from the class~$\mathscr S$, $\tau=\chi=\sigma_t=\zeta=0$ and $\kappa_{ab}=0$,
which implies $\mathfrak g^\cap=\langle M\rangle$.
Choosing various appropriate values of~$S$ in Theorem~\ref{thm:ClassSGequiv},
we show that the common point transformations of equations from the class~$\mathscr S$
are of the form~\eqref{eq:ClassSEquivGroupA}--\eqref{eq:ClassSEquivGroupB}
with $T=t$, $O=E$ and $\mathcal X^a=0=\Sigma_t=Z=0$.

\begin{proposition}\label{kernalinvariancepaper}
The kernel point-symmetry group~$G^\cap$ of equations from the class~$\mathscr S$
is constituted by the point transformations
$\tilde t=t$, $\tilde x=x$, $\tilde\psi=e^{ic}\psi$,
where $c$ is an arbitrary real constant.
The Lie algebra of~$G^\cap$ coincides with $\mathfrak g^\cap=\langle M\rangle$.
\end{proposition}

Consider the linear span $\mathfrak g_\spanindex$ of all the maximal Lie invariance algebras
of equations from the class~$\mathscr S$, $\mathfrak g_\spanindex:=\sum_S\mathfrak g_S$.
For any vector field~$Q$ of the general form from Theorem~\ref{thm:ClassSMIA},
where at least one of the parameters $\tau$, $\kappa_{ab}$, $\chi^a$ and $\zeta$ takes a nonzero value,
there are values of the arbitrary element~$S$ that each satisfies, together with the components of~$Q$, the classifying condition~\eqref{NSchEGMNClassifyingCondition},
and thus the corresponding algebra~$\mathfrak g_S$ contains~$Q$.
Therefore,
\[
\mathfrak g_\spanindex:=\big\langle\,D(\tau),\,J_{ab},\,P(\chi),\,\sigma M,\,\zeta I\big\rangle,
\]
where the parameter functions $\tau$, $\chi^a$, $\sigma$ and~$\zeta$ run through the set of real-valued smooth functions of~$t$.
The nonzero commutation relations between vector fields spanning $\mathfrak g_\spanindex$ are
\begin{gather*}\label{MultiLinSchEqs(1+1)comrel}
[D(\tau^1),D(\tau^2)]=D(\tau^1\tau^2_t-\tau^2\tau^1_t),\quad
[D(\tau),P(\chi)]=P\left(\tau\chi_t-\frac{\tau_t}2\chi\right),\\[.2ex]
[ D(\tau),\sigma M]=\tau\sigma_t M,\quad
[D(\tau),\zeta I]=\tau\zeta_t I,\quad
[J_{ab},J_{bc}]=J_{ac},\quad a\ne b\ne c\ne a,\\[1ex]
[J_{ab},P(\chi)]=P(\hat\chi)
\quad\mbox{with}\quad \hat\chi^a=\chi^b,\quad \hat\chi^b=-\chi^a,\quad \chi^c=0,\quad a\ne b\ne c\ne a,\\[.2ex]
[P(\chi),P(\tilde \chi)]=\frac12\left(\chi^a\tilde\chi^a_t-\tilde\chi^a\chi^a_t\right)M.
\end{gather*}

From these commutation relations, it is clearly seen that
the space~$\mathfrak g_\spanindex$ is closed with respect to the Lie bracket of vector fields and thus is a Lie algebra.
The subspaces
$\langle \sigma M \rangle$,
$\langle\zeta I \rangle$,
$\langle P(\chi),\sigma M\rangle$,
$\langle P(\chi),\sigma M,\zeta I\rangle$,
$\langle J_{ab},P(\chi),\sigma M\rangle$,
$\langle J_{ab},P(\chi),\sigma M,\zeta I\rangle$ and
$\langle D(\tau), P(\chi),\sigma M,\zeta I\rangle$
are ideals of~$\mathfrak g_\spanindex$.
Furthermore, the subspaces $\langle D(\tau)\rangle$, $\langle J_{ab}\rangle$ and $\langle D(\tau),J_{ab}\rangle$ are subalgebras of $\mathfrak g_\spanindex$.

The algebra~$\mathfrak g_\spanindex$ coincides, in view of Corollary~\ref{cor:ClassSgequiv},
with the projection $\pi_*\mathfrak g^\sim_\mathscr S$ of $\mathfrak g^\sim_\mathscr S$ to the space with the coordinates $(t,x,\psi,\psi^*)$.
Therefore, within the framework of the algebraic method, the group classification of the class~$\mathscr S$
reduces to the classification of the appropriate subalgebras of~$\mathfrak g_\spanindex$ up to $\pi_*G^\sim_{\mathscr S}$-equivalence.

\begin{definition}\label{def:AppropriateSubalg}
A subalgebra $\mathfrak s$ of $\mathfrak g_\spanindex$ is said to be {\it appropriate} if there exists an arbitrary smooth function $S$ such that $\mathfrak s= \mathfrak g_S$.
\end{definition}

By $\mathcal D(T)$, $\mathcal J(O)$, $\mathcal P(\mathcal X)$ with $\mathcal X=(\mathcal X^1,\dots,\mathcal X^n)$, $\mathcal M(\Sigma)$ and $\mathcal I(Z)$
we respectively denote the point transformations in the space with the coordinates $(t,x,\psi,\psi^*)$
that are of the form~\eqref{eq:ClassSEquivGroupA}--\eqref{eq:ClassSEquivGroupB},
where the parameter-functions~$T$, $O$, $\mathcal X^a$, $\Sigma$ and~$Z$, successively excluding one of them,
are set to the values corresponding to the identity transformation,
which are $t$ for~$T$, $E$ for~$O$ and zeroes for $\mathcal X^a$, $\Sigma$ and~$Z$.
These elementary transformations generate the entire group~$\pi_*G^\sim$.
The nonidentity pushforward actions of elementary transformations from~$\pi_*G^\sim$
on the vector fields spanning $\mathfrak g_\spanindex$~are
\begin{gather*}
\mathcal D_*(T)D(\tau)=D(\tilde \tau),\quad
\mathcal D_*(T)P(\chi)=P(\tilde \chi),\quad
\mathcal D_*(T)(\sigma M)=\tilde \sigma\tilde M,\quad
\mathcal D_*(T)(\zeta I)=\tilde\zeta\tilde I,\\[1.3ex]
\mathcal J_*(O)P(\chi)=\tilde P(O\chi),\\[1ex]
\mathcal P_*(\mathcal X)D(\tau)=\tilde D(\tau)+\tilde P\left(\tau \mathcal X_t-\frac{\tau_t}2\mathcal X\right)
+\left(\frac{\tau_{tt}}8\mathcal X^a\mathcal X^a-\frac{\tau_t}4\mathcal X^a\mathcal X^a_t-\frac\tau2 \mathcal X^a\mathcal X^a_{tt}\right)\tilde M,\\
\mathcal P_*(\mathcal X)J_{ab}=\tilde J_{ab}+P(\hat{\mathcal X})-\frac12(\mathcal X^a\mathcal X^b_t-\mathcal X^b\mathcal X^a_t)\tilde M,\\
\mathcal P_*(\mathcal X)P(\chi)=\tilde P(\chi)+\frac12(\chi^a\mathcal X^a_t-\chi^a_t\mathcal X^a)\tilde M,\\[1ex]
\mathcal M_*(\Sigma)D(\tau)=\tilde D(\tau)+\tau\Sigma_t\tilde M,\quad
\mathcal I_*(Z)D(\tau)=\tilde D(\tau)+\tau Z_t\tilde I.
\end{gather*}
Here tildes over the vector fields mean that these vector fields are expressed in the new variables,
where $\tilde \tau(\tilde t)=(T_t\tau)(T^{-1}(\tilde t))$,
$\tilde \chi(\tilde t)=(|T_t|^{1/2}\chi)(T^{-1}(\tilde t))$,
$\hat{\mathcal X}^a=\mathcal X^b$,
$\hat{\mathcal X}^b=-\mathcal X^a$,
$\hat{\mathcal X}^c=0,\,c\ne a,b$,
$\tilde \sigma=\sigma(T^{-1}(\tilde t))$,
$\tilde\rho=\rho(T^{-1}(\tilde t))$,
and in each pushforward by~$\mathcal D_*(T)$ we should substitute
the expression for~$t$ given by inverting the relation $\tilde t=T(t)$,
whereas $t=\tilde t$ for the other pushforwards.

\begin{lemma}\label{lem:ClassSIntersectionOfgSAndSigmaMZetaI}
$\mathfrak g_S\cap\langle \sigma M,\zeta I\rangle=\langle M\rangle=\mathfrak g^\cap$ for any $S$ with $(\rho S_\rho)_\rho\ne 0$ or with $S_{a\rho}\ne0$ for some~$a$.
\end{lemma}

\begin{proof}
If $\tau=0$, $\chi^a=0$ and $\kappa_{ab}=0$,
then the classifying condition~\eqref{NSchEGMNClassifyingCondition}
reduces to $\zeta\rho S_\rho=\sigma_t-i\zeta_t$.
Differentiating this equation with respect to~$\rho$ and~$x_a$,
we derive $\zeta(\rho S_\rho)_\rho=\zeta\rho S_{a\rho}=0$, which implies $\zeta=0$.
Then $\sigma_t=0$.
\end{proof}

\begin{lemma}\label{lem:ClassSIntersectionOfgSAndSigmaM}
$\mathfrak g_S\cap\langle \sigma M\rangle=\langle M\rangle$ for any $S$ with $S_\rho\ne 0$.
\end{lemma}

\begin{proof}
For  $\tau=0$, $\chi^a=0$, $\kappa_{ab}=0$ and $\zeta=0$,
the classifying condition~\eqref{NSchEGMNClassifyingCondition}
reduces to $\sigma_t=0$.
\end{proof}

\begin{lemma}\label{lem:ClassSUpperBoundOfDimgS}
$\dim\mathfrak g_S\leqslant\dfrac{n(n+3)}{2}+4$ for any $S$ with $S_\rho\ne 0$,
and this upper bound is the least.
\end{lemma}

\begin{proof}
The proof is similar to the one of Lemma~1 in~\cite{Kurujyibwami&Basarab-Horwath&Popovych2018}.
For each fixed value of the arbitrary element $S$,
the classifying condition~\eqref{NSchEGMNClassifyingCondition} yields
a system of linear ordinary differential equations of the following form:
\begin{gather*}
\tau_{ttt}=\gamma^{00}\tau_t+\gamma^{01}\tau+\gamma^{0,a+1}\chi^a+\gamma^{0,n+2}\zeta+\theta^{0ab}\kappa_{ab},\\
\chi^c_{tt}=\gamma^{c0}\tau_t+\gamma^{c1}\tau+\gamma^{c,a+1}\chi^a+\gamma^{c,n+2}\zeta+\theta^{cab}\kappa_{ab},\\
\zeta_t=-\frac n4\tau_{tt}+\gamma^{n+2,0}\tau_t+\gamma^{n+2,1}\tau+\gamma^{n+2,a+1}\chi^a+\gamma^{n+2,n+2}\zeta+\theta^{n+2,ab}\kappa_{ab},\\
\sigma_t=\gamma^{n+1,0}\tau_t+\gamma^{n+1,1}\tau+\gamma^{n+1,a+1}\chi^a+\gamma^{n+1,n+2}\zeta+\theta^{n+1,ab}\kappa_{ab},
\end{gather*}
where the coefficients~$\gamma^{pq}$ and $\theta^{pab}$,
$p=0,\dots,n+2$, $q=0,\dots,n+2$, $a<b$, are functions of $t$.
From this it is clear that the upper bound of $\dim \mathfrak g_S$ can not exceed
the sum of the number of pairs $(a,b)$ of rotations with $a<b$
and the number of arbitrary constants involved in the general solution of the above system, i.e., $n(n+3)/2+5$.
At the same time, Lemma~\ref{lem:ClassSIntersectionOfgSAndSigmaMZetaI} shows that this number is reduced by $1$
for any $S$ with $(\rho S_\rho)_\rho\ne 0$ or with $S_{a\rho}\ne0$ for some~$a$.
For other values of~$S$, we differentiate the classifying condition~\eqref{NSchEGMNClassifyingCondition}
with respect to~$\rho$.
In view of the constraint $S_\rho\ne 0$, the obtained equation $S_\rho\tau_t+S_{t\rho}\tau=0$ implies
a linear first-order ordinary differential equation in~$\tau$,
and thus the above number is reduced by $2$.

The dimension of $\mathfrak g_S$ for $S=\rho^{4/n}$ coincides with the found upper bound.
\end{proof}

\begin{lemma}\label{lem:ClassSUpperBoundOfDimgSIntersectedByGChiSigmaM}
$\dim\mathfrak g_S\cap\langle P(\chi),\sigma M\rangle \leqslant 2n+1$ for any $S$ with $S_\rho\ne 0$.
\end{lemma}

\begin{proof}
Modifying the proof of Lemma~\ref{lem:ClassSUpperBoundOfDimgS},
we omit the first and penultimate equations of the system from that proof
and set $\tau=0$, $\kappa_{ab}=0$ and $\zeta=0$.
\end{proof}

\section[Schr\"odinger equations with potentials and modular nonlinearity]
{Schr\"odinger equations with potentials\\ and modular nonlinearity}
\label{sec:NSchEPMN}

The class $\mathscr V$ of multidimensional nonlinear Schr\"odinger equations
with potentials and modular nonlinearities of the form~\eqref{MNLinSchEqs_2}
can be embedded into the class~$\mathcal S$ as the subclass~$\tilde{\mathscr V}$ of equations with $S=f(\rho)+V(t,x)$,
where $f$ is an arbitrary complex-valued nonlinearity depending only on $\rho:=|\psi|$ with $f_\rho\ne 0$
and $V$ is an arbitrary smooth complex-valued potential depending on~$t$ and~$x$.
The subclass~$\tilde{\mathscr V}$ is singled out from the class~$\mathscr S$
by the constraints $S_{\rho t}=S_{\rho a}=0$ and $S_\rho\ne 0$ or, equivalently,
\begin{gather}\label{nliconds}
\psi S_{\psi t}+\psi^*S_{\psi^*t}=\psi S_{\psi a}+\psi^*S_{\psi^*a}=0,\quad \psi S_\psi+\psi^*S_{\psi^*}\ne 0.
\end{gather}
Equivalence transformations of the subclass~$\tilde{\mathscr V}$ are exhausted by the elements of the group~$G^\sim_{\mathscr S}$
that preserve, in addition to~\eqref{eq:ClassS} and~\eqref{eq:ClassSAuxiliarySystem},
the constraints~\eqref{nliconds}.
It is obvious from Theorem~\ref{thm:ClassSGequiv}
that not all admissible transformations of the subclass~$\tilde{\mathscr V}$ are generated by its equivalence transformations.

\begin{proposition}\label{pro:GequivNSchEPMNEmbedded}
The class~$\tilde{\mathscr V}$ is not normalized.
The equivalence group $G^\sim_{\tilde{\mathscr V}}$ of this class consists of the point transformations~\eqref{eq:ClassSEquivGroup}
with $T_{tt}=0$ and $Z_t=0$.
\end{proposition}

The class $\mathscr V$ can be interpreted as a reparameterization of the class~$\tilde{\mathscr V}$,
where the parameter functions~$f$ and~$V$ are assumed to be arbitrary elements instead of~$S$.

\begin{proposition}\label{pro:NSchEPMNGequiv}
The class~$\mathscr V$ is not normalized.
The equivalence group $G^\sim_{\mathscr V}$ of this class consists of the point transformations
in the space with the coordinates $(t,x,\psi,\psi^*,f,f^*,V,V^*)$
whose $(t,x,\psi)$-components are of the form~\eqref{eq:ClassSEquivGroupA}--\eqref{eq:ClassSEquivGroupB}
and whose components for the arbitrary elements~$f$ and~$V$ are
\begin{gather*}
\tilde f=\dfrac{\hat f}{|T_t|}+c,\quad
\tilde V=\dfrac{\hat V}{|T_t|}
+\dfrac{\mathcal X^b_{tt}}{2|T_t|^{3/2}}O^{ba}x_a
+\dfrac{\Sigma_t}{T_t}-\dfrac{\mathcal X^a_t\mathcal X^a_t}{4T_t{}^2}-c,
\end{gather*}
where $T$, $\mathcal X^a$ and $\Sigma$ are arbitrary smooth real-valued functions of~$t$
with $T_t\ne 0$ and $T_{tt}=0$, $Z$~is an arbitrary real constant,  $\varepsilon'=\sgn T_t$,
$O=(O^{ab})$~is an arbitrary constant orthogonal $n\times n$  matrix,
and $c$ is an arbitrary complex constant.
\end{proposition}

The appearance of the gauge equivalence transformations, which are associated with the group parameter~$c$,
is related to the ambiguity in the representation of~$S$ as the sum of~$f$ and~$V$.

Generalizing Theorem~6 from~\cite{Popovych&Kunzinger&Eshragi2010} to an arbitrary $n\in\mathbb N$,
we obtain the following assertion, which is motivated by Corollary~\ref{cor:DiffInvForS}.

\begin{theorem}\label{theoremGequivNSchEPMN}
The class~$\mathscr V$ is partitioned into the normalized subclasses~$\mathscr V'$ and $\mathscr P_\lambda$, $\lambda\in\mathbb R$,
which are singled out by the conditions
that the ratio $\rho f_{\rho\rho}/f_\rho$ $({}\equiv\rho S_{\rho\rho}/S_\rho)$ is not a real constant
and that $\rho f_{\rho\rho}/f_\rho=\lambda-1$, respectively,
$\mathscr V=\mathscr V'\bigsqcup\big(\bigsqcup_{\lambda\in\mathbb R}\mathscr P_\lambda\big)$.
There are no point transformations relating equations from different subclasses among the above ones.
\end{theorem}

In other words, the above partition of the class~$\mathscr V$ induces the partition of its equivalence groupoid
$\mathcal G^\sim_{\mathscr V}=\mathcal G^\sim_{\mathscr V'}\bigsqcup\big(\bigsqcup_{\lambda\in\mathbb R}\mathcal G^\sim_{\mathscr P_\lambda}\big)$,
where the equivalence groupoids~$\mathcal G^\sim_{\mathscr V'}$ and~$\mathcal G^\sim_{\mathscr P_\lambda}$
of the subclasses~$\mathscr V'$ and $\mathscr P_\lambda$, $\lambda\in\mathbb R$,
are generated by the corresponding equivalence groups~$G^\sim_{\mathscr V'}$ and~$G^\sim_{\mathscr P_\lambda}$.
Moreover, $G^\sim_{\mathscr V'}=G^\sim_{\mathscr V}$ and $G^\sim_{\mathscr P_\lambda}\supsetneq G^\sim_{\mathscr V}$
for any $\lambda\in\mathbb R$,
cf.\ the results of Sections~\ref{sec:NSchEPMNVf(1+2)D}--\ref{sec:NSchEPPowerMN(1+2)D} below.
%Propositions~\ref{pro:NSchEPMNVfGequiv}, \ref{thm:NSchEPLogMNGequiv} and~\ref{thm:NSchEPMNPlambdaGequiv}
Therefore, a complete list of $\mathcal G^\sim_{\mathscr V}$-inequivalent Lie-symmetry extensions in the class~$\mathscr V$
is exhausted by the union of group-classification lists for the classes~$\mathscr V'$ and $\mathscr P_\lambda$, $\lambda\in\mathbb R$,
with respect to the corresponding equivalence groups.

\begin{remark}
Instead of the partition of the class~$\mathscr V$,
one can consider the associated partition of the class~$\tilde{\mathscr V}$,
where the partition subclasses are singled out by the conditions
that the ratio $\rho S_{\rho\rho}/S_\rho$ is not a real constant
and that $\rho S_{\rho\rho}/S_\rho=\lambda-1$, $\lambda\in\mathbb R$, respectively.
At the same time, one needs to handle these conditions jointly
with conditions $S_{\rho t}=S_{\rho a}=0$ and $S_\rho\ne 0$ for the class~$\mathscr S$.
This is why it is more convenient to work with the reparameterized class~$\mathscr V$
in spite of the fact that this class possesses gauge equivalence transformations.
\end{remark}

Up to gauge equivalence transformations of the class~$\mathscr V$,
for equations from the subclasses~$\mathscr P_0$ and $\mathscr P_\lambda$ with $\lambda\in\mathbb R\setminus\{0\}$
the arbitrary element~$f$ takes the form $f=\delta\ln\rho$ and $f=\delta\rho^\lambda$, respectively,
where $\delta$ is an arbitrary nonzero complex constant.
We reparameterize these classes, assuming $\delta$ as an arbitrary element instead of~$f$
and preserving the notation for these classes.
Then gauge equivalence transformations are neglected.

\subsection{General case of modular nonlinearity}\label{sec:NSchEPMNVf}

As discussed above, the class $\mathscr V'$ consists of the equations of the form
\begin{gather}\label{groupclassificationofsubclassnonlineartity}
i\psi_t+\psi_{aa}+f(\rho)\psi+V(t,x)\psi=0\quad\mbox{with}\quad f_\rho\ne 0,\ \ \rho f_{\rho\rho}/f_\rho\ne\const\in\mathbb R.
\end{gather}

Proposition~\ref{pro:NSchEPMNGequiv} implies the following assertion.

\begin{corollary}
An equation from the class $\mathscr V'$ with a potential~$V$ is reduced by a point transformation
to an equation from the same class with the zero potential if and only if
the potential~$V$ is real-valued and affine in~$x$.
\end{corollary}

\begin{remark}\label{rem:NSchEPMNVf}
The action of each element of the group $G^\sim_{\mathscr V}$ on $f$
is the composition of complex conjugation, multiplication by a nonzero real constant, shift in a complex constant
and rescaling its argument, which does not change the form of~$f$ essentially.
This is why in the course of group classification of the class $\mathscr V'$,
we can assume that the nonlinearity $f$ is fixed and the only arbitrary element is $V$.
By $\mathscr V^f$ we denote the subclass of equations in $\mathscr V'$ with a fixed value of the arbitrary element~$f$.
$\mathscr V^f=\mathscr V^{\tilde f}$ if and only if $\tilde f-f=\const\in\mathbb C$.
Hence the class~$\mathscr V'$ can be interpreted as the disjoint union of its subclasses~$\mathscr V^f$,
where $f$ runs through the set of complex-valued smooth functions of $\rho:=|\psi|$
with $f_\rho\ne 0$ and $\rho f_{\rho\rho}/f_\rho\ne\const\in\mathbb R$
modulo adding complex constants.
\end{remark}

\begin{proposition}\label{pro:NSchEPMNVfGequiv}
The class~$\mathscr V^f$ is normalized.
Its equivalence group~$G^\sim_{\mathscr V^f}$ is constituted by the projections of elements of~~$G^\sim_{\mathscr V}$
on the space with the coordinates $(t,x,\psi,\psi^*,V,V^*)$,
where $T_t=1$ (resp.\ $T_t=\pm 1$ if $f$ is a real-valued function), $Z=0$ and $c=0$.
\end{proposition}

\begin{lemma} \label{lem:NSchEPMNVfMIA}
The maximal Lie invariance algebra~$\mathfrak g_V$ of an equation~$\mathcal L_V$ from~$\mathscr V^f$
with $\rho f_{\rho\rho}/f_\rho$ not being a real constant
consists of the vector fields of the form $D(c)-\sum_{a<b}\kappa_{ab}J_{ab}+P(\chi)+\sigma M$,
where $c$ is an arbitrary real constant,
$(\kappa_{ab})$ is an arbitrary constant skew-symmetric $n\times n$ matrix
and the parameter functions $\chi^a$ and $\sigma$ are arbitrary real-valued smooth functions of $t$
that satisfy the condition
\begin{gather}\label{eq:NSchEPMNVfClassifyingCondition}
cV_t+(\kappa_{ab}x_b+\chi^a)V_a=\frac12 \chi^a_{tt}x_a+\sigma_t.
\end{gather}
The kernel Lie invariance algebra of equations from the class $\mathscr V^f$ is $\mathfrak g^\cap_{\mathscr V^f}=\langle M\rangle$.
\end{lemma}

\begin{proof}
Substituting $S=f(\rho)+V(t,x)$ into the classifying condition~\eqref{NSchEGMNClassifyingCondition}
and recalling that $\rho f_{\rho\rho}/f_\rho\ne\const\in\mathbb R$, we derive $\tau_t=0$ and $\zeta=0$.
The algebra~$\mathfrak g^\cap_{\mathscr V^f}$ is obtained by varying the arbitrary element $V$ and splitting with respect to its derivatives.
\end{proof}

Any vector field of the general form presented in Lemma~\ref{lem:NSchEPMNVfMIA},
where at least one of the parameters $c$, $\kappa_{ab}$ and $\chi^a$ takes a nonzero value,
belongs to~$\mathfrak g_V$
for a potential~$V$ satisfying the classifying condition~\eqref{eq:NSchEPMNVfClassifyingCondition}
for this vector field.
This is why we have
\[\textstyle
\mathfrak g_\spanindex:=\sum_V\mathfrak g_V=\langle\,D(1),\,J_{ab},\,a<b,\,P(\chi),\,\sigma M\,\rangle,
\]
where the parameter functions $\chi^a$ and $\sigma$ run through the set of real-valued smooth functions of~$t$.
This linear span is closed with respect to the Lie bracket of vector fields and thus is a Lie algebra,
which coincides, in view of Proposition~\ref{pro:NSchEPMNVfGequiv},
with the projection of the equivalence algebra of the class~$\mathscr V^f$ to the space with the coordinates $(t,x,\psi,\psi^*)$.
Obtaining the above coincidence is the main incentive for fixing~$f$.
As a result, the action of the group~$\pi_*G^\sim_{\mathscr V^f}$ on the algebra~$\mathfrak g_\spanindex$
is naturally consistent with this algebra.
Therefore, within the framework of the algebraic method, the group classification of the class~$\mathscr V^f$
reduces to the classification of appropriate subalgebras of~$\mathfrak g_\spanindex$ up to \smash{$\pi_*G^\sim_{\mathscr V^f}$}-equivalence.
The following conditions hold for any equation~$\mathcal L_V$ from the class~$\mathscr V^f$,
and thus for any appropriate subalgebra of~$\mathfrak g_\spanindex$,
which by Definition~\ref{def:AppropriateSubalg} coincides with $\mathfrak g_V$ for some~$V$:
\begin{gather*}
\dim\mathfrak g_V\leqslant\dfrac{n(n+3)}2+2,\quad
\dim\mathfrak g_V\cap\langle P(\chi),\sigma M\rangle \leqslant 2n+1,\quad
\mathfrak g_V\cap\langle \sigma M\rangle =\langle M\rangle.
\end{gather*}

\subsection{Logarithmic modular nonlinearity}\label{sec:NSchEPLogMN}

The class~$\mathscr P_0$ consists of the equations of the form
\begin{equation}\label{eq:NSchEPLogMN}
i\psi_t+\psi_{aa}+\delta\psi\ln\rho +V(t,x)\psi=0,
\end{equation}
where $\delta$ is an arbitrary nonzero complex number, $\delta=\delta_1+i\delta_2$, $\delta_1,\delta_2\in\mathbb R$,
and $V$ is an arbitrary complex-valued function of~$t$ and~$x$.
The corresponding subclass of the class~$\tilde{\mathscr V}$ is singled out from the class~$\mathscr S$ by
the constraints $S_{\rho t}=S_{\rho a}=0$ and $(\rho S_\rho)_\rho=0$,
i.e., $\psi S_{\psi t}+\psi^*S_{\psi^* t}=\psi S_{\psi a}+\psi^*S_{\psi^* a}=0$ and
$(\psi\p_\psi+\psi^*\p_{\psi^*})^2S=0$.

We can find the point transformations connecting two equations from the class $\mathscr P_0$ by the direct method.
However, we have already described the equivalence groupoid~$\mathcal G^\sim_{\mathscr S}$ of the class~$\mathscr S$ in Theorem~\ref{thm:ClassSGequiv},
so we can use this description and single out the equivalence groupoid of the class~$\mathscr P_0$
as a subgroupoid of~$\mathcal G^\sim_{\mathscr S}$,
substituting $S=\delta\ln\rho+V(t,x)$ and $\tilde S=\tilde\delta\ln\tilde\rho+\tilde V(\tilde t,\tilde x)$ into~\eqref{eq:ClassSEquivGroupC}.

\begin{theorem}\label{thm:NSchEPLogMNGequiv}
The class $\mathscr P_0$ is normalized.
Its equivalence group $G^\sim_{\mathscr P_0}$ is constituted
by the point transformations in the space with the coordinates $(t,x,\psi,\psi^*,\delta,\delta^*,V,V^*)$
whose components for the variables are of the form~\eqref{eq:ClassSEquivGroupA}--\eqref{eq:ClassSEquivGroupB}
and whose components for the arbitrary elements~$\delta$ and~$V$ are
\begin{gather}\label{eq:NSchEPLogMNGequiv}
\tilde\delta=\dfrac{\hat\delta}{|T_t|},\quad
\tilde V=\dfrac{\hat V}{|T_t|}
+\dfrac{\mathcal X^b_{tt}}{2|T_t|^{3/2}}O^{ba}x_a-\hat\delta\dfrac{Z}{|T_t|}
-\frac 14\dfrac{\mathcal X^a_t\mathcal X^a_t}{T_t^2}+\dfrac{\Sigma_t-iZ_t}{T_t},
\end{gather}
where the parameter functions $T$, $\mathcal X^a$, $Z$ and $\Sigma$ are arbitrary smooth real-valued functions of~$t$
with $T_t\ne 0$ and $T_{tt}=0$,
$\varepsilon'=\sgn T_t$ and $O=(O^{ab})$~is an arbitrary constant orthogonal $n\times n$ matrix.
\end{theorem}

\begin{corollary}\label{cor:NSchELogMNVanishingPotentail}
An equation from the class $\mathscr P_0$ with a potential~$V$ is reduced by a point transformation
to an equation from the same class with the zero potential if and only if
the potential~$V$ is affine in~$x$ and the coefficients of $x_a$ are real-valued,
i.e., $V_{ab}=0$ and $V_a$ are real-valued.
\end{corollary}

\begin{remark}\label{rem:NSchEPLogMNP0delta}
Theorem~\ref{thm:NSchEPLogMNGequiv} implies that any point transformation connecting
two equations in the class~$\mathscr P_0$ acts on $\delta$
by multiplication with a nonzero real constant and by complex conjugation.
Similarly to Remark~\ref{rem:NSchEPMNVf},
we can fix an arbitrary value of~$\delta$ and assume $V$ the only arbitrary element.
By $\mathscr P_0^\delta$ we denote the subclass of equations in $\mathscr P_0$ with a fixed value of~$\delta$.
The class~$\mathscr P_0$ can be interpreted as the disjoint union of its subclasses~\smash{$\mathscr P_0^\delta$},
where $\delta$ runs through the set of nonzero complex numbers.
Theorem~\ref{thm:NSchEPLogMNGequiv} implies
that equations from subclasses~\smash{$\mathscr P_0^\delta$} and~\smash{$\mathscr P_0^{\tilde\delta}$}
are related by point transformations if and only if
$\tilde\delta/\delta\in\mathbb R_{>0}$ or $\tilde\delta/\delta^*\in\mathbb R_{>0}$,
and then the point transformation~$\mathcal D(T)$
with $T=\delta\tilde\delta^{-1}t$ or with $T=-\delta^*\tilde\delta^{-1}t$, respectively,
maps the entire subclass~\smash{$\mathscr P_0^\delta$} onto the subclass~\smash{$\mathscr P_0^{\tilde\delta}$}.
In other words, up to $G^\sim_{\mathscr P_0}$-equivalence
one can set the constraints $|\delta|=1$ and $\mathop{\rm Im}\delta\geqslant 0$ on~$\delta$.
\end{remark}

\begin{proposition}\label{pro:NSchEPLogMNP0deltaGequiv}
The class~$\mathscr P_0^\delta$ is normalized.
The equivalence group~\smash{$G^\sim_{\mathscr P_0^\delta}$} of this class
is constituted by the projections of elements of~$G^\sim_{\mathscr P_0}$
on the space with the coordinates $(t,x,\psi,\psi^*,V,V^*)$,
where $T_t=1$ if $\delta_2:=\mathop{\rm Im}\delta\ne0$ and $T_t=\pm 1$ if $\delta_2=0$.
\end{proposition}

Setting $S=\delta\ln \rho+V(t,x)$ in the classifying condition~\eqref{NSchEGMNClassifyingCondition}
and splitting with respect to $\rho$ yield $\tau_t=0$.

\begin{lemma}\label{lem:NSchEPLogMNP0deltaMIA}
Any vector field $Q$ from the  maximal Lie invariance algebra~$\mathfrak g_V$ of an equation~$\mathcal L_V$ from the subclass~$\mathscr P_0^\delta$ is of the form $D(c)-\sum_{a<b}\kappa_{ab}J_{ab}+P(\chi)+\sigma M+\zeta I$,
where $c$ is an arbitrary real constant,
$(\kappa_{ab})$ is an arbitrary constant skew-symmetric $n\times n$ matrix
and the parameter functions $\chi^a$, $\sigma$ and $\zeta$ are arbitrary real-valued smooth functions of~$t$
that satisfy the classifying condition
\begin{gather}\label{eq:NSchEPLogMNP0deltaClassifyingCondition}
cV_t+(\kappa_{ab}x_b+\chi^a)V_a=\frac12 \chi^a_{tt}x_a+\sigma_t-i\zeta_t-\delta\zeta.
\end{gather}
\end{lemma}

\begin{lemma}
The kernel Lie invariance algebra of equations from the class~$\mathscr P_0^\delta$
is~\smash{$\mathfrak g^\cap_{\mathscr P_0^\delta}=\langle M,I'\rangle$},
where $I'=e^{-\delta_2t}(\delta_2I-\delta_1M)$ if $\delta_2\ne0$ and $I'=I+\delta_1tM$ if $\delta_2=0$.
The kernel point-symmetry group %~\smash{$G^\cap_{\mathscr P_0^\delta}$}
of this class consists of the transformations of the form~\eqref{eq:ClassSEquivGroupA}--\eqref{eq:ClassSEquivGroupB} with
$T=t$, $\mathcal X=0$, $O=E$,
$Z=Z_0\delta_2e^{-\delta_2 t}$,
$\Sigma=\Sigma_0-Z_0 \delta_1 e^{-\delta_2 t}$ if $\delta_2\ne0$
and $Z=Z_0$, $\Sigma=\Sigma_0+Z_0 \delta_1 t$ if $\delta_2=0$,
where $\Sigma_0$ and $Z_0$ are arbitrary real constants.
\end{lemma}

Any vector field of the general form from Lemma~\ref{lem:NSchEPLogMNP0deltaMIA},
where at least one of the parameters $c$, $\kappa_{ab}$ and $\chi^a$ takes a nonzero value,
belongs to~$\mathfrak g_V$ for a potential~$V$
satisfying the classifying condition~\eqref{eq:NSchEPLogMNP0deltaClassifyingCondition} for this vector field.
This is why we have
\[\textstyle
\mathfrak g_\spanindex:=\sum_V\mathfrak g_V=\langle\, D(1),\,J_{ab},\,a<b,\,P(\chi),\,\sigma M,\,\zeta I\,\rangle,
\]
where the parameter functions $\chi^a$, $\sigma$ and $\zeta$ run through the set of real-valued smooth functions of~$t$.
This linear span is closed with respect to the Lie bracket of vector fields and thus is a Lie algebra,
which coincides, in view of Proposition~\ref{pro:NSchEPLogMNP0deltaGequiv},
with the projection of the equivalence algebra of the class~$\mathscr P_0^\delta$ to the space with the coordinates $(t,x,\psi,\psi^*)$.
The above coincidence obtained due to fixing~$\delta$ leads to
the consistency of the action of the group~\smash{$\pi_*G^\sim_{\mathscr P_0^\delta}$} with the algebra~$\mathfrak g_\spanindex$.
Therefore, within the framework of the algebraic method, the group classification of the class~$\mathscr P_0^\delta$
reduces to the classification of appropriate subalgebras of~$\mathfrak g_\spanindex$ up to \smash{$\pi_*G^\sim_{\mathscr P_0^\delta}$}-equivalence.
The following conditions hold for any equation~$\mathcal L_V$ from the class~$\mathscr P_0^\delta$,
and thus for any appropriate subalgebra of~$\mathfrak g_\spanindex$,
which by Definition~\ref{def:AppropriateSubalg} coincides with $\mathfrak g_V$ for some~$V$:
\begin{gather*}
\dim\mathfrak g_V\leqslant \frac{n(n+3)}2+3,\quad
\dim\mathfrak g_V\cap\langle P(\chi),\sigma M, \zeta I\rangle \leqslant 2n+2,\\
\mathfrak g_V\cap \langle \sigma M, \zeta I\rangle =\langle M, I'\rangle=\mathfrak g^\cap_{\mathscr P_0^\delta}.
\end{gather*}

\subsection{Power modular nonlinearity}\label{sec:NSchEPPowerMN}

For any fixed nonzero real constant~$\lambda$,
the class~$\mathscr P_\lambda$ of nonlinear Schr\"odinger equations with potentials
and power nonlinearity consists of the equations of the form
\begin{gather}\label{subclasspowernonlinerity}
i\psi_t+\psi_{aa}+\delta\rho^\lambda\psi+V(t,x)\psi=0.
\end{gather}
Here $\delta$ is an arbitrary nonzero complex constant
and $V$ is an arbitrary complex-valued potential depending on $t$ and $x$.
The subclass of the class~$\tilde{\mathscr V}$ corresponding to the class~$\mathscr P_\lambda$
is singled out from the class~$\mathscr S$ 
by the constraints $S_{\rho t}=S_{\rho a}=0$, $(\rho S_\rho)_\rho=\lambda S_\rho$ or, equivalently, by the constraints
\begin{gather}\label{persubclasspowernonlinerity1}
\psi S_{\psi t}+\psi^*S_{\psi^*t}=\psi S_{\psi a}+\psi^*S_{\psi^*a}=0,\quad
(\psi \p_{\psi}+\psi^*\p_{\psi^*})^2S=\lambda(\psi \p_\psi+\psi^*\p_{\psi^*})S.
 \end{gather}

Using the description of the equivalence groupoid~$\mathcal G^\sim_{\mathscr S}$ of the class~$\mathscr S$ in Theorem~\ref{thm:ClassSGequiv},
we single out the equivalence groupoid of the class~$\mathscr P_\lambda$ as a subgroupoid of~$\mathcal G^\sim_{\mathscr S}$,
substituting $S=\delta\rho^\lambda+V(t,x)$ and $\tilde S=\tilde\delta\tilde\rho^\lambda+\tilde V(\tilde t,\tilde x)$ into~\eqref{eq:ClassSEquivGroupC}.

\begin{theorem}\label{thm:NSchEPMNPlambdaGequiv}
The class $\mathscr P_\lambda$ with $\lambda\in\mathbb R_{\ne0}$ is normalized.
Its equivalence group $G^\sim_{\mathscr P_\lambda}$ consists of the point transformations
in the space with the coordinates $(t,x,\psi,\psi^*,\delta,\delta^*,V,V^*)$
whose $(t,x,\psi)$-components are of the form~\eqref{eq:ClassSEquivGroupA}--\eqref{eq:ClassSEquivGroupB}
and whose components for the arbitrary elements~$\delta$ and~$V$ are
\begin{gather*}
\tilde\delta=\dfrac{\hat\delta}{\mu^\lambda},\quad
\tilde V=\dfrac{\hat V}{|T_t|}
+\frac{2T_{ttt}T_t-3T_{tt}{}^2}{16\varepsilon'T_t{}^3}x_ax_a
+\frac{\varepsilon'}2\left(\frac{\mathcal X^b_t}{T_t}\right)_{\!t}\frac{O^{ba}x_a}{|T_t|^{1/2}}
+\frac{\Sigma_t}{T_t}-\frac{\mathcal X^a_t\mathcal X^a_t}{4T_t{}^2}+i\lambda'\frac{T_{tt}}{T_t^2},
\end{gather*}
where $\lambda':=1/\lambda-n/4$,
$T$, $\mathcal X^a$ and $\Sigma$ are real-valued functions of $t$ with $T_t\ne0$,
and $e^Z=\mu |T_t|^{-1/\lambda}$ with real constant $\mu>0$.
\end{theorem}

\begin{corollary}\label{cor:NSchEPPowerMNReductionPotToXIndep}
A $(1+n)$-dimensional nonlinear Schr\"odinger equation of the form~\eqref{subclasspowernonlinerity}
with a potential~$V$ is reduced by a point transformation to an equation of the same form with potential independent of $x$
if and only if
\begin{gather}\label{persubclasspowernonlinerity4}
V=h(t)x_ax_a+h^a(t)x_a+\tilde h^0(t)+ih^0(t),
\end{gather}
where $h$, $h^a$, $h^0$ and $\tilde h^0$ are real-valued functions of $t$.
Moreover, the transformed potential can be assumed imaginary-valued.
\end{corollary}

\begin{corollary}
A $(1+n)$-dimensional nonlinear Schr\"odinger equation of the form~\eqref{subclasspowernonlinerity}
with a potential~$V$ is reduced by a point transformation to an equation of the same form with zero potential
if and only if $V$ is of the form~\eqref{persubclasspowernonlinerity4},
where $16(\lambda')^2 h=2\lambda'h^0_t+(h^0)^2$.
If $\lambda=4/n$, then the condition for~$V$ means that $V$ is an arbitrary real-valued $x$-quadratic potential.
\end{corollary}

\begin{proof}
The expression~\eqref{persubclasspowernonlinerity4} for $V$ that is reducible to zero
is obtained by setting~$\tilde V=0$ in Theorem~\ref{thm:NSchEPMNPlambdaGequiv}.
The representations
\[
h^0=-\lambda'\frac{T_{tt}}{T_t}, \quad h=-\frac{2T_{ttt}T_t-3T^2_{tt}}{16T^2_t}
\]
are jointly equivalent to the equality $16(\lambda')^2 h=2\lambda'h^0_t+(h^0)^2$.
\end{proof}

\begin{remark}
Instead of the family of the subclasses $\mathscr P_\lambda$ parameterized by $\lambda\in\mathbb R_{\ne0}$,
one can study the entire class~$\mathscr P$ of (1+$n$)-dimensional nonlinear Schr\"odinger equations
with potentials and power modular nonlinearities,
i.e., the class of equations of the form~\eqref{subclasspowernonlinerity},
where $\lambda$ is assumed to be one more constant arbitrary element running through $\mathbb R_{\ne0}$.
We can interpret the class~$\mathscr P$ as the disjoint union of the subclasses $\mathscr P_\lambda$ with $\lambda\in\mathbb R_{\ne0}$.
A drawback of the above approach is the need to consider the generalized equivalence group of~$\mathscr P$,
where the $x$-component of equivalence transformations will depend on~$\lambda$, see Theorem~\ref{thm:NSchEPMNPlambdaGequiv},
since the class~$\mathscr P$ is normalized in the generalized sense only.
\end{remark}

\begin{remark}%\looseness=-1
Arguing similarly to Remark~\ref{rem:NSchEPLogMNP0delta} using Theorem~\ref{thm:NSchEPMNPlambdaGequiv},
we fix an arbitrary value of~$\delta$ and assume $V$ the only arbitrary element.
By $\mathscr P_\lambda^\delta$ we denote the subclass of equations in $\mathscr P_\lambda$ with a fixed value of~$\delta$.
The class~$\mathscr P_\lambda$ can be interpreted as the disjoint union of its subclasses~$\mathscr P_\lambda^\delta$,
where $\delta$ runs through the set of nonzero complex numbers.
Theorem~\ref{thm:NSchEPMNPlambdaGequiv} implies
that equations from subclasses~\smash{$\mathscr P_\lambda^\delta$} and~\smash{$\mathscr P_\lambda^{\tilde\delta}$}
are related by point transformations if and only if
$\tilde\delta/\delta\in\mathbb R_{>0}$ or $\tilde\delta/\delta^*\in\mathbb R_{>0}$,
and then the point transformation~$\mathcal D(T)$
with $T=\delta\tilde\delta^{-1}t$ or with $T=-\delta^*\tilde\delta^{-1}t$, respectively,
maps the entire subclass~\smash{$\mathscr P_0^\delta$} onto the subclass~\smash{$\mathscr P_0^{\tilde\delta}$}.
In other words, up to $G^\sim_{\mathscr P_0}$-equivalence
one can set the constraints $|\delta|=1$ and $\mathop{\rm Im}\delta\geqslant0$ on~$\delta$.
\end{remark}

\begin{proposition}\label{pro:NSchEPMNPlambdadeltaGequiv}
The class~$\mathscr P_\lambda^\delta$ is normalized.
The equivalence group~\smash{$G^\sim_{\mathscr P_\lambda^\delta}$} of this class
is constituted by the projections of elements of~$G^\sim_{\mathscr P_\lambda}$
onto the space with the coordinates $(t,x,\psi,\psi^*,V,V^*)$,
where $\mu=1$ and, if $\mathop{\rm Im}\delta\ne0$, $T_t>0$.
\end{proposition}

Substituting $S=\delta\rho^\lambda+V(t,x)$ into the classifying condition~\eqref{NSchEGMNClassifyingCondition}
and splitting with respect to~$\rho$, we obtain the determining equation $\lambda\zeta +\tau_t=0$
and the classifying condition for the class~$\mathscr P_\lambda^\delta$,
\begin{equation}\label{eq:NSchEPMNPlambdadeltaClassifyingCondition}
\tau V_t+\left(\frac12\tau_tx_a+\kappa_{ab}x_b+\chi^a\right)V_a+\tau_tV=\frac18\,\tau_{ttt}x_ax_a+\frac12\chi^a_{tt}x_a+\sigma_t+i{\lambda'}\tau_{tt}.
\end{equation}
As a result, Theorem~\ref{thm:ClassSMIA}  implies the following assertion.

\begin{lemma}\label{lem:NSchEPMNPlambdadeltaMIA}
Any vector field $Q$ from the maximal Lie invariance algebra~$\mathfrak g_V$
of an equation~$\mathcal L_V$ from the class~$\mathscr P_\lambda^\delta$
is of the form $D^\lambda(\tau)-\sum_{a<b}\kappa_{ab}J_{ab}+P(\chi)+\sigma M$,
where $D^\lambda(\tau)=D(\tau)-\lambda^{-1}\tau_t I$,
$(\kappa_{ab})$ is an arbitrary constant skew-symmetric $n\times n$ matrix,
and the parameter functions $\tau$, $\chi^a$ and $\sigma$ are arbitrary real-valued smooth functions of~$t$
that satisfy the classifying condition~\eqref{eq:NSchEPMNPlambdadeltaClassifyingCondition}.
\end{lemma}

\begin{lemma}
The kernel Lie invariance algebra of the equations from the class~$\mathscr P_\lambda^\delta$
is \smash{$\mathfrak g^\cap_{\mathscr P_\lambda^\delta}=\langle M\rangle$}.
\end{lemma}

Any vector field of the general form from Lemma~\ref{lem:NSchEPMNPlambdadeltaMIA},
where at least one of the parameters $\tau$, $\kappa_{ab}$, $\chi^a$ and $\sigma$ takes a nonzero value,
belongs to~$\mathfrak g_V$ for a potential~$V$
satisfying the classifying condition~\eqref{eq:NSchEPMNPlambdadeltaClassifyingCondition} for this vector field.
This is why we have
\[\textstyle
\mathfrak g_\spanindex:=\sum_V\mathfrak g_V=\langle\, D^\lambda(\tau),\,J_{ab},\,a<b,\,P(\chi),\,\sigma M\,\rangle,
\]
\looseness=-1
where the parameter functions $\tau$, $\chi^a$ and $\sigma$ run through the set of real-valued smooth functions of~$t$.
This linear span is closed with respect to the Lie bracket of vector fields and thus is a Lie algebra,
which coincides, in view of Proposition~\ref{pro:NSchEPMNPlambdadeltaGequiv},
with the projection of the equivalence algebra of the class~$\mathscr P_\lambda^\delta$ to the space with the coordinates $(t,x,\psi,\psi^*)$.
The above coincidence is again both a consequence of and a justification for fixing~$\delta$
since the action of the group~\smash{$\pi_*G^\sim_{\mathscr P_\lambda^\delta}$} is then consistent with the algebra~$\mathfrak g_\spanindex$.
Therefore, within the framework of the algebraic method, the group classification of the class~$\mathscr P_\lambda^\delta$
reduces to the classification of the appropriate subalgebras of~$\mathfrak g_\spanindex$ up to~\smash{$\pi_*G^\sim_{\mathscr P_\lambda^\delta}$}-equivalence.
The following conditions hold for any equation~$\mathcal L_V$ from the class~$\mathscr P_\lambda^\delta$,
and thus for any appropriate subalgebra of~$\mathfrak g_\spanindex$,
which by Definition~\ref{def:AppropriateSubalg} coincides with $\mathfrak g_V$ for some~$V$:
\begin{gather*}
\dim\mathfrak g_V\leqslant \frac{n(n+3)}2+4,\quad
\mathfrak g_V\cap\langle\sigma M\rangle=\langle M\rangle,\quad
\dim\mathfrak g_V\cap\langle P(\chi),\sigma M\rangle \leqslant 2n+1.
\end{gather*}

\begin{lemma}\label{lem:NSchEPMNPlambdadeltaPigV}
For all $V$, $\pi^0_*\mathfrak g_V $ is a Lie algebra and $\dim\pi^0_*\mathfrak g_V\leqslant 3$.
Moreover,
\[
\pi^0_*\mathfrak g_V\in\{\{0\},\langle \p_t\rangle,\langle \p_t,t\p_t\rangle,\langle \p_t,t\p_t,t^2\p_t\rangle\}
\bmod \pi^0_*G^\sim_{\mathscr P_\lambda},\quad \text{where}\quad\pi^0_*\mathfrak g_\spanindex=\langle \tau\p_t\rangle,
\]
and $\pi^0$ is the projection onto the space with the coordinate~$t$.
\end{lemma}

\section{Group classification in dimension (1+2)}\label{sec:NSchEPMN(1+2)D}

We present the complete group classification of the class $\mathscr V$ for $n=2$.
We recall that this class is not normalized
but it is partitioned into the normalized subclasses
$\mathscr V'$, $\mathscr P_0$ and $\mathscr P_\lambda$, $\lambda\in\mathbb R_{\ne0}$,
consisting of equations
whose general forms are~\eqref{groupclassificationofsubclassnonlineartity},~\eqref{eq:NSchEPLogMN}
and~\eqref{subclasspowernonlinerity}, respectively.
Following results of Section~\ref{sec:NSchEPMN},
we consider the classes~$\mathscr V^f$, $\mathscr P_0^\delta$ and~$\mathscr P_\lambda^\delta$
instead of $\mathscr V'$, $\mathscr P_0$ and $\mathscr P_\lambda$, respectively.

In this section, the maximal Lie invariance algebra of the equation from one of the classes under consideration
with a potential~$V$ and the corresponding nonlinearity is denoted by~$\mathfrak g_V$.
The indices $a$ and $b$ run from~1 to~2, and we assume summation over repeated indices.
Since here the collection $\{J_{ab},a<b\}$ consists of the single vector field~$J_{12}$,
for convenience we re-denote $J:=J_{12}$ and $\kappa:=-\kappa_{12}=\kappa_{21}$.
In each subsection of this section, we will omit the indication of the class under study
in the notation of its equivalence group, shortly denoting this group by~$G^\sim$.
We also use the notations
\begin{gather*}
|x|=\sqrt{x_1^2+x_2^2},\quad
\phi=\arctan\frac{x_2}{x_1},\quad
\omega_1=x_1\cos\kappa t+x_2\sin\kappa t,\quad
\omega_2=-x_1\sin\kappa t+x_2\cos\kappa t.
\end{gather*}
For each appropriate subalgebra $\mathfrak s$ of the associated linear span $\mathfrak g_\spanindex$,
we also introduce five integers that are invariant under equivalence transformations of the class under study.
We define these integers as
\begin{gather*}
r_1:=\rank\{\chi\mid \exists\,\sigma\colon P(\chi)+\sigma M\in \mathfrak s\},\\
k_0:=\dim\mathfrak s\cap\langle\sigma M\rangle=\dim\mathfrak g^\cap=1 \quad\mbox{since}\quad \mathfrak s\supseteq\mathfrak g^\cap=\langle M\rangle,\\
k_1:=\dim\mathfrak s\cap\langle P(\chi),\,\sigma M\rangle-k_0,\\
k_2:=\dim\mathfrak s\cap\langle J,\,P(\chi),\,\sigma M \rangle-k_1-k_0,\\
k_3:=\dim\mathfrak s-k_2-k_1-k_0=\dim\pi^0_*\mathfrak s
\end{gather*}
for the classes~$\mathscr V^f$ and $\mathscr P_0^\delta$ in Sections~\ref{sec:NSchEPMNVf(1+2)D} and~\ref{sec:NSchEPPowerMN(1+2)D}
and as
\begin{gather*}
r_1:=\rank\{\chi\mid \exists\,\sigma,\zeta\colon P(\chi)+\sigma M+\zeta I\in \mathfrak s\},\\[.5ex]
k_0:=\dim\mathfrak s\cap\langle \sigma M,\,\zeta I\rangle=\dim\mathfrak g^\cap=2 \quad\mbox{since}\quad \mathfrak s\supseteq\mathfrak g^\cap=\langle M,\,I\rangle,\\
k_1:=\dim\mathfrak s\cap\langle P(\chi),\,\sigma M,\,\zeta I \rangle-k_0,\\[.5ex]
k_2:=\dim\mathfrak s\cap\langle J,\,P(\chi),\,\sigma M,\,\zeta I\rangle-k_1-k_0,\\[.5ex]
k_3:=\dim\mathfrak s-k_2-k_1-k_0=\dim\pi^0_*\mathfrak s
\end{gather*}
for the class~$\mathscr P_\lambda^\delta$ in Section~\ref{sec:NSchEPLogMN(1+2)D}.
Here and in what follows $\pi^0$ is the projection onto the space of the variable $t$, 
and the parameter functions $\chi^1$, $\chi^2$, $\sigma$ and~$\zeta$ in linear spans run through the set of real-valued smooth functions of~$t$. 
It is obvious that in each case 
\[
\dim\mathfrak s=k_0+k_1+k_2+k_3,\quad r_1\in\{0,1,2\},\quad k_2\in\{0,1\},\quad r_1\leqslant k_1,
\]
%$\dim\mathfrak s=k_0+k_1+k_2+k_3$, $r_1\in\{0,1,2\}$, $k_2\in\{0,1\}$, $r_1\leqslant k_1$,
and $r_1=0$ if and only if $k_1=0$.
Results of Sections~\ref{sec:PreliminaryAnalysisOfLieSymsOfClassS} and~\ref{sec:NSchEPMN} imply that $k_1\in \{0,\dots,4\}$;
%$k_0=1$ within the classes~$\mathscr V^f$ and $\mathscr P_\lambda^\delta$ and $k_0=2$ within the class~$\mathscr P_0^\delta$
$k_3\in\{0,1\}$ within the classes~$\mathscr V^f$ and $\mathscr P_0^\delta$ and $k_3\in\{0,1,2,3\}$ within the class~$\mathscr P_\lambda^\delta$,
and $\dim\mathfrak s\leqslant7,8,9$ for the classes~$\mathscr V^f$, $\mathscr P_0^\delta$ and~$\mathscr P_\lambda^\delta$, respectively.

\begin{lemma}\label{lem:NSchEPMN(1+2)Dr1k1}
If $r_1=1$, then $k_2=0$.
\end{lemma}

\begin{proof} 
Suppose that $r_1=1$ and $k_2=1$.
Then the algebra~$\mathfrak g_V$ contains vector fields
$Q^1=P(\chi^{11},\chi^{12})+Q^{10}$ with $(\chi^{11},\chi^{12})\ne(0,0)$
and $Q^0=J+P(\chi^{01},\chi^{02})+Q^{00}$. 
The condition $Q^2:=[Q^0,Q^1]=P(\chi^{12},-\chi^{11})+Q^{20}\in\mathfrak g_V$ implies $r_1=2$, contradicting the condition $r_1=1$.
Here $Q^{j0}\in\langle\sigma M\rangle$ for the classes~$\mathscr V^f$ and $\mathscr P_0^\delta$
and $Q^{j0}\in\langle\sigma M,\,\zeta I\rangle$ for the class~$\mathscr P_\lambda^\delta$, $j=0,1,2$.
\end{proof}

Below we separately solve the group classification problems
for each of the classes~$\mathscr V^f$, $\mathscr P_0^\delta$ and~$\mathscr P_\lambda^\delta$, $\lambda\in\mathbb R_{\ne0}$.
For this purpose, we classify appropriate subalgebras of the corresponding algebra~$\mathfrak g_\spanindex$,
taking into account the above constraints on the introduced invariant integers
and on the structure of appropriate subalgebras in general
as well as the classification conditions
\eqref{eq:NSchEPMNVfClassifyingCondition},
\eqref{eq:NSchEPLogMNP0deltaClassifyingCondition} and
\eqref{eq:NSchEPMNPlambdadeltaClassifyingCondition}
for the classes~$\mathscr V^f$, $\mathscr P_0^\delta$ and~$\mathscr P_\lambda^\delta$, respectively.
Simultaneously or successively with this, we integrate the equations on~$V$
implied by the corresponding classification condition for each appropriate subalgebra to be listed.
In the course of solving the classification problems,
we find more constraints for the invariant integers,
whose derivation, in contrast to the presented constraints, 
is intricate for separating it from the process of group classification:
\begin{gather*}
k_1\ne3; \quad
r_1=1\ \ \mbox{if and only if}\ \ k_1\in\{1,2\};\quad
r_1=2\ \ \mbox{if and only if}\ \ k_1=4;\\[.5ex]
\mbox{for the class $\mathscr P_\lambda^\delta$:}\quad
k_3\in\{0,1,2\}\ \ \mbox{if}\ \ \lambda\ne2, \quad
k_3\in\{0,1,3\}\ \ \mbox{if}\ \ \lambda=2.
\end{gather*}

\subsection{General case of modular nonlinearity}\label{sec:NSchEPMNVf(1+2)D}

In view of results of Section~\ref{sec:NSchEPMNVf},
the following conditions hold for any equation~$\mathcal L_V$ from the class $\mathscr V^f$ with $n=2$:
\[
\dim\mathfrak g_V\leqslant 7,\quad
r_1\in \{0,1,2\},\quad
k_0=1,\quad
k_1\in \{0,\dots,4\},\quad
k_2\in \{0,1\},\quad
k_3\in\{0,1\}.
\]
Therefore, any appropriate subalgebra of $\mathfrak g_\spanindex$ is spanned by
\begin{itemize}\itemsep=0.5ex
\item the basis vector field $M$ of the kernel $\mathfrak g^\cap$,
\item $k_1$ vector fields $P(\chi^{p1},\chi^{p2})+\sigma^pM$
with linearly independent tuples~$(\chi^{p1},\chi^{p2})$, $p=1,\dots,k_1$,
\item $k_2$ vector fields of the form $J+P(\chi^{01},\chi^{02})+\sigma^0M$,
\item $k_3$ vector fields of the form $D(1)+\kappa_q J+P(\chi^{q1},\chi^{q2})+\sigma^qM$ with $q=k_1+k_3$.
\end{itemize}

\begin{theorem}\label{thm:NSchEPMNVf(1+2)DGroupClassification}
A complete list of inequivalent Lie symmetry extensions
in the class~$\mathscr V^f$ is exhausted by the cases listed below,
where $U$ is an arbitrary complex-valued smooth function of its arguments or an arbitrary complex constant,
and the other functions and constants take real values.
\end{theorem}

\newcommand{\cc}{\item\label{case:Vf\theenumi}}
\begin{enumerate}
\setcounter{enumi}{-1}\itemsep=0.5ex

\cc
$V=V(t,x)$:\quad $\mathfrak g_V=\mathfrak g^\cap=\langle M\rangle$.

\cc
$V=U(x_1,x_2)$:\quad $\mathfrak g_V=\langle M,\, D(1)\rangle$.

\cc
$V=U(\omega_1,\omega_2)$, \ $\kappa\ne0$:\quad $\mathfrak g_V=\langle M,\, D(1)+\kappa J\rangle$.

\cc
$V=U(t,|x|)+\sigma_t(t)\phi $:\quad $\mathfrak g_V=\langle M,\ J+\sigma(t) M\rangle$.

\cc
$V=U(|x|)+\mu\phi$:\quad $\mathfrak g_V=\langle M,\, J+\mu tM,\, D(1)\rangle$.

\cc
$V=U(t,x_2)+\frac14h^{11}(t)x^2_1$:\quad
$\mathfrak g_V=\langle M,\,P(\theta^1,0),\,P(\theta^2,0)\rangle$,\\
where $\theta^1$ and $\theta^2$ are linearly independent solutions of the equation $\theta_{tt}=h^{11}\theta$.

\cc
$V=U(x_2)+\frac14\alpha x^2_1$:\quad
$\mathfrak g_V=\langle M,\,P(\theta^1,0),\, P(\theta^2,0),\,D(1)\rangle$,\\
where $\theta^1$ and $\theta^2$ are linearly independent solutions of the equation $\theta_{tt}=\alpha\theta$.

\cc
$V=U(t,\varpi)+\dfrac{\chi^1_{tt}}{4\chi^1}x^2_1+\dfrac{\chi^2_{tt}}{4\chi^2}x^2_2$, \
$\chi^1\chi^2_t-\chi^2\chi^1_t\ne 0$, \ $\varpi=\chi^1x_2-\chi^2x_1$:\quad
$\mathfrak g_V=\langle M,\,P(\chi^1,\chi^2)\rangle$.

\cc
$V=U(\omega_2)+\frac14(\alpha^2-\kappa^2)\omega^2_1+\alpha\kappa\omega_1\omega_2$, \ $\alpha\kappa\ne0$:\\
$\mathfrak g_V=\langle M,\, P(e^{\alpha t}\cos\kappa t,e^{\alpha t}\sin\kappa t),\, D(1)+\kappa J\rangle$.

\cc
$V=U(\omega_2)-\frac14\kappa^2\omega^2_1+\mu\omega_1$, \ $\kappa\ne0$, \ $\mu\geqslant0$:\quad
$\mathfrak g_V=\langle M,\,P(\cos\kappa t,\sin\kappa t)+\mu tM,\,D(1)+\kappa J\rangle$.

\cc
$V=\frac14h^{ab}(t)x_ax_b +ih^{00}(t)$, \ $h^{12}=h^{21}$:\quad
$\mathfrak g_V=\langle M,\, P(\chi^{p1}, \chi^{p2}),\, p=1,\dots,4\rangle,$\\
where $\{(\chi^{p1}(t),\chi^{p2}(t))\}$ is a fundamental set of solutions
of the system $\chi^a_{tt}=h^{ab}\chi^b$.

\cc
$V=\frac14\alpha x_1^2+\frac14\beta x^2_2+i\nu$, \ $\beta\ne\alpha\ne0$:\quad
$\mathfrak g_V=\langle M,\,P(\theta^1,0),\,P(\theta^2,0),\,P(0,\theta^3),\,P(0,\theta^4),\,D(1)\rangle$,\\
where $\{\theta^1(t),\theta^2(t)\}$ and $\{\theta^3(t),\theta^4(t)\}$
are fundamental sets of solutions of the equations $\theta_{tt}=\alpha\theta$ and $\theta_{tt}=\beta\theta$, respectively.

\cc
$V=\frac14\alpha\omega_1^2+\frac14\beta\omega^2_2+i\nu$, \ $\beta\ne\alpha\ne0$, \ $\kappa\ne0$:\\
$\mathfrak g_V=\langle M,\,
P(\theta^{p1}\cos\kappa t-\theta^{p2}\sin \kappa t,\, \theta^{p1}\sin\kappa t+\theta^{p2}\cos\kappa t),\,
D(1)+\kappa J,\, p=1,\dots,4 \rangle$,\\
where $(\theta^{p1}(t),\theta^{p2}(t))$ are linearly independent solutions of the system\\
$\theta^1_{tt}-2\kappa\theta^2_t=(\kappa^2+\alpha)\theta^1$,
$\theta^2_{tt}+2\kappa\theta^1_t=(\kappa^2+\beta )\theta^2$.

\cc
$V=\frac 14h^{11}(t)|x|^2+i h^{00}(t)$:\quad
$\mathfrak g_V=\langle M,\, P(\theta^1,0),\, P(\theta^2,0),\, P(0,\theta^1),\, P(0,\theta^2),\, J\rangle,$\\
where $\{\theta^1(t),\theta^2(t)\}$ is a fundamental set of solutions of the equation $\theta_{tt}=h^{11}\theta$.

\cc
$V=\frac14\alpha|x|^2+i\nu$:\quad
$\mathfrak g_V=\langle M,\, P(\theta^1,0),\, P(\theta^2,0),\, P(0,\theta^1),\, P(0,\theta^2),\, J,\, D(1)\rangle$,\\
where $\{\theta^1(t),\theta^2(t)\}$ is a fundamental set of solutions of the equation $\theta_{tt}=\alpha\theta$.

\end{enumerate}

\begin{remark}\label{rem:NSchEPMNVf(1+2)DMaxConditions}
Lie invariance algebras listed in Theorem~\ref{thm:NSchEPMNVf(1+2)DGroupClassification}
are indeed maximal for the corresponding potentials
if these potentials are $G^\sim$-inequivalent to listed potentials with larger Lie invariance algebras.
Here we discuss the conditions of maximality of Lie invariance algebras for a few classification cases.
Thus, in Case~\ref{case:Vf3} the maximality condition is
\[
(\sigma_{tt},\sigma_tU_{t\,|x|},\sigma_t\mathop{\rm Im}U_{tt})\ne(0,0,0)
\quad\mbox{or, \ if \ $\sigma_t=0$,}\quad
(U_{|x|\,|x|}-2|x|U_{|x|},\mathop{\rm Im}U_{|x|\,|x|})\ne(0,0),
\]
which excludes the values of~$V$ that are $G^\sim$-equivalent to those from Cases~\ref{case:Vf4} and~\ref{case:Vf13}.
Analogously, the condition ($\mu\ne0$ or $U_{|x|\,|x|}\ne2|x|U_{|x|}$ or $\mathop{\rm Im}U_{|x|\,|x|}\ne0$)
singles out the potentials of Case~\ref{case:Vf4} that are $G^\sim$-inequivalent to those from Case~\ref{case:Vf14}.
The condition associated with Case~\ref{case:Vf5} includes the inequality
$(U_t,h^{11}_t)\ne(0,0)$ in order to exclude potentials $G^\sim$-equivalent to those from Case~\ref{case:Vf6}.
To avoid the equivalence of potentials of Case~\ref{case:Vf13} to those from Case~\ref{case:Vf14}
we require that $h_t(t)\ne0$ or $h^0_t(t)\ne0$ in~Case~\ref{case:Vf13}.
Analogous inequalities should be satisfied by each tuple of parameter functions
appearing in listed potentials and depending only on the variable~$t$.
Similarly, potentials in Cases~\ref{case:Vf5}--\ref{case:Vf9} are $G^\sim$-inequivalent to ones in Cases~\ref{case:Vf10}--\ref{case:Vf14}
if and only if $U_{\varpi\varpi\varpi}\ne0$ or $\mathop{\rm Im}U_\varpi\ne0$,
where $\varpi:=x_2$ in Cases~\ref{case:Vf5} and~\ref{case:Vf6} and $\varpi:=\omega_2$ in Cases~\ref{case:Vf8} and~\ref{case:Vf9}.
A potential of the form given in Case~\ref{case:Vf10} is $G^\sim$-inequivalent to a potential from Case~\ref{case:Vf13}
if and only if the matrix $(h^{ab})$ is a multiple of the identity matrix, $h^{12}=h^{21}=0$ and $h^{11}=h^{22}$,
i.e., the former potential itself belongs to Case~\ref{case:Vf13}.
%Similar remarks are relevant to Theorems~\ref{thm:NSchEPLogMN(1+2)D} and~\ref{thm:NSchEPPowerMN(1+2)D} below.
\end{remark}

\begin{remark}
In Theorem~\ref{thm:NSchEPMNVf(1+2)DGroupClassification} we neglect
possible gauges of constant parameters in~$V$ by discrete and scaling equivalence transformations.
Thus, alternating the sign of~$x_2$, we can set
$\mu\geqslant0$ in Case~\ref{case:Vf4}, and
$\kappa>0$ in Cases~\ref{case:Vf2}, \ref{case:Vf8}, \ref{case:Vf9}, and~\ref{case:Vf12}.
If $\mathop{\rm Im}f=0$, then we can additionally set $\nu\geqslant0$
in Cases~\ref{case:Vf11}, \ref{case:Vf12} and~\ref{case:Vf14} using the Wigner time reflection.
Deriving the group-classification list for the class~$\mathcal V'$ from that for the class~$\mathcal V^f$ via varying~$f$,
we need to regularly take into account the Wigner time reflection and the scale equivalence transformations,
which correspond to linear (in~$t$) values of the parameter function~$T$.
This leads to gauges
$\mu\in\{0,1\}$ in Case~\ref{case:Vf4},
$\alpha\in\{-1,0,1\}$ in Case~\ref{case:Vf6},
$\kappa=1$ in Cases~\ref{case:Vf2}, \ref{case:Vf8}, \ref{case:Vf9}, and~\ref{case:Vf12},
$\alpha=\pm1$ and $\nu\geqslant0$ in Case~\ref{case:Vf11},
$\alpha\in\{-1,0,1\}$, $\nu\geqslant0$ and, if $\alpha=0$, $\nu\in\{0,1\}$ in Case~\ref{case:Vf14}.
\end{remark}

\begin{proof}
Following the discussion of Section~\ref{sec:NSchEPMNVf},
we single out different classification cases using the invariant integers~$r_1$, $k_2$ and $k_3$.
For the class $\mathscr V^f$, the general form of basis vector fields of $\mathfrak g_V$ from the complement of~$\mathfrak g^\cap$
is
\[
Q^s=D(c_s)+\kappa_s J+P(\chi^{s1},\chi^{s2})+\sigma^sM,
\]
where the index~$s$ runs from~1 to $\dim\mathfrak g_V-1$,
$c_s$ and $\kappa_s$ are real constants,
and $\chi^{s1}$, $\chi^{s2}$ and~$\sigma^s$ are real-valued functions of~$t$,
and the parameter tuples $(c_s,\kappa_s,\chi^{s1},\chi^{s2})$ are linearly independent.
Substituting the tuples $(c_s,\kappa_s,\chi^{s1},\chi^{s2},\sigma^s)$
into the classifying condition~\eqref{eq:NSchEPMNVfClassifyingCondition},
we derive the system of $\dim\mathfrak g_V-1$ equations with respect to~$V$,
\begin{gather}\label{eq:NSchEPMNVf(1+2)DClassifyingCondition}
c_sV_t+(-\kappa_sx_2+\chi^{s1})V_1+(\kappa_sx_1+\chi^{s2})V_2=\frac12\chi^{sa}_{tt}x_a+\sigma^s_t.
\end{gather}

\par\noindent
$\boldsymbol{r_1=k_2=0.}$
The value $k_3=0$ corresponds to the general Case~\ref{case:Vf0} with no Lie-symmetry extension, $\mathfrak g_V=\mathfrak g^\cap$.
If $k_3=1$, then an extension of the algebra $\mathfrak g_V$
is given by the vector field $Q^1=D(1)+\kappa_1 J+P(\chi^{11},\chi^{12})+\sigma^1 M$.
Up to $\pi_*G^\sim$-equivalence we can set $(\chi^{11},\chi^{12})=(0,0)$, $\sigma^1=0$
and reduce $Q^1$ to $D(1)+\kappa_1J$, which results in Cases~\ref{case:Vf1} and~\ref{case:Vf2} depending on
whether or not the parameter~$\kappa$ vanishes.

\medskip\par\noindent
$\boldsymbol{r_1=0,\ k_2=1.}$
The algebra $\mathfrak g_V$ contains, apart from the kernel, a vector field $Q^0$ with $\kappa_0=1$.
Using $\mathcal P_*(\mathcal X)$ with an appropriate~$\mathcal X$,
we can set $\chi^{0a}=0$, which gives $Q^0=J+\sigma^0 M$.
For constant $\sigma^0$, $Q^0$ is reduced to~$J$ by combining with~$M$.
Otherwise, no further reduction is possible.
If $k_3=0$, then this extension is maximal.
Integrating the equation~\eqref{eq:NSchEPMNVf(1+2)DClassifyingCondition} with $s=0$
yields the potential presented in Case~\ref{case:Vf3}.

For $k_3=1$, an additional extension is provided  by a vector field $Q^1$ with $c_1=1$.
Combining $Q^1$ with~$Q^0$ allows us to set $\kappa_1=0$.
From the condition $[Q^0,Q^1]\in\mathfrak g_V$ we derive that $ \chi^{1a}=0$.
Acting by $\mathcal M_*(\Sigma)$ with an appropriate~$\Sigma$ on $\mathfrak g_V$,
we also set $\sigma^1=0$.
Finally, recalling the condition $[Q^0,Q^1]\in\mathfrak g_V$ gives $\sigma^0=\mu t$ with $\mu=\const$.
Solving the system~\eqref{eq:NSchEPMNVf(1+2)DClassifyingCondition}, where $s=0,1$,
leads to Case~\ref{case:Vf4}.

\medskip\par\noindent
$\boldsymbol{r_1=1.}$
The algebra~$\mathfrak g_V$ contains a vector field
$Q^1=P(\chi^{11},\chi^{12})+\sigma^1 M$ with $(\chi^{11},\chi^{12})\ne (0,0)$,
and $k_2=0$ in view of Lemma~\ref{lem:NSchEPMN(1+2)Dr1k1}.
The further consideration splits into two cases depending on
whether or not the tuple $(\chi^{11},\chi^{12})$ is proportional to a constant tuple.

\medskip\par\noindent
{\bf 1.} Suppose that $\chi^{11}\chi^{12}_t-\chi^{11}_t\chi^{12}=0$.
Modulo $\pi_*G^\sim$-equivalence, we can set $\sigma^1=0$, $\chi^{12}=0$
and thus reduce $Q^1$ to $P(\theta^1,0)$.
Integrating the equation~\eqref{eq:NSchEPMNVf(1+2)DClassifyingCondition} with $s=1$,
we construct the potential $V=U(t,x_2)+\frac14h^{11}(t)x_1^2$, where $h^{11}:=\theta^1_{tt}/\theta^1$.
The classifying condition~\eqref{eq:NSchEPMNVfClassifyingCondition} implies
that for this value of~$V$, the equation~$\mathcal L_V$ admits one more
Lie-symmetry vector field of the similar form, $Q^2=P(\theta^2,0)$,
where $\theta^2$ is a solution of the equation $\theta_{tt}=h^{11}\theta$
that is not proportional to~$\theta^1$.
If no further extensions are possible, then we have~Case~\ref{case:Vf5}.

Otherwise, an additional extension is provided by a vector field $Q^3$ with $\tau^3=1$.
The conditions $r_1=1$ and $[Q^1,Q^3]\in\mathfrak g_V$ require $\kappa_3=0$.
Successively acting on~$\mathfrak g_V$ by~$\mathcal P_*(0,\mathcal X^2)$ and~$\mathcal M_*(\Sigma)$
with appropriately chosen values of the parameter functions~$\mathcal X^2$ and~$\Sigma$,
we can set $\chi^{32}=0$ and $\sigma^3=0$.
Further, the condition $[Q^1,Q^3]=P(\theta^1_t,0)+\frac12(\chi^{31}\theta^1_t-\chi^{31}_t\theta^1)M\in\langle Q^1,Q^2,M\rangle$
implies that $\chi^{31}\theta^1_t-\chi^{31}_t\theta^1=\const$, i.e.,
$\chi^{31}\theta^1_{tt}-\chi^{31}_{tt}\theta^1=0$ and thus $\chi^{31}_{tt}=h^{11}\chi^{31}$.
This means that $\chi^{31}\in\langle\theta^1,\theta^2\rangle$.
Therefore this parameter function can be set to zero by linear combining~$Q^3$ with~$Q^1$ and~$Q^2$,
which reduces $Q^3$ to $D(1)$ and gives~Case~\ref{case:Vf6}, where $V_t=0$, i.e., $h^{11}_t=0$ and $U_t=0$.

\medskip\par\noindent
{\bf 2.} Let $\chi^{11}\chi^{12}_t-\chi^{11}_t\chi^{12}\ne 0$.

If $k_3=0$, then the vector field $Q^1$ reduces to the form $Q^1=P(\chi^{11},\chi^{12})$ up to $\pi_*G^\sim$-equivalence.
The equation~\eqref{eq:NSchEPMNVf(1+2)DClassifyingCondition} with $s=1$ integrates
to the expression for~$V$ presented in~Case~\ref{case:Vf7}.

For $k_3=1$, there is an additional symmetry extension provided by a vector field~$Q^2$ with $c_2=1$.
We reduce~$Q^2$ to $D(1)+\kappa_2J$.
Since the commutator $[Q^1,Q^2]$ belongs to $\mathfrak g_V$, we have the equations
\begin{gather}\label{eq:kappa1}
\chi^{11}_t+\kappa_2\chi^{12}=\alpha\chi^{11},\quad
\chi^{12}_t-\kappa_2\chi^{11}=\alpha\chi^{12},
\quad \sigma^1_t=\alpha\sigma^1+\beta,
\end{gather}
where $\alpha$ and $\beta$ are real constants.
In view of~\eqref{eq:kappa1}, for $\kappa_2=0$ the functions~$\chi^{11}$ and~$\chi^{12}$ are necessarily proportional.
Consequently, $\kappa_2=:\kappa\ne0$.

For $\alpha\ne 0$, modulo shifts of~$t$ and combining with~$M$,
we have $Q^1=P(e^{\alpha t}\cos\kappa t,e^{\alpha t}\sin\kappa t)+\beta e^{\alpha t}M$
with a real constant~$\beta$.
Push-forwarding vector fields from $\mathfrak g_V$
by $\mathcal P_*(-\beta\sin\kappa t,\beta\cos\kappa t)$, we set $\beta=0$.
The general solution of the corresponding equations~\eqref{eq:NSchEPMNVf(1+2)DClassifyingCondition} with $s=1,2$
is presented in Case~\ref{case:Vf8}.

Let $\alpha=0$.
Solving the system~\eqref{eq:kappa1}, we obtain the vector field~$Q^1$,
which takes, after linearly combining with $M$, the form $Q^1=P(\cos\kappa t,\sin\kappa t)+\mu tM$.
No further simplification of~$Q^1$ preserving~$Q^2$ is possible.
Thus, we obtain~Case~\ref{case:Vf9}.

\medskip\par\noindent
$\boldsymbol{r_1=2.}$
The algebra $\mathfrak g_V$ contains two vector fields of the form $Q^a=P(\chi^{a1},\chi^{a2})+\sigma^aM$,
where $a=1,2$ and~$\chi^{11}\chi^{22}-\chi^{12}\chi^{21}\ne 0$.
The equations~\eqref{eq:NSchEPMNVf(1+2)DClassifyingCondition} with $s=1,2$ imply that
$V_a=\frac 12h^{ab}(t)x_b+h^{0a}(t)$,
where the coefficients~$h^{ab}$, $h^{0a}$ are real-valued functions of~$t$.
Since $V_{12}=V_{21}$, the matrix $(h^{ab})$ is symmetric
and hence the potential $V$ is a quadratic polynomial in $x_1$ and $x_2$ with the coefficients being functions of $t$,
and only the coefficient of the zeroth-degree summand may have a nonzero imaginary part,
\begin{gather}\label{eq:ClassVfQuadraticPotential}
V=\frac 14h^{ab}(t)x_ax_b+ h^{0b}(t)x_b+\tilde h^{00}(t)+ih^{00}(t).
\end{gather}
The subclass~$\mathscr V^f_{\rm q}$ of the class $\mathscr V^f$ with potentials of the form~\eqref{eq:ClassVfQuadraticPotential} is normalized.
The coefficients $ h^{0b}$ and $\tilde h^{00}$ can be set equal to zero up to $\pi_*G^\sim$-equivalence,
which reduces the potential~$V$ to the form
\begin{gather}\label{eq:ClassVfReducedQuadraticPotential}
 V=\frac 14h^{ab}(t)x_ax_b+ih^{00}(t).
\end{gather}
The equations~\eqref{eq:NSchEPMNVf(1+2)DClassifyingCondition} with this~$V$
are split with respect to different powers of $(x_1,x_2)$, yielding the systems
\begin{gather}\label{eq:ClassVfSystemForH}
c_sh^{11}_t+2\kappa_s h^{12}=0,\quad
c_sh^{12}_t+ \kappa_s (h^{22}-h^{11})=0,\quad
c_sh^{22}_t-2\kappa_s h^{12}=0,\quad
c_sh^{00}_t=0, \\\label{eq:ClassVfChiSigma}
\chi^{sa}_{tt}=h^{ab}\chi^{sb},\quad \sigma^s_t=0.
\end{gather}
In view of the system~\eqref{eq:ClassVfChiSigma},
the algebra $\mathfrak g_V$ in fact contains, in addition to~$M$,
four vector fields $Q^p=P(\chi^{p1},\chi^{p2})$,
where the tuples~$(\chi^{p1},\chi^{p2})$, $p=1,\dots,4$, constitute
a fundamental set of solutions of the system $\chi^a_{tt}=h^{ab}\chi^b$.
Moreover, we can set $\chi^{s1}=\chi^{s2}=\sigma^s=0$ for $s\ne1,\dots,4$ by
linearly combining~$Q^s$ with $Q^1$, \dots, $Q^4$ and~$M$.
Note that due to~\eqref{eq:ClassVfSystemForH},
the condition $k_2=1$ is equivalent to having $h^{11}=h^{22}$ and $h^{12}=0$.

We analyze different cases depending on values of~$k_2$ and~$k_3$.

\medskip\par\noindent
{\bf 1.} $k_2=k_3=0$.
This is the general case within the subclass~$\mathscr V^f_{\rm q}$ with no further Lie symmetry extensions,
which is represented by Case~\ref{case:Vf10}.

\medskip\par\noindent
{\bf 2.} $k_2=0$, $k_3=1$.
The additional Lie symmetry extension is provided by the vector field~$Q^5=D(1)+\kappa_5 J$.

If $\kappa_5=0$, the system~\eqref{eq:ClassVfSystemForH} with $s=5$ is equivalent to
that all $h^{ab}$ and $h^{0a}$ are constants.
Up to rotations, we can reduce the matrix $(h^{ab})$ to a diagonal matrix ${\rm diag}(\alpha,\beta)$
with $\beta\ne\alpha\ne0$, obtaining Case~\ref{case:Vf11}.

Otherwise, up to translations of time, the general solution of the system~\eqref{eq:ClassVfSystemForH} with $s=5$ is
$h^{11}=\alpha\cos^2t+\beta\sin^2t$, $h^{12}=h^{21}=(\alpha-\beta)\cos t\sin t$,
$h^{22}=\alpha\sin^2t+\beta\cos^2t$, $h^{00}=\nu$,
where $\alpha$, $\beta$ and $\nu$ are arbitrary real constants with $\alpha\ne\beta$
due to the auxiliary conditions for the subclass~$\mathscr V^f_0$.
We integrate the system~\eqref{eq:ClassVfChiSigma} with the above values of~$h^{ab}$,
rearranging the potential $V$ in terms of $\omega_a$.
This provides Case~\ref{case:Vf12}.

\medskip\par\noindent
{\bf 3.} $k_2=1$, $k_3=0$.
The additional Lie-symmetry extension is provided by the vector field~$Q^0=J$,
which corresponds to Case~\ref{case:Vf13}.

\medskip\par\noindent
{\bf 4.} $k_2=1$, $k_3=1$. We obviously obtain Case~\ref{case:Vf14},
where there is one more Lie-symmetry vector field $Q^5=D(1)$ in comparison to Case~\ref{case:Vf13}.
\end{proof}

\subsection{Logarithmic modular nonlinearity}\label{sec:NSchEPLogMN(1+2)D}

We recall that the class of Schr\"odinger equations with a logarithmic modular nonlinearity with a fixed $\delta$
is denoted by~$\mathscr P_0^\delta$ and consists of equations of the form~\eqref{eq:NSchEPLogMN}.
For any equation~$\mathcal L_V$ from the class~$\mathscr P_0^\delta$ with $n=2$,
the results from Section~\ref{sec:NSchEPLogMN} imply that
\begin{gather*}
\dim\mathfrak g_V\leqslant 8,\quad
r_1\in\{0,1,2\}, \quad
k_0=2,\quad
k_1\in \{0,\dots,4\},\quad
k_2\in\{0,1\}, \quad
k_3\in\{0,1\}.
\end{gather*}
It then follows that any appropriate subalgebra of $\mathfrak g_\spanindex$ is spanned by
\begin{itemize}\itemsep=0ex
\item the basis vector fields $M$ and $I$ of the kernel $\mathfrak g^\cap$,
\item $k_1$ vector fields $P(\chi^{p1},\chi^{p2})+\sigma^pM+\zeta^pI$ with linearly
independent tuples $(\chi^{p1},\chi^{p2})$, $p=1,\dots,k_1$,
\item $k_2$ vector fields $J+P(\chi^{01},\chi^{02})+\sigma^0M+\zeta^0I$,
\item $k_3$ vector fields $D(1)+\kappa_q J+P(\chi^{q1},\chi^{q2})+\sigma^qM+\zeta^qI$ with $q=k_1+k_3$.
\end{itemize}

We will also use the following notation:
\begin{gather*}
I'=e^{-\delta_2t}(\delta_2I-\delta_1M)\quad\text{if}\quad\delta_2\ne0\quad \text{and} \quad
I'=I+\delta_1tM\quad \text{if} \quad\delta_2=0,\\
P'(\chi^1,\chi^2)=P(\chi^1,\chi^2)-\hat\zeta I-\delta_1\int\hat\zeta\,{\rm d}t M\quad \text{with}\quad
\hat\zeta=e^{-\delta_2t}\int e^{\delta_2t}h^{0b}\chi^b\,{\rm d}t,
\end{gather*}
where all involved parameters will be explained in the corresponding places,
$\delta_1:=\mathop{\rm Re}\delta$ and $\delta_2:=\mathop{\rm Im}\delta$,
and we assume that
$h^{0a}$ is the imaginary part of the coefficient of the summand of the first degree in~$x_a$
if this summand is explicitly presented in the corresponding potential~$V$,
otherwise $h^{0a}:=0$.
For example, $h^{01}=\nu^1\cos\kappa t-\nu_2\sin\kappa t$ and $h^{02}=\nu^1\sin\kappa t+\nu_2\cos\kappa t$
in Case~\ref{case:Log12} below.

\begin{theorem}\label{thm:NSchEPLogMN(1+2)D}
A complete list of inequivalent Lie symmetry extensions
in the class~$\mathscr P_0^\delta$ is exhausted by the cases listed below,
where $U$ is an arbitrary complex-valued smooth function of its arguments or an arbitrary complex constant,
and the other functions and constants take real values.
\end{theorem}

\renewcommand{\cc}{\item\label{case:Log\theenumi}}
\begin{enumerate}
\setcounter{enumi}{-1}\itemsep=0.5ex

\cc
$V=V(t,x)$:\quad $\mathfrak g_V=\mathfrak g^\cap=\langle M,\, I'\rangle$.

\cc
$V=U(x_1,x_2)$:\quad $\mathfrak g_V=\langle M,\, I',\, D(1)\rangle$.

\cc
$V=U(\omega_1,\omega_2)$:\quad $\mathfrak g_V=\langle M,\, I',\, D(1)+\kappa J,\, \kappa \ne 0\rangle$.

\cc
$V=U(t,|x|)+(\sigma_t-i\zeta_t-\delta\zeta)\phi\,$:\quad $\mathfrak g_V=\langle M,\, I',\, J+\sigma(t) M+\zeta(t)I\rangle$.

\cc
$V=U(|x|)+(\mu-\delta'\nu)\phi$,\\
$\delta_2\ne0$:\quad $\delta':=\delta_1+i\delta_2$,\quad $\mathfrak g_V=\langle M,\, I',\, J+\mu tM+\nu I,\, D(1)\rangle$,\\
$\delta_2=0$:\quad $\delta':=i$,\quad $\mathfrak g_V=\langle M,\, I',\, J+(\mu t+\frac12\nu\delta_1t^2)M+\nu tI,\, D(1)\rangle$.

\cc
$V=U(t,x_2)+\frac14 h^{11}(t)x^2_1+ih^{01}(t)x_1$:\quad
$\mathfrak g_V=\langle M,\, I',\,P'(\theta^1,0),\, P'(\theta^2,0)\rangle$,\\
where $\theta^1$ and $\theta^2$ are linearly independent solutions of the equation $\theta_{tt}=h^{11}\theta$.

\cc
$V=U(x_2)+\frac14\alpha x^2_1+i\nu x_1$:\quad
$\mathfrak g_V=\langle M,\, I',\, P'(\theta^1,0),\, P'(\theta^2,0),\, D(1)\rangle$,\\
where $\theta^1$ and $\theta^2$ are linearly independent solutions of the equation $\theta_{tt}=\alpha\theta$.

\cc
$V=U(t,\varpi)+\dfrac{\chi^{11}_{tt}}{4\chi^{11}}x_1^2+\dfrac{\chi^{12}_{tt}}{4\chi^{12}}x_2^2+ih^{01}(t)x_1, \ \varpi=\chi^{11}x_2-\chi^{12}x_1$, \ $\chi^{11}_t\chi^{12}\ne\chi^{11}\chi^{12}_t$:\quad
$\mathfrak g_V=\langle M,\, I',\, P'(\chi^{11},\chi^{12})\rangle$.
%where $h^{11}$, $h^{12}$ and $h^{01}$ are real-valued smooth functions of~$t$,

\cc
$V=U(\omega_2)+\frac14(\alpha^2-\kappa^2)\omega^2_1+\alpha\kappa\omega_1\omega_2+i\nu\omega_1$, \ $\alpha\kappa\ne0$:\\
$\mathfrak g_V=\langle M,\, I',\, P'(e^{\alpha t}\cos\kappa t,e^{\alpha t}\sin\kappa t),\, D(1)+\kappa J\rangle$.

\cc
$V=U(\omega_2)-\frac14\kappa^2\omega^2_1+(\mu+i\nu)\omega_1$, \ $\kappa\ne0$:\quad
$\mathfrak g_V=\langle M,\, I',\, P'(\cos\kappa t,\sin\kappa t)+\mu tM,\, D(1)+\kappa J\rangle$.

\cc
$V=\frac14h^{ab}(t)x_ax_b +ih^{0b}(t)x_b$, \ $h^{12}=h^{21}$:\quad
$\mathfrak g_V=\langle M,\,I',\,P'(\chi^{p1},\chi^{p2}),\,p=1,\dots,4\rangle$,\\
where $\{(\chi^{p1}(t),\chi^{p2}(t))\}$ is a fundamental set of solutions of the system $\chi^a_{tt}=h^{ab}\chi^b$.

\cc
$V=\frac14\alpha x_1^2+\frac14\beta x^2_2+i\nu_ax_a$, \ $\beta\ne\alpha\ne0$:\\
$\mathfrak g_V=\langle M,\, I',\,P'(\theta^1,0),\,P'(\theta^2,0),\,P'(0,\theta^3),\,P'(0,\theta^4),\,D(1)\rangle$,\\
where $\{\theta^1(t),\theta^2(t)\}$ and $\{\theta^3(t),\theta^4(t)\}$
are fundamental sets of solutions of the equations $\theta_{tt}=\alpha\theta$ and $\theta_{tt}=\beta\theta$, respectively.

\cc
$V=\frac14\alpha\omega_1^2+\frac14\beta\omega_2^2+i\nu_a\omega_a$, \ $\beta\ne\alpha\ne0$, \ $\kappa\ne0$:\\
$\mathfrak g_V=\langle M,\, I',\, P'(\theta^{p1}\cos\kappa t-\theta^{p2}\sin\kappa t,\theta^{p1}\sin\kappa t+\theta^{p2}\cos\kappa t),\,
D(1)+\kappa J,\, p=1,\dots,4\rangle$,\\
where $(\theta^{p1}(t),\theta^{p2}(t))$ are linearly independent solutions of the system \\
$\theta^1_{tt}-2\kappa\theta^2_t=(\kappa^2+\alpha)\theta^1$,
$\theta^2_{tt}+2\kappa\theta^1_t=(\kappa^2+\beta )\theta^2$.

\cc
$V=\frac 14h(t)|x|^2$:\quad
$\mathfrak g_V=\langle M,\, I',\, P(\chi^{p1}, \chi^{p2}),\,p=1,\dots,4,\, J\rangle,$\\
where $\{(\chi^{p1}(t),\chi^{p2}(t))\}$ is a fundamental set of solutions of the system $\chi_{tt}=h\chi$.

\cc
$V=\frac14\alpha|x|^2$:\quad
$\mathfrak g_V=\langle M,\, I',\, P(\chi^1,0),\, P(\chi^2,0),\, P(0,\chi^1),\, P(0,\chi^2),\, J, \, D(1)\rangle$,\\
where $(\chi^1(t),\chi^2(t))$ is a fundamental set of solutions of the equation $\chi_{tt}=\alpha\chi$.

\end{enumerate}

\begin{remark}\label{rem:NSchEPLogMN(1+2)DMaxConditions}
The conditions of maximality of Lie invariance algebras presented in Theorem~\ref{thm:NSchEPLogMN(1+2)D}
are similar to those discussed in Remark~\ref {rem:NSchEPMNVf(1+2)DMaxConditions}
for algebras listed in Theorem~\ref{thm:NSchEPMNVf(1+2)DGroupClassification}.
The required modifications are obvious.
For example, in Case~\ref{case:Log3} the maximality condition is
$(\sigma_{tt}-i\zeta_{tt}-\delta\zeta_t,(\sigma_t-i\zeta_t-\delta\zeta)U_{t\,|x|})\ne(0,0)$
or, if $\sigma_t-i\zeta_t-\delta\zeta=0$, it is $(U_{|x|\,|x|}-2|x|U_{|x|},\mathop{\rm Im}U_{|x|\,|x|})\ne(0,0)$,
which excludes the values of~$V$ that is $G^\sim$-equivalent to those from Cases~\ref{case:Log4} and~\ref{case:Log13}.
A potential of the form presented in Case~\ref{case:Log10} is $G^\sim$-equivalent to a potential from Case~\ref{case:Log13}
if and only if the matrix $(h^{ab})$ is a multiple of the identity matrix, $h^{12}=h^{21}=0$ and $h^{11}=h^{22}$, and $h^{0a}=0$,
i.e., the former potential itself belongs to Case~\ref{case:Log13}.
\end{remark}

\begin{remark}\label{rem:NSchEPLogMN(1+2)DScalings}
In Theorem~\ref{thm:NSchEPLogMN(1+2)D}, we also neglect
some possible gauges of constant parameters in~$V$.
Thus, alternating the signs of~$x_2$ and/or~$x_1$, we can set
$\mu\geqslant0$ and, if $\mu=0$, we can make $\nu\geqslant0$ in Cases~\ref{case:Log4} and~\ref{case:Log9} and
$\kappa>0$ in Cases~\ref{case:Log2}, \ref{case:Log8}, \ref{case:Log9} and~\ref{case:Log12}.
If $\mathop{\rm Im}\delta=0$, then using the Wigner time reflection
we can additionally set $\nu\geqslant0$ in Cases~\ref{case:Log4}, \ref{case:Log6}, \ref{case:Log8} and~\ref{case:Log9}
and make one of nonzero $\nu_a$ positive in Cases~\ref{case:Log11}, \ref{case:Log12} and~\ref{case:Log14}.
In contrast to the class~$\mathcal V'$, in the course of group classification of the entire class~$\mathscr P_0$,
we have no additional possibilities for gauging constant parameters in~$V$.
\end{remark}

\begin{proof}
Similar to the proof of Theorem~\ref{thm:NSchEPMNVf(1+2)DGroupClassification},
we follow the discussion of Section~\ref{sec:NSchEPLogMN}
and single out different classification cases using the invariant integers~$r_1$, $k_2$ and $k_3$.
At the same time, this proof essentially differs in some points
from the proof of Theorem~\ref{thm:NSchEPMNVf(1+2)DGroupClassification}
due to the appearance of the vector fields $\zeta I$ in $\mathfrak g_\spanindex$
and the extension of the kernel Lie invariance algebra with~$I'$.
For the class~$\mathscr P_0^\delta$, the general form of basis vector fields of $\mathfrak g_V$
from the complement of~$\mathfrak g^\cap$ is
\[
Q^s=D(c_s)+\kappa_s J+P(\chi^{s1},\chi^{s2})+\sigma^sM+\zeta^s I,
\]
where the range of~$s$ is equal to $\dim\mathfrak g_V-2$,
$c_s$ and $\kappa_s$ are real constants,
$\chi^{s1}$, $\chi^{s2}$, $\sigma^s$ and~$\zeta^s$ are real-valued functions of~$t$,
and the parameter tuples $(c_s,\kappa_s,\chi^{s1},\chi^{s2})$ are linearly independent.
Substituting the tuples $(c_s,\kappa_s,\chi^{s1},\chi^{s2},\sigma^s,\zeta^s)$
into the classifying condition~\eqref{eq:NSchEPLogMNP0deltaClassifyingCondition},
we derive the system of $\dim\mathfrak g_V-2$ equations with respect to~$V$,
\begin{gather}\label{eq:NSchEPLogMN(1+2)DClassifyingCondition}
c_sV_t+(-\kappa_sx_2+\chi^{s1})V_1+(\kappa_sx_1+\chi^{s2})V_2
=\frac12\chi^{sa}_{tt}x_a+\sigma^s_t-\delta_1\zeta^s-i(\zeta^s_t+\delta_2\zeta^s).
\end{gather}
We operate with the parameters~$c_s$, $\kappa_s$, $\chi^{s1}$, $\chi^{s2}$ and $\sigma^s$
within the class~$\mathscr P_0^\delta$
in the same way as we did within the class~$\mathscr V^f$.
This is why below we assume all the reductions for these parameters having been carried out
and present only the points where the presence of $\zeta I$ in $\mathfrak g_\spanindex$
crucially modifies the proof.
Each case of this theorem corresponds to the case of Theorem~\ref{thm:NSchEPMNVf(1+2)DGroupClassification}
with the same number.
The derivation of Cases~\ref{case:Log0}--\ref{case:Log3} is essentially the same
as in the proof of Theorem~\ref{thm:NSchEPMNVf(1+2)DGroupClassification}.

\medskip\par\noindent
$\boldsymbol{r_1=0,\ k_2=1,\ k_3=1.}$
Up to $\pi_*G^\sim$-equivalence, the Lie-symmetry extension can be assumed to be provided
by $Q^0=J+\sigma M+\zeta I$ and $Q^1=D(1)$.
The condition $[Q^1,Q^0]\in\mathfrak g_V$ reduces to
$[Q^1,Q^0]=\sigma_tM+\zeta_tI\in\langle M,I'\rangle$.
This gives expressions for the derivatives~$\sigma_t$ and~$\zeta_t$,
which depend on whether or not $\delta_2$ vanishes.
Integrating these expressions up to $\pi_*G^\sim$-equivalence,
we find $\sigma=\mu t$, $\zeta=\nu$ if $\delta_2\ne0$
and $\sigma=\mu t+\frac12\nu\delta_1t^2$, $\zeta=\nu t$ if $\delta_2=0$.
The equations~\eqref{eq:NSchEPLogMN(1+2)DClassifyingCondition} with $s=0,1$
leads to the expression for~$V$ in Case~\ref{case:Log4}.

\medskip\par\noindent
$\boldsymbol{r_1=1.}$
Then $k_2=0$, and, the algebra $\mathfrak g_V$ contains, in addition to~$M$,
at least a vector field~$Q^1=P(\chi^{11},\chi^{12})+\sigma^1M+\zeta^1I$ with $(\chi^{11},\chi^{12})\ne(0,0)$.
The further analysis depends on whether or not the parameter functions~$\chi^{11}$ and~$\chi^{12}$ are linearly dependent.

\medskip\par\noindent
{\bf 1.} $\chi^{11}_t\chi^{12}=\chi^{11}\chi^{12}_t$.
Up to $G^\sim$-equivalence, we can set $\chi^{12}=0$.
It is then convenient to re-denote $\theta^1:=\chi^{11}$.
Then integrating the equation~\eqref{eq:NSchEPLogMN(1+2)DClassifyingCondition} with $s=1$ with respect to~$V$ gives
\begin{gather}\label{multinonlicomponentswithproportionality}
V=U(t,x_2)+\frac14 h^{11}(t)x^2_1+ih^{01}(t)x_1+\tilde h^{01}(t)x_1,
\end{gather}
where $h^{11}$, $h^{01}$ and $\tilde h^{01}$ are smooth real-valued functions of $t$.

If $k_3=0$, then we make the coefficient $\tilde h^{01}$ zero using equivalence transformations.
For this potential, the classifying condition~\eqref{eq:NSchEPLogMNP0deltaClassifyingCondition}
implies that the functions involved in~$Q^1$ satisfy the system
$\theta^1_{tt}=h^{11}\theta^1$,
$\zeta^1_t+\delta_2\zeta^1=-h^{01}\chi^{11}$
and $\sigma^1_t=\delta_1\zeta^1$, i.e., $Q^1=P'(\theta^1,0)$.
Moreover, the algebra~$\mathfrak g_V$ contains one more vector field of the similar form,
$Q^2=P'(\theta^2,0)$, where $\theta^2$ is a solution of the equation $\theta_{tt}=h^{11}\theta$
that is linearly independent with~$\theta^1$.
As a result, we have~Case~\ref{case:Log5}.

For $k_3=1$, an additional Lie-symmetry extension is provided by a vector field $Q^3$ with $\tau^3=1$.
The conditions $r_1=1$ and $[Q^1,Q^3]\in\mathfrak g_V$ require $\kappa_3=0$.
Successively acting on~$\mathfrak g_V$ by~$\mathcal P_*(\mathcal X)$, $\mathcal M_*(\Sigma)$ and~$\mathcal I_*(Z)$
with appropriately chosen values of the parameter functions~$\mathcal X^a$, $\Sigma$ and~$Z$,
we can set $\chi^{3a}=0$, $\sigma^3=0$ and $\zeta^3=0$.
Then the equation~\eqref{eq:NSchEPLogMN(1+2)DClassifyingCondition} with $s=3$ takes the form $V_t=0$
and implies that the coefficients $h^{11}$, $h^{01}$ and $\tilde h^{01}$
in the representation~\eqref{multinonlicomponentswithproportionality} for~$V$ are constants.
Acting by~$\mathcal P(\tilde h^{01},0)$ on~$\mathcal L_V$, we annihilate~$\tilde h^{01}$.
Repeating the argumentation from the case $k_3=0$ leads to Case~\ref{case:Log6}.

\medskip\par\noindent
{\bf 2.}  $\chi^{11}_t\chi^{12}\ne\chi^{11}\chi^{12}_t$.
The equation~\eqref{eq:NSchEPLogMN(1+2)DClassifyingCondition} with $s=1$
integrates with respect to~$V$ to
\begin{gather}\label{eq:NSchEPLogMN(1+2)Dr11k31ExpressionForV}
V=U(t,\varpi)+\frac{\chi^{11}_{tt}}{4\chi^{11}}x_1^2+\frac{\chi^{12}_{tt}}{4\chi^{12}}x_2^2+ih^{01}(t)x_1+\tilde h^{01}(t)x_1,
\end{gather}
where $\varpi:=\chi^{11}x_2-\chi^{12}x_1$, and
$h^{01}$ and $\tilde h^{01}$ are smooth real-valued functions of $t$.

If $k_3=0$, then we make the coefficient $\tilde h^{01}$ zero using equivalence transformations.
For this potential, the classifying condition~\eqref{eq:NSchEPLogMNP0deltaClassifyingCondition}
implies that the functions involved in~$Q^1$ satisfy the system
$\zeta^1_t+\delta_2\zeta^1=-h^{01}\chi^{11}$ and $\sigma^1_t=\delta_1\zeta^1$,
i.e., $Q^1=P'(\chi^{11},\chi^{12})$, which corresponds to Case~\ref{case:Log7}.

If $k_3=1$, then modulo $\pi_*G^\sim$-equivalence we reduce a vector field $Q^2$ with $\tau^2=1$,
which provides an additional Lie-symmetry extension, to the form $Q^2=D(1)+\kappa_2 J$.
In view of the condition $[Q^2,Q^1]\in\mathfrak g_V$ and modulo shifts of~$t$ and combining with~$M$,
we have $(\chi^{11},\chi^{12})=(e^{\alpha t}\cos\kappa_2t,e^{\alpha t}\sin\kappa_2t)$.
The equations~\eqref{eq:NSchEPLogMN(1+2)DClassifyingCondition} with $s=1,2$
simultaneously integrate with respect to~$V$ to
\[
V=U(\omega_2)+\frac14(\alpha^2-\kappa^2_2)\omega^2_1+\alpha\kappa_2\omega_1\omega_2+(\mu+i\nu)\omega_1.
\]
Here $\alpha$, $\mu$ and~$\nu$ are real constants such that
$\zeta^1_t+\delta_2\zeta^1=-\nu$ and $\sigma^1_t=\delta_1\zeta^1+\mu$.
If $\alpha\ne 0$, then acting by $\mathcal P(-\mu\cos t,\mu\sin t)$ on $\mathcal L_V$, we set $\mu=0$.
For $\alpha=0$, annihilating~$\mu$ if it is nonzero is impossible.
After integrating the above system for~$(\sigma^1,\zeta^1)$ and linearly combining~$Q^1$ with~$M$ and~$I'$,
we obtain Cases~\ref{case:Log8} and~\ref{case:Log9} for $\alpha\ne 0$ and $\alpha=0$, respectively.

\medskip\par\noindent
$\boldsymbol{r_1=2.}$
Similarly to the case $r_1=2$ in the proof of Theorem~\ref{thm:NSchEPMNVf(1+2)DGroupClassification},
here the potential~$V$ is quadratic in~$x$,
\[
V=\frac 14h^{ab}(t)x_ax_b+\tilde h^{0b}(t)x_b+ih^{0b}(t)x_b+h^{00}(t)+i\tilde h^{00}(t),
\]
where the functions $\tilde h^{0b}$, $h^{00}$, $\tilde h^{00}$ can be set to zero up to~$G^\sim$-equivalence.
So, it suffices to study potentials of the form 
\begin{gather}\label{eq:ClassP0deltaReducedQuadraticPotential}
V=\frac 14h^{ab}(t)x_ax_b+ih^{0b}(t)x_b.
\end{gather}
Splitting the equations~\eqref{eq:NSchEPLogMN(1+2)DClassifyingCondition}
in view of~\eqref{eq:ClassP0deltaReducedQuadraticPotential},
with respect to different powers of $(x_1,x_2)$ leads to the systems
\begin{gather}\label{eq:ClassP0deltaSystemForHA}
c_sh^{11}_t+2\kappa_s h^{12}=0,\quad
c_sh^{12}_t+ \kappa_s (h^{22}-h^{11})=0,\quad
c_sh^{22}_t-2\kappa_s h^{12}=0,
\\\label{eq:ClassP0deltaSystemForHB}
c_sh^{01}_t+\kappa_s h^{02}=0,\quad
c_sh^{02}_t-\kappa_s h^{01}=0,
\\\label{eq:ClassP0deltaChiSigma}
\chi^{sa}_{tt}=h^{ab}\chi^{sb},\quad
\sigma^s_t=\delta_1\zeta^s,\quad
\zeta^s_t+\delta_2\zeta^s=-h^{0b}\chi^{sb}.
\end{gather}
In view of the system~\eqref{eq:ClassP0deltaChiSigma},
the algebra $\mathfrak g_V$ necessarily contains, in additional to~$M$ and~$I'$,
four vector fields $Q^p=P'(\chi^{p1},\chi^{p2})$,
where the tuples~$(\chi^{p1},\chi^{p2})$, $p=1,\dots,4$, constitute
a fundamental set of solutions of the system $\chi^a_{tt}=h^{ab}\chi^b$.
Moreover, we can set $\chi^{s1}=\chi^{s2}=\sigma^s=\zeta^s=0$, $s\ne1,\dots,4$, by
linearly combining~$Q^s$ with $Q^1$, \dots, $Q^4$, $M$ and~$I'$.
Further analysis, which depends on values of~$k_2$ and~$k_3$,
is analogous to the proof of Theorem~\ref{thm:NSchEPMNVf(1+2)DGroupClassification}
and results in Cases~\ref{case:Log10}--\ref{case:Log14}.
Note only that due to~\eqref{eq:ClassP0deltaSystemForHA} and~\eqref{eq:ClassP0deltaSystemForHB}, 
here the condition $k_2=1$ is equivalent to having $h^{11}=h^{22}$ and $h^{12}=h^{0b}=0$.
\end{proof}

\subsection{Power modular nonlinearity}\label{sec:NSchEPPowerMN(1+2)D}

We solve the group classification problem of the class~$\mathscr P_\lambda^\delta$ of nonlinear Schr\"odinger equations with potentials
and power nonlinearity with a fixed $\delta$, $\delta\in\mathbb C_0$ for $n=2$.
For any potential $V$ in an equation of this class we have
\begin{gather*}
\dim\mathfrak g_V\leqslant 9,\quad
r_1\in\{0,1,2\},\quad
k_0=1,\quad
k_1\in \{0,\dots,4\},\quad
k_2\in\{0,1\},\quad
k_3\in\{0,1,2,3\},
\end{gather*}
and $\lambda'=\frac {1}{\lambda}-\frac 12$.
The algebra $\mathfrak g_V$ is spanned by the following vector fields:
\begin{itemize}
\item the basis vector field $M$ of the kernel $\mathfrak g^\cap$,
\item $k_1$ vector fields $P(\chi^{p1},\chi^{p2})+\sigma^pM$ with linearly independent tuples~$(\chi^{p1},\chi^{p2})$
and~$p=1,\dots,k_1$,
\item $k_2$ vector fields of the form $J+P(\chi^{01},\chi^{02})+\sigma^0M$,
\item $k_3$ vector fields $D^\lambda(\tau^q)+\kappa_q J+P(\chi^{q1},\chi^{q2})+\sigma^qM$,
$q=k_1+1,\dots,k_1+k_3$ with linearly independent $\tau^{k_1+1}$, \dots, $\tau^{k_1+k_3}$.
\end{itemize}

\begin{theorem}\label{thm:NSchEPPowerMN(1+2)D}
A complete list of inequivalent Lie symmetry extensions
in the class~$\mathscr P_\lambda^\delta$ is exhausted by the cases listed below,
where $U$ is an arbitrary complex-valued smooth function of its arguments or an arbitrary complex constant,
and the other functions and constants take real values.
\end{theorem}

\renewcommand{\cc}{\item\label{case:Power\theenumi}}
\begin{enumerate}
\setcounter{enumi}{-1}\itemsep=0.5ex

\cc
$V=V(t,x)$:\quad $\mathfrak g_V=\mathfrak g^\cap=\langle M\rangle$.

\cc
$V=U(x_1,x_2)$:\quad $\mathfrak g_V=\langle M,\, D(1)\rangle$.

\cc
$V=U(\omega_1,\omega_2)$, \ $\kappa\ne0$:\quad $\mathfrak g_V=\langle M,\, D(1)+\kappa J\rangle$.

\cc
$V=|x|^{-2}U(\varpi)$,
$\varpi=\phi-2\kappa\ln|x|$, \ $\kappa\ne0$, \ $U_\varpi\ne0$:\quad
$\mathfrak g_V=\langle M ,\, D(1),\, D^\lambda(t)+\kappa J\rangle$.

\cc
$V=|x|^{-2}U(\phi)$, \ $U_\phi\ne0$,\\
$\lambda\ne2$:\quad $\mathfrak g_V=\langle M,\, D(1),\, D^\lambda(t)\rangle$,\\
$\lambda=2  $:\quad $\mathfrak g_V=\langle M,\, D(1),\, D^\lambda(t),\, D^\lambda(t^2)\rangle$.

\cc
$V=U(t,|x|)+\mu\phi$:\quad $\mathfrak g_V=\langle M,\, J+\mu t M\rangle$.

\cc
$V=U(|x|)+\mu\phi$:\quad $\mathfrak g_V=\langle M,\, J+\mu tM,\, D(1)\rangle$.

\cc
$V=|x|^{-2}U$, \ $U\ne0$,\\
$\lambda\ne2$:\quad $\mathfrak g_V=\langle M,\, J,\, D(1),\, D^\lambda(t)\rangle$,\\
$\lambda=2  $:\quad $\mathfrak g_V=\langle M,\, J,\, D(1),\, D^\lambda(t),\, D^\lambda(t^2)\rangle$.

\cc
$V=U(t,x_2)$:\quad
$\mathfrak g_V=\langle M,\, P(1,0),\, P(t,0)\rangle$.

\cc
$V=U(\varpi)$, \ $\varpi=x_2$:\quad $\mathfrak g_V=\langle M,\, P(1,0),\, P(t,0),\, D(1)\rangle$.

\cc
$V=t^{-1}U(\varpi)$, \ $\varpi=|t|^{-1/2}x_2$:\quad $\mathfrak g_V=\langle M,\, P(1,0),\, P(t,0),\, D^\lambda(t)\rangle$.

\cc
$V=(t^2+1)^{-1}\big(U(\varpi)+2i\lambda't\big)$, \ $\varpi=(t^2+1)^{-1/2}x_2$:\\
$\mathfrak g_V=\langle M,\,  P(1,0),\, P(t,0),\,  D^\lambda(t^2+1)\rangle$.

\cc
$V=x_2^{-2}U$, \ $U\ne0$,\\
$\lambda\ne2 $:\quad $\mathfrak g_V=\langle M,\, P(1,0),\, P(t,0),\, D(1),\, D^\lambda(t)\rangle$,\\
$\lambda=2   $:\quad $\mathfrak g_V=\langle M,\, P(1,0),\, P(t,0),\, D(1),\, D^\lambda(t),\, D^\lambda(t^2)\rangle$.

\cc
$V=U(t,\omega_2)+\frac14(h_{tt}-h)h^{-1}\omega^2_1+h_th^{-1}\omega_1$, \ $h=h(t)\ne0$:\quad
$\mathfrak g_V=\langle M,\, P(h\cos t,h\sin t)\rangle$.

\cc

$V=U(\omega_2)+\frac14(\alpha^2-\kappa^2)\omega^2_1+\alpha\kappa\omega_1\omega_2$, \ $\alpha\kappa\ne0$:\\
$\mathfrak g_V=\langle M,\, P(e^{\alpha t}\cos\kappa t,e^{\alpha t}\sin\kappa t),\, D(1)+\kappa J\rangle\rangle $.

\cc
$V=U(\omega_2)-\frac14\kappa^2\omega^2_1+\mu\omega_1$, \ $\kappa\ne0$:\quad
$\mathfrak g_V=\langle M,\, P(\cos\kappa t,\,\sin\kappa t)+\mu t M,\, D(1)+\kappa J\rangle$.

\cc
$V=\frac14h^{ab}(t)x_ax_b +ih^{00}(t)$, \ $h^{12}=h^{21}$:\quad
$\mathfrak g_V=\langle M,\, P(\chi^{p1}, \chi^{p2}),\, p=1,\dots,4\rangle$,\\
where $\{(\chi^{p1}(t),\, \chi^{p2}(t))\}$ is a fundamental set of solutions of the system $\chi^a_{tt}=h^{ab}\chi^b$.

\cc
$V=\frac14\alpha x_1^2+\frac14\beta x^2_2+i\nu$, \ $\beta\ne\alpha\ne0$:\quad
$\mathfrak g_V=\langle M,\,P(\theta^1,0),\,P(\theta^2,0),\,P(0,\theta^3),\,P(0,\theta^4),\,D(1)\rangle$,\\
where $\{\theta^1(t),\theta^2(t)\}$ and $\{\theta^3(t),\theta^4(t)\}$
are fundamental sets of solutions of the equations $\theta_{tt}=\alpha\theta$ and $\theta_{tt}=\beta\theta$, respectively.

\cc
$V=\frac14\alpha\omega_1^2+\frac14\beta\omega^2_2+i\nu$, \ $\beta\ne\alpha\ne0$, \ $\kappa\ne0$:\\
$\mathfrak g_V=\langle M,\,
P(\theta^{p1}\cos\kappa t-\theta^{p2}\sin \kappa t,\, \theta^{p1}\sin\kappa t+\theta^{p2}\cos\kappa t),\,
D(1)+\kappa J,\, p=1,\dots,4 \rangle$,\\
where $(\theta^{p1}(t),\theta^{p2}(t))$ are linearly independent solutions of the system\\
$\theta^1_{tt}-2\kappa\theta^2_t=(\kappa^2+\alpha)\theta^1$,
$\theta^2_{tt}+2\kappa\theta^1_t=(\kappa^2+\beta )\theta^2$.

\cc
$V=ih^{00}(t)$:\quad
$\mathfrak g_V=\langle M,\, P(1,0),\, P(t,0),\, P(0,1),\, P(0,t),\, J)\rangle$.

\cc
$V=i\nu$, \ $\nu\ne0$:\quad
$\mathfrak g_V=\langle M,\, P(1,0),\, P(t,0),\, P(0,1),\, P(0,t),\, J ,\,D(1)\rangle$.

\cc
$V=i\nu t^{-1}$, \ $\nu\ne0$:\quad
$\mathfrak g_V=\langle M,\, P(1,0),\, P(t,0),\, P(0,1),\, P(0,t),\,J,\, D^\lambda(t)\rangle$.

\cc
$V=i(t^2+1)^{-1}\left(2\lambda't+\nu\right)$, \ $(\lambda',\nu)\ne(0,0)$:\\
$\mathfrak g_V=\langle M,\, P(1,0),\, P(t,0),\, P(0,1),\, P(0,t),\, J,\, D^\lambda(t^2+1)\rangle$.

\cc
$V=0\colon$\\
$\lambda\ne2$:\quad $\mathfrak g_V=\langle M,\, P(1,0),\, P(t,0),\, P(0,1),\, P(0,t),\, J,\,
D(1),\, D^\lambda(t)\rangle$,\\
$\lambda=2$:\quad $\mathfrak g_V=\langle M,\, P(1,0),\, P(t,0),\, P(0,1),\, P(0,t),\, J,\,
D(1),\, D^\lambda(t),\, D^\lambda(t^2)\rangle.$

\end{enumerate}

\begin{remark}\label{rem:NSchEPPowerMN(1+2)DMaxConditions}
As in Theorems~\ref{thm:NSchEPMNVf(1+2)DGroupClassification} and~\ref{thm:NSchEPLogMN(1+2)D},
Lie invariance algebras listed in Theorem~\ref{thm:NSchEPPowerMN(1+2)D}
are indeed maximal  and coincide with~$\mathfrak g_V$ for the corresponding potentials
if these potentials are $G^\sim$-inequivalent to listed potentials with larger Lie invariance algebras.
For Cases~\ref{case:Power3}, \ref{case:Power4}, \ref{case:Power12} \mbox{and~\ref{case:Power20}--\ref{case:Power22}},
the maximality conditions are obvious and presented directly in the corresponding cases.
A~potential of the form given in Case~\ref{case:Power16} is $G^\sim$-equivalent to a potential from Case~\ref{case:Power19}
if and only if $h^{12}=h^{21}=0$ and $h^{11}=h^{22}$, i.e., the former potential itself belongs to Case~\ref{case:Power19}.
At the same time, the maximality conditions for a number of cases of Lie-symmetry extensions
within the class~$\mathscr P_\lambda^\delta$ are more cumbersome
than similar conditions within the classes~$\mathscr V^f$ and~$\mathscr P_0^\delta$.
This is caused by the fact that the class~$\mathscr P_\lambda^\delta$ admits equivalence transformations with nonconstant~$T_t$,
which complicates describing $G^\sim$-equivalent cases.
Moreover, there are a number of cases of Lie symmetry extensions within the class~$\mathscr P_\lambda^\delta$
that have no counterparts among Lie symmetry extensions within the classes~$\mathscr V^f$ and~$\mathscr P_0^\delta$
since they are associated with the values of~$k_3$ in $\{2,3\}$.
For example, the maximality condition for Case~\ref{case:Power5},
\[
\mu(U_{t\,|x|},\mathop{\rm Im}U_{tt})\ne(0,0)
\quad\mbox{or, \ if \ $\mu=0$,}\quad
(U_{|x|\,|x|}-2|x|U_{|x|},\mathop{\rm Im}U_{|x|\,|x|})\ne(0,0),
\]
is analogous to that from Remark~\ref{rem:NSchEPMNVf(1+2)DMaxConditions}
for Case~\ref{case:Vf3} of Theorem~\ref{thm:NSchEPMNVf(1+2)DGroupClassification},
and similarly excludes the values of~$V$ being $G^\sim$-equivalent to those from Cases~\ref{case:Power6} and~\ref{case:Power16}.
Nevertheless, the maximality condition for Case~\ref{case:Power6} is modified, in comparison with
the condition from Remark~\ref{rem:NSchEPMNVf(1+2)DMaxConditions}
for Case~\ref{case:Vf4} of Theorem~\ref{thm:NSchEPMNVf(1+2)DGroupClassification},~to
\[
\mu\ne0
\quad\mbox{or, \ if \ $\mu=0$,}\quad
(U_{|x|\,|x|}-2|x|U_{|x|},\mathop{\rm Im}U_{|x|\,|x|})
\ne(0,0)\quad\mbox{and}\quad
(|x|^2U)_{|x|}\ne0
\]
to single out the potentials from Case~\ref{case:Power4} 
that are $G^\sim$-inequivalent to those in both Cases~\ref{case:Power7} and~\ref{case:Power16}.
Inequalities to be satisfied by tuples of parameter functions
appearing in listed potentials and depending only on the variable~$t$
also become more complicated.
Thus, the maximality condition in Case~\ref{case:Power19} is that
$\big((\alpha t^2+\beta t+\gamma)h^{00}\big)_t\ne2\alpha\lambda'$
for any nonzero constant tuple $(\alpha,\beta,\gamma)$,
which excludes further Lie symmetry extensions $G^\sim$-equivalent to Cases~\ref{case:Power20}--\ref{case:Power22}.
For Case~\ref{case:Power8}, the maximality condition is
$(U_{222},\mathop{\rm Im}U_2)\ne(0,0)$ and
$\big((\alpha t^2+\beta t+\gamma)U\big)_t\ne2\alpha\lambda'$ for any nonzero constant tuple $(\alpha,\beta,\gamma)$.
In Cases~\ref{case:Power9}--\ref{case:Power11} the maximality condition is $(U_{\varpi\varpi\varpi},\mathop{\rm Im}U_\varpi)\ne(0,0)$ and $(\varpi^2U)_\varpi\ne0$,
and this excludes values of~$V$ that are $G^\sim$-equivalent to those from Cases~\ref{case:Power12} and~\ref{case:Power16}--\ref{case:Power23}.
Similarly, potentials from Cases~\ref{case:Power13}--\ref{case:Power15} are $G^\sim$-inequivalent to those from Cases~\ref{case:Power16}--\ref{case:Power23}
if and only if $(U_{\omega_2\omega_2\omega_2},\mathop{\rm Im}U_{\omega_2})\ne(0,0)$.
This condition is necessary and sufficient for the maximality of Lie symmetry extensions given in Cases~\ref{case:Power14} and~\ref{case:Power15},
but for Case~\ref{case:Power13} we need a further condition to guarantee the inequivalence with potentials from Cases~\ref{case:Power14} and~\ref{case:Power15}.
\end{remark}

\begin{remark}\label{rem:NSchEPPowerMN(1+2)DScalings}
To make the presentation in Theorem~\ref{thm:NSchEPPowerMN(1+2)D} consistent with
Theorems~\ref{thm:NSchEPMNVf(1+2)DGroupClassification} and~\ref{thm:NSchEPLogMN(1+2)D},
we avoid scalings of constant parameters and alternating their signs
as well as gauging some parameter functions.
In contrast to the classes~$\mathcal V^f$ and~$\mathscr P_0^\delta$,
the class~$\mathscr P_\lambda^\delta$ admits scaling equivalence transformations
and even more general transformations with nonconstant~$T_t$.
As a result, we can set
$\mu\in\{0,1\}$ in Cases~\ref{case:Power5} and~\ref{case:Power6},
$\mu\geqslant0$ in Case~\ref{case:Power15},
$\kappa=1$ in Cases~\ref{case:Power2}, \ref{case:Power3}, \ref{case:Power14}, \ref{case:Power15}, \ref{case:Power17} and~\ref{case:Power18},
$h^{11}=0$ (or $h^{22}=0$) in Case~\ref{case:Power16},
$\mathop{\rm Im}U\geqslant0$ if $\mathop{\rm Im}\delta=0$ in Cases~\ref{case:Power7} and~\ref{case:Power12},
$\nu=\pm1$ if $\mathop{\rm Im}\delta\ne0$ and $\nu=1$ if $\mathop{\rm Im}\delta=0$ in Case~\ref{case:Power20},
$\nu\geqslant\lambda'$ in Case~\ref{case:Power21} using the equivalence transformation $\mathcal D(-t^{-1})\mathcal I(2\lambda^{-1}\ln|t|)$,
and $\nu\geqslant0$ if $\mathop{\rm Im}\delta=0$ in Case~\ref{case:Power22}.
\end{remark}

\begin{proof}
As in the proofs of Theorems~\ref{thm:NSchEPMNVf(1+2)DGroupClassification} and~\ref{thm:NSchEPLogMN(1+2)D},
we follow the discussion from the corresponding Section~\ref{sec:NSchEPPowerMN}
and single out different classification cases using the invariant integers~$r_1$, $k_2$ and~$k_3$.
For the class~$\mathscr P_\lambda^\delta$, the general form of basis vector fields of $\mathfrak g_V$ from the complement of~$\mathfrak g^\cap$
is
\[
Q^s=D^\lambda(\tau^s)+\kappa_s J+P(\chi^{s1},\chi^{s2})+\sigma^sM,
\]
where the range of~$s$ is equal to $\dim\mathfrak g_V-1$,
$\kappa_s$ are real constants,
$\tau^s$, $\chi^{s1}$, $\chi^{s2}$ and~$\sigma^s$ are real-valued functions of~$t$,
and the parameter tuples $(\tau^s,\kappa_s,\chi^{s1},\chi^{s2})$ are linearly independent.
Substituting $(\tau^s,\kappa_s,\chi^{s1},\chi^{s2},\sigma^s)$
into the classifying condition~\eqref{eq:NSchEPMNPlambdadeltaClassifyingCondition},
we derive the system of $\dim\mathfrak g_V-1$ equations with respect to~$V$,
\begin{gather}\label{eq:NSchEPPowerMN(1+2)DClassifyingCondition}
\begin{split}
\tau^sV_t+{}&\frac12\tau^s_tx_a V_a+\kappa_s(x_1V_2-x_2V_1)+\chi^{sa}V_a+\tau^s_tV\\
&=\frac18\,\tau^s_{ttt}x_ax_a+\frac12\chi^{sa}_{tt}x_a+\sigma^s_t+i{\lambda'}\tau^s_{tt}.
\end{split}
\end{gather}

\medskip\par\noindent
$\boldsymbol{r_1=k_2=0.}$
The classification of Lie-symmetry extensions depends only on~$k_3$,
and Lemma~\ref{lem:NSchEPMNPlambdadeltaPigV} reduces it
to the classification of finite-dimensional algebras of vector fields on the real line.
The algebra $\mathfrak g_V$ contains $k_3$ vector fields~$Q^s$, $s=1,\dots,k_3$,
where $\tau^s$ are linearly independent.

The condition $k_3=0$ means that there is no Lie-symmetry extension,
$\mathfrak g_V=\mathfrak g^\cap=\langle M\rangle$, and we have Case~\ref{case:Power0}.

If $k_3=1$, then up to $\pi_*G^\sim$-equivalence we can set
$\tau^1=1$, $\chi^{11}=\chi^{12}=\sigma^1=0$ and, if $\delta_2=0$, then $\kappa_3\geqslant0$.
Depending on whether or not $\kappa$ is zero, this splits into Cases~\ref{case:Power2} and~\ref{case:Power3}.

For $k_3\geqslant2$, modulo $\pi_*G^\sim$-equivalence and changing basis of~$\mathfrak g_V$
we can assume that $(\tau^1,\tau^2)=(1,t)$ and further set $\chi^{1a}=\sigma^1=0$.
In view of the equations $\kappa_1=0$, $\chi^{2a}_t=\sigma^2_t=0$ following from
the condition $[Q^1,Q^2]\in\mathfrak g_V$,
we can annihilate $\sigma^2$ and $(\chi^{21},\chi^{22})$
by combining $Q^2$ with~$M$ and by acting with $\mathcal P_*(2\chi^{21},2\chi^{22})$
on~$\mathfrak g_V$, respectively.
Solving the equations~\eqref{eq:NSchEPPowerMN(1+2)DClassifyingCondition} with $s=1,2$,
we obtain the expression for~$V$ from Case~\ref{case:Power4}.
Then the classifying condition~\eqref{eq:NSchEPMNPlambdadeltaClassifyingCondition} with such~$V$
implies that there is no further Lie-symmetry extension for $\lambda\ne2$,
and for $\lambda=2$ the algebra $\mathfrak g_V$ contains one more vector field,
$Q^3=D^\lambda(t^2)$.

\medskip\par\noindent
$\boldsymbol{r_1=0,\ k_2=1.}$
The algebra $\mathfrak g_V$ necessarily contains the vector field $Q^0$ with $\tau^1=0$ and $\kappa_1=1$,
where $\chi^{0a}$ can be set to zero up to $\pi_*G^\sim$-equivalence.
The further consideration depends on the value of~$k_3$.

\medskip\par\noindent
$k_3=0$.
Modulo $\pi_*G^\sim$-equivalence, we can assume that $\sigma^0=\mu t$ with $\mu \in \{0,1\}$,
which leads to Case~\ref{case:Power5}.

\medskip\par\noindent
$k_3=1$.
The algebra $\mathfrak g_V$ additionally contains the vector field $Q^1$ with $\tau^1\ne 0$.
Up to $\pi_*G^\sim$-equivalence we can set $\tau^1=1$, $\sigma^1=0$.
The condition $[Q^0,Q^1]\in\mathfrak g_V$ implies $\chi^{1a}=0$, $\sigma^0_t=\mu=\const$
and thus, linearly combining~$Q^0$ with~$M$, we set $\sigma^0=\mu t$ and obtain~Case~\ref{case:Power6}.

\medskip\par\noindent
$k_3\geqslant 2$.
Modulo $\pi_*G^\sim$-equivalence and changing basis of~$\mathfrak g_V$,
we can again assume that $(\tau^1,\tau^2)=(1,t)$.
We successively set $\kappa_1=\kappa_2=0$ by linearly combining~$Q^1$ and $Q^2$ with $Q^0$
and $\sigma^1=0$ by acting with appropriate $\mathcal M_*(\Sigma)$ on~$\mathfrak g_V$.
The condition $[Q^0,Q^s]\in\mathfrak g_V$, $s=1,2$, gives $\chi^{sa}=0$, $\sigma^0_t=0$.
Therefore, we can set $\sigma^0=0$ by linearly combining~$Q^0$ with~$M$.
The simultaneous integration of the equations~\eqref{eq:NSchEPPowerMN(1+2)DClassifyingCondition} with $s=0,1,2$
provides the expression for~$V$ from Case~\ref{case:Power7}.
In view of the classifying condition~\eqref{eq:NSchEPMNPlambdadeltaClassifyingCondition} with such~$V$,
there is no further Lie-symmetry extension for $\lambda\ne2$,
and the algebra $\mathfrak g_V$ additionally contains the vector field $Q^3=D^\lambda(t^2)$ if $\lambda=2$.

\medskip\par\noindent
$\boldsymbol{r_1=1.}$
Then $k_2=0$, and the algebra $\mathfrak g_V$ contains, in addition to~$M$,
at least a vector field~$Q^1=P(\chi^{11},\chi^{12})+\sigma^1M$.
We again have two possibilities for the tuple $(\chi^{11},\chi^{12})$.

\medskip\par\noindent
{\bf 1.} $\chi^{11}_t\chi^{12}=\chi^{11}\chi^{12}_t$.
That is, the tuple $(\chi^{11},\chi^{12})$ is proportional to a constant tuple.
Up to $\pi_*G^\sim$-equivalence we can set $Q^1=P(1,0)$.
Then the equation~\eqref{eq:NSchEPPowerMN(1+2)DClassifyingCondition} with $s=1$ is $V_1=0$,
in view of which the classifying condition~\eqref{eq:NSchEPMNPlambdadeltaClassifyingCondition} implies
that the vector field $P(t,0)$ also belongs to $\mathfrak g_V$.
Thus we have $V=U(t,x_2)$, where $(U_{222},\mathop{\rm Im}U_2)\ne(0,0)$, since otherwise $r_1=2$, cf.\ the case $r_1=2$ below.
Consider the subclass~$\mathscr U$ of the class~$\mathscr P_\lambda^\delta$ of equations with potentials of this form,
which is singled out from~$\mathscr P_\lambda^\delta$ by the constraints
$V_1=0$ and $(V_{222},\mathop{\rm Im}V_2)\ne(0,0)$.
This subclass is normalized
and its equivalence group is singled out from the equivalence group~$G^\sim$ of~$\mathscr P_\lambda^\delta$
by the conditions that
$T$ is fractional linear in~$t$, $O\in\{\mathop{\rm diag}(\varepsilon_1,\varepsilon_2)\mid \varepsilon_1,\varepsilon_2=\pm1\}$,
and $\mathcal X^1$ is affine in~$T$.
We successively split the equations~\eqref{eq:NSchEPPowerMN(1+2)DClassifyingCondition} with $V=U(t,x_2)$
with respect to~$x_1$ and~$U_2$, obtaining $\tau^s_{ttt}=0$, $\kappa_s=0$ and $\chi^{s1}_{tt}=0$.
Therefore, linearly combining $Q^s$, $s\ne1,2$ with $Q^1$ and~$Q^2$,
we can annihilate $\chi^{s1}$ for $s\ne1,2$.
The classification of Lie-symmetry extensions within the subclass~$\mathscr U$ is reduced to
the classification of subalgebras of the algebra~$\langle \p_t,t\p_t,t^2\p_t\rangle\simeq{\rm sl}(2,\mathbb R)$.
A list of inequivalent subalgebras of this algebra is exhausted by
\[
\{0\},\quad \langle \p_t\rangle,\quad \langle t\p_t\rangle,\quad \langle (t^2+1)\p_t\rangle,\quad \langle \p_t,t\p_t\rangle,\quad \langle \p_t,t\p_t,t^2\p_t\rangle.
\]
%$\{0\}$, $\langle \p_t\rangle$, $\langle t\p_t\rangle$, $\langle (t^2+1)\p_t\rangle$, $\langle \p_t,t\p_t\rangle$ and $\langle \p_t,t\p_t,t^2\p_t\rangle$.
The zero subalgebra corresponds to Case~\ref{case:Power8}, where $k_3=0$, which is the general case for the subclass~$\mathscr U$.
For $k_3=1$, we should consider the listed one-dimensional subalgebras.
In other words, here the additional Lie symmetry extension  is provided by a vector field~$Q^3$
with $\tau^3\in\{1,t,t^2+1\}$.
Acting by $\mathcal P_*(0,\mathcal X^2)$ and $\mathcal M_*(\Sigma)$
with appropriate values of the parameter functions~$\mathcal X^2$ and~$\Sigma$,
we can make~$\chi^{32}$ and~$\sigma^3$ vanishing,
deriving Cases~\ref{case:Power9}--\ref{case:Power11}.
If $k_3\geqslant2$, then modulo $\pi_*G^\sim$-equivalence the Lie-symmetry extension
contains at least two vector fields~$Q^3$ and~$Q^4$ with $\tau^3=1$ and~$\tau^4=t$.
We again reduce $Q^3$ to the form $Q^3=D(1)$.
The condition $[Q^3,Q^4]\in\mathfrak g_V$ implies that $\chi^{42}_t=\sigma^4_t=0$.
This is why we can annihilate $\sigma^4$ and~$\chi^{42}$
by combining $Q^4$ with~$M$ and by acting with $\mathcal P_*(0,2\chi^{42})$
on~$\mathfrak g_V$, respectively.
Then the equations~\eqref{eq:NSchEPPowerMN(1+2)DClassifyingCondition} with $s=3,4$
jointly integrate to $V=Ux_2^{-2}$, and it is obvious
that $k_3$ is equal to either 2 or 3 if $\lambda\ne2$ or $\lambda=2$, respectively.
This gives Case~\ref{case:Power12}.

\medskip\par\noindent
{\bf 2.} $\chi^{11}_t\chi^{12}\ne \chi^{11}\chi^{12}_t$,
i.e., the tuple $(\chi^{11},\chi^{12})$ is not proportional to a constant tuple.

If $k_3=0$, then we can, up to $\pi_*G^\sim$-equivalence, make $Q^1=P(h\cos t, h\sin t)$, where
$h(t)$ is a nonvanishing smooth function of $t$, which gives Case~\ref{case:Power13}.

The condition $k_3=1$ means that the algebra~$\mathfrak g_V$ contains exactly one (up to linear combining)
vector field~$Q^2$ with $\tau^2\ne 0$.
Up to $\pi_*G^\sim$-equivalence, the parameter functions $\chi^{1a}$, $\sigma^1$ can be set to zero,
and $Q^2$ reduces to $Q^2=D(1)+\kappa_2J$.
Further consideration, which is similar to that
in the respective case of the proof of Theorem~\ref{thm:NSchEPMNVf(1+2)DGroupClassification},
leads to Cases~\ref{case:Power14} and~\ref{case:Power15}.

The case $k_2\geqslant2$ is not possible.
Indeed, otherwise the algebra~$\mathfrak g_V$ would contain, modulo $\pi_*G^\sim$-equivalence,
vector fields~$Q^2$ with $\tau^2=1$ and~$Q^3$ with $\tau^3=t$.
The condition $[Q^2,Q^3]\in\mathfrak g_V$ gives $\kappa_2=0$.
Then the condition $[Q^2,Q^1]\in\mathfrak g_V$ implies that
$(\chi^{11}_t,\chi^{12}_t)\in\langle(\chi^{11},\chi^{12})\rangle$,
which contradicts case's assumption.

\medskip\par\noindent
$\boldsymbol{r_1=2.}$
Following the corresponding case of the proof of Theorem~\ref{thm:NSchEPMNVf(1+2)DGroupClassification},
up to~$G^\sim$-equivalence we reduce the potential~$V$ to the form~\eqref{eq:ClassVfReducedQuadraticPotential}
and substitute it into the equations~\eqref{eq:NSchEPPowerMN(1+2)DClassifyingCondition}.
Then we split the obtained equations with respect to different powers of $(x_1,x_2)$
and derive the systems
\begin{gather}\label{eq:ClassPlambdadeltaSystemForH}
\begin{split}
&\tau^sh^{11}_t+2\tau^s_t h^{11}+2\kappa_s h^{12}=\frac{\tau_{ttt}}2,\quad
\tau^sh^{22}_t+2\tau^s_t h^{22}-2\kappa_s h^{12}=\frac{\tau_{ttt}}2,\\
&\tau^sh^{12}_t+2\tau^s_t h^{12}+ \kappa_s (h^{22}-h^{11})=0,\quad
\tau^sh^{00}_t+\tau^s_t h^{00}=\tau_{tt},
\end{split}
\\\label{eq:ClassPlambdadeltaChiSigma}
\chi^{sa}_{tt}=h^{ab}\chi^{sb},\quad \sigma^s_t=0.
\end{gather}
In view of the system~\eqref{eq:ClassPlambdadeltaChiSigma},
the algebra $\mathfrak g_V$ in fact contains, in addition to~$M$,
four vector fields $Q^p=P(\chi^{p1},\chi^{p2})$,
where tuples~$(\chi^{p1},\chi^{p2})$, $p=1,\dots,4$, constitute
a fundamental set of solutions of the system $\chi^a_{tt}=h^{ab}\chi^b$.
Moreover, we can set $\chi^{s1}=\chi^{s2}=\sigma^s=0$ for $s\ne1,\dots,4$ by
linearly combining~$Q^s$ with $Q^1$, \dots, $Q^4$ and~$M$.
Note that due to~\eqref{eq:ClassPlambdadeltaSystemForH},
the condition $k_2=1$ is equivalent to having $h^{11}=h^{22}$ and $h^{12}=0$.

The consideration of the cases $(r_1,k_2,k_3)=(2,0,0)$ and $(r_1,k_2,k_3)=(2,0,1)$
is similar to that in the proof of Theorem~\ref{thm:NSchEPMNVf(1+2)DGroupClassification}
and leads to Cases~\ref{case:Power16}, \ref{case:Power17} and~\ref{case:Power18}.
The case $(r_1,k_2,k_3)=(2,0,2)$ is impossible.
Indeed, otherwise the algebra~$\mathfrak g_V$ could contain, modulo $\pi_*G^\sim$-equivalence,
vector fields~$Q^5$ and~$Q^6$ with $\tau^5=1$ and $\tau^6=t$.
The condition $[Q^5,Q^6]\in\mathfrak g_V$ implies that $\kappa_1=0$,
and thus the system of equations~\eqref{eq:ClassPlambdadeltaSystemForH} with $s=5,6$
admits, as a system with respect to~$h$'s, only the zero solution,
which contradicts the condition $k_2=0$.

In view of Corollary~\ref{cor:NSchEPPowerMNReductionPotToXIndep},
the condition $k_2=1$ implies that, modulo $G^\sim$-equivalence,
the potential $V$ does not depend on~$x$ and is purely imaginary, i.e.,
$V=ih^{00}(t)$, where $h^{00}$ is a smooth real-valued function of $t$.
Substituting~$V$ into the classifying condition~\eqref{eq:NSchEPMNPlambdadeltaClassifyingCondition}
and splitting with respect to the powers of $x$, we obtain the equations
$\tau_{ttt}=0$, $\chi^a_{tt}=0$, $\sigma_t=0$,
$(\tau h^{00})_t=\lambda'\tau_{tt}$.
They imply that the vector fields $J$, $P(1,0)$, $P(t,0)$, $P(0,1)$, $P(0,t)$
belong to the maximal Lie invariance algebra of any equation with potentials of the above form.
Additional $G^\sim$-inequivalent Lie-symmetry extensions are related to inequivalent nonzero subalgebras
of the algebra $\langle D(1),D^\lambda(t),D^\lambda(t^2)\rangle$,
which are exhausted by
$\langle D(1)\rangle$,
$\langle D^\lambda(t)\rangle$,
$\langle D^\lambda(t^2+1)\rangle$,
$\langle D(1),D^\lambda(t)\rangle$ and the entire algebra itself.
This leads to Cases~\ref{case:Power19}--\ref{case:Power23}.
\end{proof}

\begin{remark}\label{rem:NSchEPPowerMN(1+2)DStationaryPots}
Cases~\ref{case:Power10}, \ref{case:Power11}, \ref{case:Power21} and~\ref{case:Power22}
of Theorem~\ref{thm:NSchEPPowerMN(1+2)D} can be replaced by
the corresponding $G^\sim$-equivalent cases with $t$-independent potentials,
\begin{enumerate}\itemsep=0.5ex

\item[\ref{case:Power10}$'$.]
$V=\frac14x_1^2+\tilde U(x_2)$:\quad $\mathfrak g_V=\langle M,\, P(e^{-t},0),\, P(e^t,0),\, D(1)\rangle$,

\item[\ref{case:Power11}$'$.]
$V=-\frac14x_1^2+\tilde U(x_2)$:\quad
$\mathfrak g_V=\langle M,\,  P(\cos t,0),\, P(\sin t,0),\, D(1)\rangle$,

\item[\ref{case:Power21}$'$.]
$V=\frac14|x|^2+i\tilde\nu$:\quad
$\mathfrak g_V=\langle M,\, P(e^{-t},0),\, P(e^t,0),\, P(0,e^{-t}),\, P(0,e^t),\,J,\, D(1)\rangle$,

\item[\ref{case:Power22}$'$.]
$V=-\frac14|x|^2+i\nu$:\quad
$\mathfrak g_V=\langle M,\, P(\cos t,0),\, P(\sin t,0),\, P(0,\cos t),\, P(0,\sin t),\, J,\, D(1)\rangle$.

\end{enumerate}
Here $\tilde\nu:=2(\nu-\lambda')\ne-2\lambda'$ in Case~\ref{case:Power21}$'$ and
$(\lambda',\nu)\ne(0,0)$ in Case~\ref{case:Power22}$'$.
The mappings of Cases~\ref{case:Power10}$'$, \ref{case:Power21}$'$  to Cases~\ref{case:Power10}, \ref{case:Power21}
and of Cases~\ref{case:Power11}$'$, \ref{case:Power22}$'$  to Cases~\ref{case:Power11}, \ref{case:Power22}
are realized by the equivalence transformations
$\mathcal D(e^{2t})\mathcal I(-2\lambda^{-1}t)$ and $\mathcal D(\tan t)\mathcal I(\lambda^{-1}\ln\cos^2t)$,
respectively.
\end{remark}

\subsection{Classification list for the entire class}%2\label{sec:NSchEPMN(1+2)DclassificationResult}

Summing up the results of this section, we obtain the following assertion.

\begin{theorem}\label{thm:NSchEPMN(1+2)D}
A complete list of $\mathcal G^\sim_{\mathscr V}$-inequivalent Lie symmetry extensions in the class $\mathscr V$ with $n=2$
is the union of the group classification lists for the subclasses~$\mathscr V'$, $\mathscr P_0$ and $\mathscr P_\lambda$, $\lambda\in\mathbb R_{\ne0}$,
up to $G^\sim_{\mathscr V'}$-, $G^\sim_{\mathscr P_0}$- and $G^\sim_{\mathscr P_\lambda}$-equivalences,
which are presented in
Theorems~\ref{thm:NSchEPMNVf(1+2)DGroupClassification}, \ref{thm:NSchEPLogMN(1+2)D} and~\ref{thm:NSchEPPowerMN(1+2)D}, respectively.
\end{theorem}

It is easy to select
$\mathcal G^\sim_{\mathscr V}$-inequivalent Lie symmetry extensions with stationary potentials.
These are cases involving the vector field~$D(1)$,
which are exhausted by
Cases~\ref{case:Vf1}, \ref{case:Vf4}, \ref{case:Vf6}, \ref{case:Vf11} and~\ref{case:Vf14} of Theorem~\ref{thm:NSchEPMNVf(1+2)DGroupClassification},
Cases~\ref{case:Log1}, \ref{case:Log4}, \ref{case:Log6}, \ref{case:Log11} and~\ref{case:Log14} of Theorem~\ref{thm:NSchEPLogMN(1+2)D},
Cases~\ref{case:Power1}, \ref{case:Power3}, \ref{case:Power4}, \ref{case:Power6}, \ref{case:Power7}, \ref{case:Power9},
\ref{case:Power12}, \ref{case:Power17}, \ref{case:Power20} and~\ref{case:Power23} of Theorem~\ref{thm:NSchEPPowerMN(1+2)D}
and the modified Cases~\ref{case:Power10}$'$, \ref{case:Power11}$'$, \ref{case:Power21}$'$ and~\ref{case:Power22}$'$ given in Remark~\ref{rem:NSchEPPowerMN(1+2)DStationaryPots}.

\subsection*{Acknowledgements}

The authors are pleased to thank Peter Basarab-Horwath,
Vyacheslav Boyko, Michael Kunzinger, Anatoly Nikitin, Galyna Popovych, Dmytro Popovych,
Alexander Sakhnovich and Olena Vaneeva for stimulating discussions.
The research of C.K. was supported by International Science Programme (ISP) in collaboration
with Eastern Africa Universities Mathematics Programme (EAUMP).
The research of R.O.P. was supported by the Austrian Science Fund (FWF), projects P25064, P29177 and P30233.
The authors are grateful to  the Abdus Salam International Centre for Theoretical Physics (ICTP)
for hospitality and excellent scientific environment.

\end{document}